\documentclass[reqno,11pt,nosumlimits]{article}
\usepackage{amssymb,amsfonts,amscd}
\usepackage[mathscr]{eucal}
\usepackage{amsmath}
\usepackage{graphics}
\usepackage{color}

\usepackage{ulem}
\usepackage{graphicx}

\allowdisplaybreaks[1]
\numberwithin{equation}{section}

\newtheorem{theorem}{Theorem}[section]

\newtheorem{proposition}[theorem]{Proposition} 
\newtheorem{corollary}[theorem]{Corollary} 
\newtheorem{lemma}{Lemma}

\newtheorem{definition}[theorem]{Definition}

\newtheorem{remark}[theorem]{Remark}
\newtheorem{example}[theorem]{Example}
\newtheorem{choice}{Choice}
\newtheorem{observation}[theorem]{Observation}
\newtheorem{convention}{Convention}

\newenvironment{proof}{{\it Proof. }}{{ \vskip 0.1cm
 \hfill{$\square$}}  \vspace{0.0cm} \medskip}

 \DeclareMathOperator{\Map}{Map}

\DeclareMathOperator{\start}{\bullet}

\DeclareMathOperator{\cov}{cov}

\DeclareMathOperator{\bE}{{\mathbb E}}

\DeclareMathOperator{\Image}{Image}
\DeclareMathOperator{\arc}{arc}
 \DeclareMathOperator{\Tr}{Tr}
\DeclareMathOperator{\GL}{GL}
\DeclareMathOperator{\gl}{gl}

 \DeclareMathOperator{\Mat}{Mat}
\DeclareMathOperator{\supp}{supp}
 \DeclareMathOperator{\Det}{Det}
\DeclareMathOperator{\dom}{dom}

 \DeclareMathOperator{\End}{End}
\DeclareMathOperator{\Hom}{Hom}
\DeclareMathOperator{\Ad}{Ad}
\DeclareMathOperator{\ad}{ad}

 \DeclareMathOperator{\WLO}{WLO}

\DeclareMathOperator{\wind}{wind}

\DeclareMathOperator{\sgn}{sgn}

\DeclareMathOperator{\gleam}{gleam}

\DeclareMathOperator{\aff}{aff}
\DeclareMathOperator{\Aff}{Aff}

\DeclareMathOperator{\Hol}{Hol}

\DeclareMathOperator{\mean}{mean}

\renewcommand{\l}{\lambda}
\renewcommand{\r}{\rho}
\renewcommand{\th}{\theta}
\renewcommand{\a}{\alpha}

\newcommand{\be}{\begin{equation}}
\newcommand{\ee}{\end{equation}}
\newcommand{\vf}{\varphi}

\newcommand{\bN}{{\mathbb N}}
\newcommand{\bR}{{\mathbb R}}
\newcommand{\bC}{{\mathbb C}}
\newcommand{\bZ}{{\mathbb Z}}
\newcommand{\bS}{{\mathbb S}}

\newcommand{\ct}{{\mathfrak t}}
\newcommand{\cG}{{\mathfrak g}}

\newcommand{\ck}{{\mathfrak k}}
\newcommand{\cg}{c_{\cG}}

\newcommand{\cA}{{\mathcal A}}
\newcommand{\cB}{{\mathcal B}}
\newcommand{\cC}{{\mathcal C}}

\newcommand{\cK}{{\mathcal K}}

\newcommand{\cP}{{\mathcal P}}

\newcommand{\cR}{{\mathcal R}}

\newcommand{\cW}{{\mathcal W}}

\newcommand{\CW}{{\mathcal C}}
\newcommand{\G}{{\mathcal G}}

\newcommand{\iprime}{j}

\newcommand{\face}{\mathfrak F}

\newcommand{\orth}{\perp}

\newcommand{\eps}{\epsilon}

\newcommand{\id}{{ \rm id}}

\begin{document}

\title{From simplicial Chern-Simons theory to the shadow invariant II}


\maketitle

\begin{center} \large
Atle Hahn
\end{center}

\begin{center}
\it   \large  Grupo de  F{\'i}sica Matem{\'a}tica da Universidade de Lisboa \\
Av. Prof. Gama Pinto, 2\\
PT-1649-003 Lisboa, Portugal\\
Email: atle.hahn@gmx.de
  \end{center}

\begin{abstract} This is the second of a series of  papers
in which we introduce and study a rigorous ``simplicial'' realization
of the non-Abelian Chern-Simons path integral for manifolds
$M$ of the form $M=\Sigma \times S^1$ and arbitrary simply-connected
compact structure groups $G$.
More precisely, we introduce, for general links $L$ in $M$,
 a rigorous simplicial version $\WLO_{rig}(L)$  of  the corresponding
 Wilson loop observable $\WLO(L)$ in the so-called
 ``torus gauge''  by Blau and Thompson (Nucl. Phys. B408(2):345--390, 1993).
For a simple class of links $L$ we then evaluate $\WLO_{rig}(L)$
explicitly in a non-perturbative way,
 finding agreement with Turaev's shadow invariant $|L|$.
\end{abstract}

\medskip


\medskip

\section{Introduction}
\label{sec1}

Recall from  \cite{Ha7a} that our goal is to find,  for manifolds $M$ of the form $M=\Sigma \times S^1$,
 a rigorous realization of the
non-Abelian Chern-Simons path integral in the torus gauge.
We want to achieve this with the help of a suitable ``simplicial'' approach.
In \cite{Ha7a} we introduced  such an approach  and stated without proof our main result, i.e. Theorem 6.4 in \cite{Ha7a} (= Theorem \ref{main_theorem} in Sec. \ref{subsec4.10} below, to be  proven in Sec. \ref{sec5}
below).

In the present paper  we will briefly recall the simplicial approach in \cite{Ha7a}.
In order to understand the motivation for some of the definitions and constructions
(e.g. in Sec. \ref{subsec4.1} and Sec. \ref{subsec4.3} below) the reader will probably find it helpful  to have a look at Sec. 5 in \cite{Ha7a} (and at  Sec. 1 and Sec. 3 of \cite{Ha7a}).\par

The paper is organized as follows:\par

In Sec. \ref{sec2}  we will recall some of the notation from \cite{Ha7a}
and we will recall the basic heuristic formula from \cite{Ha7a}, cf. Eq. \eqref{eq2.48} below.
In  Sec. \ref{sec4} we recall the discretization approach of \cite{Ha7a}
and restate the main result,  i.e. Theorem \ref{main_theorem}.
In Sec. \ref{sec3} we study ``oscillatory Gauss-type measures'' on Euclidean spaces
and discuss some of their properties.
In Sec. \ref{sec5} we prove Theorem \ref{main_theorem}.
 In Sec. \ref{sec6} we comment briefly on the case of general (simplicial ribbon) links.
   In Sec. \ref{sec7}  we then give an outlook on some promising further directions within the framework of $BF_3$-theory before we conclude the main part of this paper
   with a short discussion of our results in  Sec. \ref{sec8}. \par

 The present paper has an appendix  consisting of four parts:
part \ref{appB} contains a list of the Lie theoretic notation which will be relevant
in the present paper. (This list is a continuation of the list in Appendix A in
\cite{Ha7a}).
In part \ref{appA} we recall the definition of Turaev's shadow invariant $|L|$ for links $L$
in  3-manifolds $M$ of the type $M=\Sigma \times S^1$.
In part \ref{appC'} we recall the definition of $BF$-theory in 3 dimensions and we
briefly comment on the relationship between $BF_3$-theory and CS theory.
In  part \ref{appJ} we sketch  possible reformulations/modifications of the discretization approach
in Sec. \ref{subsec7.3}.

\section{The basic heuristic formula in  \cite{Ha7a}}
\label{sec2}

\subsection{Basic  spaces}
\label{subsec2.1}

As in \cite{Ha7a} we fix
a   simply-connected compact Lie group
 $G$  and  a maximal torus $T$ of $G$.
By  $\cG$ and $\ct$ we  will denote the Lie algebras of $G$ and  $T$
and by $\langle \cdot , \cdot \rangle_{\cG}$ or simply by $\langle \cdot , \cdot \rangle$
 the unique $\Ad$-invariant scalar product
on $\cG$ satisfying the normalization condition $\langle \Check{\alpha} , \Check{\alpha} \rangle = 2$
for every short coroot $\Check{\alpha}$ w.r.t. $(\cG,\ct)$, cf. part \ref{appB} of the Appendix.
   For later use let us also fix a  Weyl chamber   $\CW \subset \ct$.\par

Moreover, we will fix  a compact oriented 3-manifold $M$
of the form $M = \Sigma \times S^1$ where $\Sigma$ is a (compact oriented) surface,
and an ordered oriented link  $L=(l_1, \ldots, l_m)$, $m \in \bN$,
 in $M= \Sigma \times S^1$.
Each $l_i$ is ``colored'' with an irreducible, finite-dimensional, complex representation $\rho_i$ of $G$.

\smallskip

As in  \cite{Ha7a} we will use the following notation\footnote{recall that  $\Omega^p(N,V)$ denotes the space of $V$-valued  $p$-forms
on a smooth manifold $N$}
\begin{subequations} \label{eq_basic_spaces_cont}
\begin{align}
\cB & = C^{\infty}(\Sigma,\ct)  \cong \Omega^0(\Sigma,\ct)\\
\cA & =  \Omega^1(M,\cG)\\
\cA_{\Sigma} & =  \Omega^1(\Sigma,\cG) \\
\cA_{\Sigma,\ct} & = \Omega^1(\Sigma,\ct), \quad  \cA_{\Sigma,\ck}  = \Omega^1(\Sigma,\ck) \\
\cA^{\orth} & =  \{ A \in \cA \mid A(\partial/\partial t) = 0\}\\
\label{eq_part_f}
\Check{\cA}^{\orth} & = \{ A^{\orth} \in \cA^{\orth} \mid \int A^{\orth}(t) dt \in \cA_{\Sigma,\ck} \} \\
\label{eq_part_g}   \cA^{\orth}_c & = \{ A^{\orth} \in \cA^{\orth} \mid \text{ $A^{\orth}$ is constant and
 $\cA_{\Sigma,\ct}$-valued}\}
 \end{align}
\end{subequations}
Here $\ck$ is the orthogonal complement of $\ct$ in $\cG$ w.r.t.
$\langle \cdot, \cdot \rangle$, $dt$ is the normalized (translation-invariant) volume form on $S^1$,
 $\partial/\partial t$ is the vector field on $M=\Sigma \times S^1$
obtained by ``lifting'' the standard vector field $\partial/\partial t$ on $S^1$
 and  in Eqs. \eqref{eq_part_f} and \eqref{eq_part_g}
 we used the ``obvious'' identification (cf. Sec. 2.3.1 in \cite{Ha7a})
\begin{equation}
\cA^{\orth}  \cong C^{\infty}(S^1,\cA_{\Sigma})
\end{equation}
where  $C^{\infty}(S^1,\cA_{\Sigma})$ is the space of maps
$f:S^1 \to \cA_{\Sigma}$ which are ``smooth'' in the sense that
$\Sigma \times S^1 \ni (\sigma,t) \mapsto (f(t))(X_{\sigma}) \in \cG$
is smooth for every smooth vector field $X$ on $\Sigma$.
It follows from the definitions above that
\begin{equation} \label{eq_cAorth_decomp}
\cA^{\orth} = \Check{\cA}^{\orth} \oplus  \cA^{\orth}_c
\end{equation}

\subsection{The heuristic Wilson loop observables}
\label{subsec2.1b}

Recall that in the special case when $G$ is simple\footnote{see Remark \ref{rm2.1} below for the
case when $G$  is not simple} the Chern-Simons action function $S_{CS}: \cA \to \bR$
associated to $M$, $G$, and the ``level''  $k \in \bZ \backslash \{0\}$ is   given by
 \begin{equation} \label{eq2.2'} S_{CS}(A) = - k \pi \int_M \langle A \wedge dA \rangle
   + \tfrac{1}{3} \langle A\wedge [A \wedge A]\rangle, \quad
A \in \cA \end{equation}
where $[\cdot \wedge \cdot]$  denotes the wedge  product associated to the
Lie bracket $[\cdot,\cdot] : \cG \times \cG \to \cG$
and where   $\langle \cdot \wedge  \cdot \rangle$
  denotes the wedge product  associated to the
 scalar product $\langle \cdot , \cdot \rangle : \cG \times \cG \to \bR$.\par

Recall also that the heuristic Wilson loop observable $\WLO(L)$ of a link $L=(l_1,l_2,\ldots,l_m)$ in $M$
with ``colors'' $(\rho_1,\rho_2,\ldots,\rho_m)$ is given by the informal expression
\begin{equation} \label{eq_WLO_orig}
\WLO(L) := \int_{\cA} \prod_i  \Tr_{\rho_i}(\Hol_{l_i}(A)) \exp( i S_{CS}(A)) DA
\end{equation}
where $\Hol_l(A)$ is the holonomy of $A \in \cA$ around the loop $l \in \{l_1, \ldots, l_m\}$.
The following explicit formula for $\Hol_l(A)$ proved to be useful in \cite{Ha7a}:
\begin{equation} \label{eq_Hol_heurist}
\Hol_l(A) = \lim_{n \to \infty} \prod_{k=1}^n \exp\bigl(\tfrac{1}{n}  A(l'(t))\bigr)_{| t=k/n}
\end{equation}
where $\exp:\cG \to G$ is the exponential map of $G$.

\begin{remark} \label{rm_sec1.1} \rm One can assume without loss of generality
that  $G$ is a closed subgroup of $U({\mathbf N})$ for some $\mathbf N \in \bN$.
In the special case where   $G$ is simple we can then   rewrite Eq. \eqref{eq2.2'}  as
$$ S_{CS}(A) =  k \pi \int_M \Tr\bigl( A \wedge dA
   + \tfrac{2}{3} A\wedge A \wedge A\bigr)$$
with
$\Tr:= c \cdot \Tr_{\Mat({\mathbf N},\bC)}$ where $c \in \bR$
is chosen such that $\langle A, B \rangle = - \Tr(A \cdot B)$ for all $A, B \in \cG \subset u(N) \subset
\Mat({\mathbf N},\bC)$.
Clearly, making this assumption is a bit inelegant
but it has some practical advantages, which is why we made
use of it in \cite{Ha7a}.
In the present paper we will use
this assumption only at a later stage, namely in part \ref{appC'} of the Appendix below (with $G$
replaced by $\tilde{G}$).
\end{remark}

\begin{remark} \label{rm2.1}  \rm
Recall from Remark 2.2 in \cite{Ha7a}
that if $G$ is a general  simply-connected compact Lie group then
$G$ will be  of the form $G = \prod_{i=1}^{r} G_i$, $r \in
\bN$, where each $G_i$ is a simple simply-connected compact Lie
group. We can generalize the definition of $S_{CS}$  to this general situation
by setting  -- for any fixed  sequence $(k_i)_{i \le r}$ of non-zero integers --
$$S_{CS}(A) :=  \sum_{i=1}^r S_{CS,i}(A_i) \quad \forall A \in \cA$$
where $S_{CS,i}$ is the  Chern-Simons action function associated to $M$, $G_i$, and  $k_i$
and where $(A_i)_i$ are the components of $A$ w.r.t. to the decomposition
 $\cG = \oplus_{i=1}^r \cG_i$ ($\cG_i$ being the Lie algebra of $G_i$).

 \smallskip

In the present paper  only two special cases will play a role, namely the case
 $r=1$ (i.e. $G$ simple) and the case $r=2$, $G_2 = G_1$ and $k_2
= - k_1$, cf. Sec. \ref{sec7} below.
\end{remark}

\subsection{The basic heuristic formula}
\label{subsec2.2}

The starting point for the main part of \cite{Ha7a} was a second heuristic formula for $\WLO(L)$
which one obtains from Eq. \eqref{eq_WLO_orig} above by applying ``torus gauge fixing'',
 cf. Sec. 2.2.4 in  \cite{Ha7a}.

Let $\sigma_0 \in \Sigma \backslash \bigl( \bigcup_{i=1}^m \arc(l^i_{\Sigma})\bigr)$ be fixed. Here we have set
$l^i_{\Sigma} := \pi_{\Sigma} \circ l_i$,  $i \le m$,  where $\pi_{\Sigma}: \Sigma \times S^1 \to \Sigma$
is the canonical projection.
Then we have (cf. Eq. (2.53) in \cite{Ha7a})
\begin{multline}  \label{eq2.48} \WLO(L)
 \sim \sum_{y \in I}  \int_{\cA^{\orth}_c \times \cB} \biggl\{
 1_{C^{\infty}(\Sigma,\ct_{reg})}(B)  \Det_{FP}(B)\\
 \times   \biggl[ \int_{\Check{\cA}^{\orth}} \biggl( \prod_{i=1}^m  \Tr_{\rho_i}\bigl(
 \Hol_{l_i}(\Check{A}^{\orth} + A^{\orth}_c, B)\bigr) \biggr)
\exp(i  S_{CS}( \Check{A}^{\orth}, B)) D\Check{A}^{\orth} \biggr] \\
 \times \exp\bigl( - 2\pi i k  \langle y, B(\sigma_0) \rangle \bigr) \biggr\}
 \exp(i S_{CS}(A^{\orth}_c, B)) (DA^{\orth}_c \otimes DB)
\end{multline}
where $I:= \ker(\exp_{| \ct})  \subset \ct$
and  $\ct_{reg}:= \exp^{-1}(T_{reg})$ (with $T_{reg}$ denoting the set of regular elements of $T$)
 and where
 for each $B \in \cB$, $A^{\orth} \in \cA^{\orth}$ we have set
\begin{align}
S_{CS}(A^{\orth},B) & := S_{CS}(A^{\orth} + B dt )\\
\label{eq_Hol_heurist_abbr} \Hol_{l}(A^{\orth},  B) &  := \Hol_{l}(A^{\orth}    + B dt)
\end{align}
Here  $dt$ is the real-valued 1-form on  $M=\Sigma \times S^1$
obtained by  pulling back the 1-form $dt$ on $S^1$ in the obvious way.
Finally, $\Det_{FP}(B)$ is the informal expression   given by
\begin{equation} \label{eq_DetFP} \Det_{FP}(B) :=    \det\bigl(1_{\ck}-\exp(\ad(B))_{|\ck}\bigr)
\end{equation}

For the rest of this paper we will now
 fix an auxiliary Riemannian metric ${\mathbf g}$ on $\Sigma$.
 After doing so  we obtain a scalar product   $\ll \cdot , \cdot \gg_{\cA^{\orth}}$ on
  $\cA^{\orth} \cong C^{\infty}(S^1, \cA_{\Sigma})$
 in a natural way. Moreover, we then have a well-defined
Hodge star operator
 $\star: \cA_{\Sigma} \to \cA_{\Sigma}$ which induces an operator
 $\star: C^{\infty}(S^1, \cA_{\Sigma}) \to C^{\infty}(S^1, \cA_{\Sigma})$ in the obvious way.
 According to  Eq. (2.48) in \cite{Ha7a}
 we then have the following explicit formula
 \begin{equation} \label{eq_SCS_expl0} S_{CS}(A^{\orth},B)  =  \pi k  \ll A^{\orth},
\star  \bigl(\tfrac{\partial}{\partial t} + \ad(B) \bigr) A^{\orth} \gg_{\cA^{\orth}}
 +  2 \pi k  \ll\star  A^{\orth},  dB \gg_{\cA^{\orth}}
\end{equation}
for all $B \in \cB$ and $A^{\orth} \in \cA^{\orth}$,
 which implies
 \begin{align} \label{eq_SCS_expl}
S_{CS}(\Check{A}^{\orth},B) & =  \pi k  \ll \Check{A}^{\orth},
\star  \bigl(\tfrac{\partial}{\partial t} + \ad(B) \bigr) \Check{A}^{\orth} \gg_{\cA^{\orth}} \\
\label{eq_SCS_expl2}
 S_{CS}(A^{\orth}_c,B) & =   2 \pi k  \ll\star  A^{\orth}_c,  dB \gg_{\cA^{\orth}}
\end{align}
for  $B \in \cB$, $\Check{A}^{\orth} \in \Check{\cA}^{\orth}$, and $A^{\orth}_c \in \cA^{\orth}_c$.

\begin{remark} \label{rm2.2} \rm
In view of Sec. \ref{subsec4.10a} below we recall that --
according to Remark 2.7 in \cite{Ha7a}  --
we can replace  space $\cB$ appearing in the outer integral  $\int_{\cA^{\orth}_c \times \cB} \cdots (DA^{\orth}_c \otimes DB)$ in Eq. \eqref{eq2.48} above by the space
$$ \cB^{loc}_{\sigma_0}:= \{B \in \cB | B \text{ is locally constant around $\sigma_0$} \}$$
\end{remark}

\section{Simplicial realization of $\WLO(L)$}
\label{sec4}

In the present section we briefly recall the definition of the rigorous simplicial analogue
 $\WLO_{rig}(L)$ for the RHS of Eq. \eqref{eq2.48} above which we gave in Sec. 5 in \cite{Ha7a}
 and we recall the main result of \cite{Ha7a}, namely Theorem 6.4 (= Theorem \ref{main_theorem} below).

 \smallskip

 Anyway, the reader will probably find it useful to
 have a look at Sec. 4 and Sec. 5 in \cite{Ha7a} where we explain
  in  much more detail the motivation of our  constructions.

\setcounter{subsection}{-1}
\subsection{Review of the simplicial setup in Sec. 4 in \cite{Ha7a}}
\label{subsec4.00}

Recall from Sec. 4.1 in \cite{Ha7a}
that for a finite oriented polyhedral cell complex $\cP$
we denote by $\face_p(\cP)$, $ p \in \bN_0$, the set of $p$-faces of $\cP$,
 and --  for  every fixed real vector space $V$ --
we denoted by  $C^p(\cP,V)$ the space of maps $\face_p(\cP) \to V$  (``$V$-valued
 $p$-cochains of $\cP$''). The elements of $\face_0(\cP)$ (resp. $\face_1(\cP)$)
  will be called the ``vertices'' (resp. ``edges'') of $\cP$. Instead of  $C^p(\cP,\bR)$ we will often write $C_p(\cP)$.  By $d_{\cP}$ we will denote the usual coboundary operator $C^p(\cP,V) \to C^{p+1}(\cP,V)$.

 \medskip

In \cite{Ha7a} we actually only considered the special situation
$\cP \in \{\bZ_N,\cK,\cK',q\cK, \cK \times \bZ_N,\cK' \times \bZ_N, q\cK \times \bZ_N \}$
where $\bZ_N$, $\cK$, $\cK'$, and $q\cK$ are given as follows:

\smallskip

Recall that in Sec. 4.4 in \cite{Ha7a}
we fixed $N \in \bN$ and used  the finite cyclic group
$\bZ_{N}$ with the ``obvious''\footnote{i.e. the set of edges is given by $\{(t,t+1) \mid t \in \bZ_N\}$} (oriented) graph  structure as a discrete analogue of the Lie group $S^1$.

\medskip

Moreover, we fixed a finite oriented smooth polyhedral cell decomposition
$\cC$  of $\Sigma$. By $\cC'$ we denoted the dual polyhedral cell decomposition, equipped with an orientation.
By $\cK$ and $\cK'$  we denoted the corresponding (oriented)
  polyhedral cell complexes, i.e. $\cK := (\Sigma,\cC)$ and $\cK' := (\Sigma,\cC')$.\par

Instead of $\cK$ (resp. $\cK'$)
 we usually wrote $K_1$ (resp. $K_2$)
  and  we set $K:= (K_1,K_2)$.

\smallskip

We then introduced a joint sub division $q\cK:= (\Sigma,q\cC)$ of $\cK = K_1$ and $\cK' = K_2$
which can be characterized by the conditions
$$\face_0(q\cK) = \face_0(b\cK), \quad
\face_1(q\cK) = \face_1(b\cK) \backslash \{e \in \face_1(b\cK) \mid \text{ both endpoints of $e$
lie in $\face_0(K_1) \sqcup  \face_0(K_2)$} \}$$
where $b\cK$ is the barycentric sub division of $\cK$.
The set $\face_2(q\cK)$ is uniquely determined by $\face_0(q\cK)$ and $\face_1(q\cK)$.
Observe that each $F \in \face_2(q\cK)$ is a tetragon.
We introduced the notation
 \begin{equation}  \label{eq_def_K1|K2} \face_0(K_1 | K_2) := \face_0(q\cK) \backslash (\face_0(K_1) \cup \face_0(K_2))
\end{equation}
Recall that also the faces of $q\cK$ were equipped with an orientation.
 For convenience we chose the orientation on the edges of  $q\cK$  to be ``compatible''\footnote{more precisely,
 for  each $e \in \face_1(q\cK)$ we choose the orientation which is induced by orientation of the unique
 edge $e' \in \face_1(K_1) \cup \face_1(K_2)$ which contains $e$} with
 the orientation on the edges of $K_1$ and $K_2$.

\medskip

Recall from Sec. 4.2 in \cite{Ha7a}  that
a  ``simplicial curve'' in $\cP$ is a finite sequence
$x=(x^{(k)})_{k \le n}$, $n \in \bN$, of  vertices in $\cP$
 such that for every $k \le n$
the two vertices  $x^{(k)}$ and  $x^{(k+1)}$
   either coincide or are the two endpoints
 of an edge $e \in \face_1(\cP)$.
  If $x^{(n)} = x^{(1)}$ we will call
$x= (x^{(k)})_{k \le n}$ a ``simplicial loop'' in $\cP$.\par
Recall also that every simplicial curve  $x= (x^{(k)})_{k \le n}$ with $n > 1$ induces a
sequence $(e^{(k)})_{k \le n-1}$ of ``generalized edges'', i.e. elements of
$\face_1(\cP) \cup \{0\} \cup (- \face_1(\cP))\subset
C_1(\cP)$ in a canonical way.
Unless $x=(x^{(k)})_{k \le n}$ is constant, we can reconstruct $x=(x^{(k)})_{k \le n}$
from $(e^{(k)})_{k \le n-1}$. We write $\start e^{(k)}$ instead of $x^{(k)}$ (for $k \le n-1$).

\medskip

Recall from Sec. 4.3 in \cite{Ha7a}  that a
``(closed)\footnote{we will often omit the word ``closed''}
 simplicial  ribbon''  in $\cP$ is a finite sequence $R = (F_i)_{i \le n}$ of 2-faces of $\cP$
such that every $F_i$ is a tetragon
and such that  $F_i \cap  F_{j} = \emptyset$ unless $i = j$ or $j= i \pm 1$ (mod n).
In the latter case  $F_i$ and $F_{j}$ intersect in  a (full) edge.\par

\begin{convention} \label{conv1} \rm

Let $V$ be a fixed finite-dimensional real vector space.
\begin{enumerate}
\item We set
$C_1(K):= C_1(K_1) \oplus C_1(K_2)$ \ and \ $C^1(K,V):= C^1(K_1,V) \oplus C^1(K_2,V) $.

\item Let $\psi: C_1(K) \to C_1(q\cK)$ be the (injective) linear map given by
$$\psi(e)=e_1 + e_2 \quad \text{ for all } \quad e \in \face_1(K_1) \cup  \face_1(K_2)$$
where $e_1 = e_1(e), e_2 = e_2(e) \in  \face_1(q\cK)$ are the two edges of $q\cK$
``contained'' in $e$.
In the following we will identify $C_1(K)$ with the subspace $\psi(C_1(K))$ of $C_1(q\cK)$.
Moreover, using the identifications
$  C^1(q\cK,V) \cong C_1(q\cK) \otimes_{\bR} V$ and $C^1(K,V) \cong C_1(K) \otimes_{\bR} V$
 we  naturally obtain the linear map
$$\psi^V := \psi \otimes \id_V: C^1(K,V) \to C^1(q\cK,V)$$
We will identify $C^1(K,V)$ with the subspace  $\psi^V(C^1(K,V))$ of $C^1(q\cK,V)$.
\end{enumerate}

\end{convention}

\subsection{The basic spaces}
\label{subsec4.0}

As in Sec. 5.1 in \cite{Ha7a} we introduce
the following discrete analogues of the spaces $\cB$, $\cA_{\Sigma}$ and $\cA^{\orth}$
in Sec. \ref{subsec2.1} above:
\begin{subequations} \label{eq_basic_spaces}
\begin{align}
 \cB(qK) & :=C^0(q\cK,\ct) \\
 \cA_{\Sigma}(q\cK) & := C^1(q\cK,\cG) \\
\cA^{\orth}(q\cK) & := \Map(\bZ_N,\cA_{\Sigma}(q\cK))
 \end{align}
 \end{subequations}

\noindent
  Clearly, the scalar product $\langle \cdot, \cdot \rangle_{\cG}$ on $\cG$  induces
 scalar products $\ll \cdot, \cdot\gg_{\cB(q\cK)}$ and $\ll \cdot, \cdot\gg_{\cA_{\Sigma}(q\cK)}$
 on $\cB(q\cK)$ and $\cA_{\Sigma}(q\cK)$  in the standard way.
 We introduce a  scalar product
  $\ll \cdot , \cdot \gg_{\cA^{\orth}(q\cK)}$      on $\cA^{\orth}(q\cK) = \Map(\bZ_N,\cA_{\Sigma}(q\cK))$ by
   \begin{equation} \label{eq_norm_scalarprod}
\ll A^{\orth}_1 , A^{\orth}_2 \gg_{\cA^{\orth}(q\cK)}  =   \tfrac{1}{N} \sum_{t \in \bZ_N} \ll
A^{\orth}_1(t) , A^{\orth}_2(t) \gg_{\cA_{\Sigma}(q\cK)}
\end{equation} for all $A^{\orth}_1 , A^{\orth}_2 \in \cA^{\orth}(q\cK)$.

\begin{convention} \rm \label{conv1neu}
We identify $\cA_{\Sigma}(q\cK)$ with the subspace
$\{ A^{\orth} \in \Map(\bZ_N,\cA_{\Sigma}(q\cK)) \mid A^{\orth}  \text{ is constant}\}$  of $\cA^{\orth}(q\cK)$
 in the obvious way.
\end{convention}

For technical reasons\footnote{namely, in order to obtain a nice simplicial analogue of the Hodge star operator,
cf. Sec. \ref{subsec4.2} below}  we will not only work with the full spaces $\cA_{\Sigma}(q\cK)$ and $\cA^{\orth}(q\cK)$
but also their subspaces  (cf. Convention \ref{conv1} above)
 $\cA_{\Sigma}(K)$ and $\cA^{\orth}(K)$ given by
\begin{align}
\cA_{\Sigma}(K) & :=C^1(K_1,\cG)  \oplus  C^1(K_2,\cG) \subset \cA_{\Sigma}(q\cK)\\
\cA^{\orth}(K) & := \Map(\bZ_N,\cA_{\Sigma}(K)) \subset \cA^{\orth}(q\cK)
\end{align}

\subsubsection*{The decomposition $\cA^{\orth}(K) =  \Check{\cA}^{\orth}(K) \oplus \cA^{\orth}_c(K)$}

In order to obtain  a  discrete analogue of the decomposition $\cA^{\orth} =  \Check{\cA}^{\orth}
\oplus \cA^{\orth}_c$ in  Eq. \eqref{eq_cAorth_decomp} above
let us introduce the following spaces:
\begin{subequations}
\begin{align}
  \cA_{\Sigma,\ct}(K) & := C^1(K_1,\ct) \oplus C^1(K_2,\ct)\\
    \cA_{\Sigma,\ck}(K) & := C^1(K_1,\ck) \oplus C^1(K_2,\ck)\\
\label{eq_CheckcA_disc} \Check{\cA}^{\orth}(K) & :=
 \{ A^{\orth} \in  \cA^{\orth}(K) \mid
\sum\nolimits_{t \in \bZ_{N}}  A^{\orth}(t) \in  \cA_{\Sigma,\ck}(K) \} \\
 \cA^{\orth}_c(K) & := \{ A^{\orth} \in  \cA^{\orth}(K) \mid
 \text{ $A^{\orth}(\cdot)$ is constant and $ \cA_{\Sigma,\ct}(K)$-valued}\} \cong \cA_{\Sigma,\ct}(K)
 \end{align}
 \end{subequations}
 Observe that  we have
\begin{equation} \label{eq4.30}
\cA^{\orth}(K) =  \Check{\cA}^{\orth}(K) \oplus \cA^{\orth}_c(K)
\end{equation}
which is indeed a  discrete analogue of the decomposition $\cA^{\orth} =  \Check{\cA}^{\orth}
\oplus \cA^{\orth}_c$ in  Eq. \eqref{eq_cAorth_decomp} above.

\subsection{Discrete analogue of the operator $\tfrac{\partial}{\partial t} + \ad(B):\cA^{\orth} \to \cA^{\orth}$}
\label{subsec4.1}

Let us recall the definition of the operator $L^{(N)}(B):\cA^{\orth}(K) \to \cA^{\orth}(K)$
which we introduced in \cite{Ha7a} as the discrete analogue of the continuum operator
$\tfrac{\partial}{\partial t} + \ad(B):\cA^{\orth} \to \cA^{\orth}$ in Eq. \eqref{eq_SCS_expl0} above.\par

Let $\tau_x$, for  $x \in \bZ_N$,  denote   the translation operator  $\Map(\bZ_N,\cG) \to  \Map(\bZ_N,\cG)$  given by $(\tau_x f)(t) = f(t +x)$ for all $t \in \bZ_N$ and $f \in \Map(\bZ_N,\cG)$.
 Instead of $\tau_0$ we will simply write $1$ in the following.\par

In \cite{Ha7a} we introduced, for fixed $b \in \ct$,
the following  natural
discrete analogues $L^N(b): \Map(\bZ_N,\cG) \to \Map(\bZ_N,\cG)$
 of  the continuum operator $L(b):= \tfrac{\partial}{\partial t} + \ad(b): C^{\infty}(S^1,\cG) \to  C^{\infty}(S^1,\cG)$
 (cf.  Sec. 5.2 in \cite{Ha7a} where we also explain why these operators really are natural):
\begin{subequations}  \label{eq_def_LOp}
\begin{align}
\label{eq_def_LOp_a} \hat{L}^{(N)}(b) & := N( \tau_1 e^{\ad(b)/N} -1)\\
\label{eq_def_LOp_b} \Check{L}^{(N)}(b) & := N(1 - \tau_{-1} e^{-\ad(b)/N}) \\
\label{eq_def_LOp_c} \bar{L}^{(N)}(b) & := \tfrac{N}{2}( \tau_1 e^{\ad(b)/N} - \tau_{-1} e^{-\ad(b)/N}) \quad \text{ if $N$ is even }
\end{align}
\end{subequations}

Let $B \in \cB(q\cK)$. The operator  $L^{(N)}(B):\cA^{\orth}(K) \to \cA^{\orth}(K)$
mentioned above is the linear operator which, under the identification
 $$\cA^{\orth}(K) \cong \Map(\bZ_N, C^1(K_1,\cG)) \oplus \Map(\bZ_N, C^1(K_2,\cG)),$$
 is   given by
 \begin{equation} \label{def_LN} L^{(N)}(B) = \left( \begin{matrix}
 \hat{L}^{(N)}(B) && 0 \\
0 && \Check{L}^{(N)}(B)
\end{matrix} \right)
\end{equation}
Here the linear operators $\hat{L}^{(N)}(B):  \Map(\bZ_N, C^1(K_1,\cG))  \to \Map(\bZ_N, C^1(K_1,\cG)) $
and $\Check{L}^{(N)}(B):  \Map(\bZ_N, C^1(K_2,\cG))  \to \Map(\bZ_N, C^1(K_2,\cG)) $
are given by
\begin{subequations}
\begin{align} \label{eq_LN_ident1}
\hat{L}^{(N)}(B)) & \cong \oplus_{\bar{e} \in  \face_0(K_1 | K_2)}
 \hat{L}^{(N)}(B(\bar{e})) \\
 \label{eq_LN_ident2} \Check{L}^{(N)}(B)) & \cong \oplus_{\bar{e} \in  \face_0(K_1 | K_2)}
 \Check{L}^{(N)}(B(\bar{e}))
\end{align}
\end{subequations}
where $\face_0(K_1 | K_2)$ is as in Eq. \eqref{eq_def_K1|K2} above.
In Eqs.  \eqref{eq_LN_ident1} and \eqref{eq_LN_ident2} we used the obvious identification
$$ \Map(\bZ_N, C^1(K_j,\cG)) \cong \oplus_{e \in  \face_1(K_j)}
\Map(\bZ_N,\cG) \cong \oplus_{\bar{e} \in  \face_0(K_1 | K_2)}
\Map(\bZ_N,\cG)$$

We remark that  $L^{(N)}(B)$ leaves the subspace $\Check{\cA}^{\orth}(K)$ of $\cA^{\orth}(K)$
 invariant. The restriction of $L^{(N)}(B)$  to $\Check{\cA}^{\orth}(K)$
 will also be denoted by $L^{(N)}(B)$ in the following.

\subsection{Definition of $S^{disc}_{CS}(A^{\orth},B)$}
\label{subsec4.2}

Recall that   in \cite{Ha7a} we introduced
discrete Hodge operators
$\star_{K_1}: C^1(K_1,\cG) \to C^{1}(K_2,\cG)$ and $\star_{K_2}: C^1(K_2,\cG) \to C^{1}(K_1,\cG)$, cf. Sec. 4.5 in \cite{Ha7a}.
Moreover, we introduced two different operators denoted by $\star_K$ (cf. Sec. 4.5 and Sec. 5.3 in \cite{Ha7a}).
Firstly, the operator   $\star_K:  \cA_{\Sigma}(K) \to  \cA_{\Sigma}(K) =  C^1(K_1,\cG) \oplus C^1(K_2,\cG)$
given by
\begin{equation} \label{eq_Hodge_matrix}
\star_K := \left(\begin{matrix} 0 && \star_{K_2} \\ \star_{K_1} && 0 \end{matrix}
 \right)
\end{equation}
and,  secondly, the operator   $\star_K: \cA^{\orth}(K)  \to  \cA^{\orth}(K) $ given by
\begin{equation} \label{eq_star_K_vor_rm}
 (\star_K A^{\orth})(t) = \star_K (A^{\orth}(t)) \quad \quad \forall A^{\orth} \in \cA^{\orth}(K), t \in \bZ_N
 \end{equation}

As the  discrete analogues of the continuum expression
$S_{CS}(A^{\orth},B)$ in Eq. \eqref{eq_SCS_expl0} above
we  use the expression
 \begin{subequations}  \label{eq_SCS_expl_disc}
\begin{equation}  S^{disc}_{CS}(A^{\orth},B) :=   \pi  k \biggl[ \ll A^{\orth},
\star_K  L^{(N)}(B)
   A^{\orth} \gg_{\cA^{\orth}(q\cK)}
 + 2 \ll  \star_K A^{\orth},  d_{q\cK}  B \gg_{\cA^{\orth}(q\cK)}  \biggr]
\end{equation}
  for $B \in \cB(q\cK)$,  $A^{\orth} \in  \cA^{\orth}(K) \subset \cA^{\orth}(q\cK)$.
Observe that this implies
 \begin{align}  \label{eq_SCS_expl_discb} S^{disc}_{CS}(\Check{A}^{\orth},B) & =   \pi  k  \ll \Check{A}^{\orth},
\star_K  L^{(N)}(B)    \Check{A}^{\orth} \gg_{\cA^{\orth}(q\cK)} \\
\label{eq_SCS_expl_discc} S^{disc}_{CS}(A^{\orth}_c,B) & =  2 \pi  k  \ll  \star_K A^{\orth}_c,  d_{q\cK}  B \gg_{\cA^{\orth}(q\cK)}
\end{align}
 \end{subequations}
for $B \in \cB(q\cK)$,  $\Check{A}^{\orth} \in  \Check{\cA}^{\orth}(K)$, $A^{\orth}_c \in  \cA^{\orth}_c(K)$.

\smallskip

Recall that we have (cf. Proposition 5.3 in \cite{Ha7a}):
\begin{proposition} \label{prop4.2}
The  operator  $\star_K L^{(N)}(B): \cA^{\orth}(K) \to \cA^{\orth}(K)$
is symmetric  w.r.t to the scalar product $\ll  \cdot, \cdot  \gg_{\cA^{\orth}(q\cK)}$.
\end{proposition}

\subsection{Definition of $\Hol^{disc}_{R}(A^{\orth},  B)$  }
\label{subsec4.3}

Let  $A^{\orth} \in \cA^{\orth}(K) \subset \cA^{\orth}(q\cK)$ and $B \in \cB(q\cK)$.
Moreover, let  $R  = (F_k)_{k \le n}$, $n \in \bN$, be a  closed simplicial
ribbon in $q\cK \times \bZ_N$.
According to Remark 4.3 in \cite{Ha7a}
 $R$ induces a pair  $(l,l')$ of
  simplicial loops $l = (l^{(k)})_{k \le n}$ and   $l' = (l^{'(k)})_{k \le n}$
      in $q\cK \times \bZ_N$ in the obvious way ($l$ and $l'$ are simply the two loops ``on the boundary'' of $R$).
Let $l_{\Sigma}$, $l'_{\Sigma}$, $l_{S^1}$, $l'_{S^1}$  denote
the corresponding ``projected'' simplicial loops in $q\cK$ and $\bZ_N$,
 cf. Sec. 4.4.4 in \cite{Ha7a}.\par

The simplicial analogue of the continuum expression
$\Hol_{l}(A^{\orth},   B)$ we used in \cite{Ha7a} (cf.  Sec. 5.4 in \cite{Ha7a}; here it is very helpful
to recall the discussion in Sec. 5.4 in \cite{Ha7a} since otherwise the motivation
for the RHS of Eq. \eqref{eq4.21} may not become clear) was
\begin{multline} \label{eq4.21}
\Hol^{disc}_{R}(A^{\orth},   B) :=
  \prod_{k=1}^n \exp\biggl( \tfrac{1}{2}\bigl(A^{\orth}(\start l^{(k)}_{S^1})\bigr)(l^{(k)}_{\Sigma}) + \tfrac{1}{2} \bigl(A^{\orth}(\start l^{'(k)}_{S^1})\bigr)(l^{'(k)}_{\Sigma}) \\
  +  \tfrac{1}{2} B(\start l^{(k)}_{\Sigma})
 \cdot dt^{(N)}(l^{(k)}_{S^1}) + \tfrac{1}{2} B(\start l^{'(k)}_{\Sigma})  \cdot dt^{(N)}(l^{'(k)}_{S^1}) \biggr)
\end{multline}
  where $dt^{(N)} \in C^1(\bZ_{N},\bR) \cong \Hom_{\bR}(C_1(\bZ_{N}),\bR)$ is   given by
 \begin{equation}dt^{(N)}(e)= \tfrac{1}{N} \quad \quad \forall e \in \face_1(\bZ_N)
 \end{equation}
  and where we have
  made    the identification $\cA_{\Sigma}(q\cK) = C^1(q\cK,\cG) \cong \Hom(C_1(q\cK),\cG)$.

\begin{remark} \label{rm_full_ribbons} \rm In view of Remark \ref{rm_full_ribbons2} below let us point out that
 instead of working with simplicial ribbons in $q\cK \times \bZ_N$ (``half ribbons''\footnote{observe that
     every simplicial ribbon $R$  in $\cK \times \bZ_N$
      can be considered as the union of two   simplicial ribbons $R_+$
      and $R_-$  in $q\cK \times \bZ_N$ in a natural way})
   one could also work with simplicial ribbons in $\cK \times \bZ_N$ (``full ribbons'').
   Observe that every closed simplicial ribbon $R$ in
     $\cK \times \bZ_N$ induces three loops $l^+ = (l^{+(k)})_{k \le n}$, $l^- = (l^{-(k)})_{k \le n}$, and
     $l = (l^{(k)})_{k \le n}$, $n \in \bN$, in $q\cK \times \bZ_N$
     in a natural way,  $l^+$ and $l^-$ being the two boundary loops
     and $l$ being the loop ``inside''\footnote{More precisely, the loop $l$
      is the loop ``lying on the intersection'' of the two associated half ribbons $R_+$ and $R_-$}  $R$.
      We now define $\Hol^{disc}_{R}(A^{\orth},   B) $ by
   \begin{multline} \label{eq4.21_full}
\Hol^{disc}_{R}(A^{\orth},   B) :=
  \prod_{k=1}^n \exp\biggl( \bigl( \sum_{\pm} \tfrac{1}{4}\bigl(A^{\orth}(\start l^{\pm (k)}_{S^1})\bigr)(l^{\pm (k)}_{\Sigma}) \bigr) + \tfrac{1}{2} \bigl(A^{\orth}(\start l^{(k)}_{S^1})\bigr)(l^{(k)}_{\Sigma}) \\
  + \bigl( \sum_{\pm} \tfrac{1}{4} B(\start l^{\pm(k)}_{\Sigma})
 \cdot dt^{(N)}(l^{\pm(k)}_{S^1}) \bigr) + \tfrac{1}{2} B(\start l^{(k)}_{\Sigma})  \cdot dt^{(N)}(l^{(k)}_{S^1}) \biggr)
\end{multline}
where we use  the notation $\sum_{\pm} \cdots $ in the obvious way\footnote{eg,
$\sum_{\pm} \tfrac{1}{4}  A^{\orth}(\start l^{\pm(k)}_{S^1})(l^{\pm(k)}_{\Sigma})$
is a short form of $\tfrac{1}{4} A^{\orth}(\start l^{+(k)}_{S^1})(l^{+(k)}_{\Sigma}) +
\tfrac{1}{4} A^{\orth}(\start l^{-(k)}_{S^1})(l^{-(k)}_{\Sigma})$}.
 \end{remark}

\subsection{Definition of $\Det^{disc}_{FP}(B)$}
\label{subsec4.4}

 In  \cite{Ha7a} our original\footnote{recall that in \cite{Ha7a} we later modified this definition; we will do this also in the present paper,
 cf. Sec. \ref{subsec4.10a} below} ansatz  for the discrete analogue $\Det^{disc}_{FP}(B)$ of  the heuristic expression $\Det_{FP}(B) =
 \det(1_{\ck}-\exp(\ad(B))_{|\ck})$ given by Eq. \eqref{eq_DetFP} was
    \begin{align}  \label{eq_def_DetFPdisc}
 \Det^{disc}_{FP}(B) & :=
   \prod_{x \in \face_0(q\cK)}  \det\bigl(1_{{\ck}}-\exp(\ad(B(x)))_{| {\ck}}\bigr)
  \end{align}
 for every $B \in \cB(q\cK)$.

\subsection{Discrete version of $1_{C^{\infty}(\Sigma,\ct_{reg})}(B)$}
\label{subsec4.6}

Let us fix a family
$(1^{(s)}_{\ct_{reg}})_{s > 0}$
of elements of $C^{\infty}_{\bR}(\ct)$  with the following properties:
\begin{itemize}
\item $\Image(1^{(s)}_{\ct_{reg}}) \subset [0,1]$ \ and \
 $\supp(1^{(s)}_{\ct_{reg}}) \subset  \ct_{reg}$ \quad for each $s>0$,
\item $1^{(s)}_{\ct_{reg}} \to 1_{\ct_{reg}}$ pointwise as $s \to 0$,
\item Each $1^{(s)}_{\ct_{reg}}$, $s>0$,  is  invariant under the operation of the affine Weyl group $\cW_{\aff}$
 on $\ct$.
\end{itemize}

\noindent
For fixed $s > 0$ and  $B \in \cB(q\cK)$ we will now take
 the expression
\begin{equation}
\prod_{x} 1^{(s)}_{\ct_{reg}}(B(x)):=
\prod_{x \in \face_0(q\cK)}
1^{(s)}_{\ct_{reg}}(B(x))
\end{equation}
as  the discrete analogue of $1_{C^{\infty}(\Sigma,\ct_{reg})}(B)$.
Later we will let $s \to  0$.

\subsection{Discrete versions of the two Gauss-type measures in Eq. \eqref{eq2.48}}
  \label{subsec4.8}

\begin{convention}  \label{conv_EucSpaces} \rm
In the following we will always consider $\cB(q\cK)$, $\cA^{\orth}(K)$
and their subspaces
 as Euclidean spaces in the ``obvious''\footnote{More precisely, we will assume that the space
$\cB(q\cK)$ (or any subspace of $\cB(q\cK)$)  is equipped with the (restriction of the) scalar product $\ll \cdot, \cdot\gg_{\cB(q\cK)}$ on $\cB(q\cK)$, and the space $\cA^{\orth}(K)$ (or any subspace of $\cA^{\orth}(K)$)
is equipped with the  restriction of the scalar product $\ll \cdot, \cdot\gg_{\cA^{\orth}(q\cK)}$,
introduced in Sec. \ref{subsec4.0} above} way.
\end{convention}

i) Let $D\Check{A}^{\orth}$ denote the (normalized) Lebesgue measure
on $\Check{\cA}^{\orth}(K)$.
According to Eq. \eqref{eq_SCS_expl_disc} the complex measure
 \begin{equation} \exp(iS^{disc}_{CS}(\Check{A}^{\orth},B))  D\Check{A}^{\orth}
\end{equation}
is a centered oscillatory Gauss type measure on $\Check{\cA}^{\orth}(K)$ in the sense of
Definition \ref{def3.1} in Sec. \ref{sec3} below.

 \medskip

\noindent ii) Let  $DA^{\orth}_c$  denote the (normalized) Lebesgue measure on $\cA^{\orth}_c(K)$
and  $DB$  the (normalized) Lebesgue measure on $\cB(q\cK)$.\par
According to Eq. \eqref{eq_SCS_expl_disc} above, the complex measure
\begin{equation}  \exp(i  S^{disc}_{CS}(A^{\orth}_c,B))    (DA^{\orth}_c \otimes  DB)
\end{equation}
is a  centered  oscillatory Gauss type measure
on $\cA^{\orth}_c(K) \oplus \cB(q\cK)$ in the sense of Definition \ref{def3.1} in Sec. \ref{sec3} below.

\subsection{Definition of $\WLO^{disc}_{rig}(L)$ and $\WLO_{rig}(L)$}
 \label{subsec4.9}

 For the rest of this
paper we will fix a simplicial ribbon link $L= (R_1, R_2, \ldots, R_m)$ in
$q\cK \times \bZ_{N}$ with ``colors''
$(\rho_1,\rho_2,\ldots,\rho_m)$, $m \in \bN$.

\smallskip

Using the definitions of the previous
subsections  we then arrive at the following
 simplicial   analogue $\WLO^{disc}_{rig}(L)$
of the heuristic expression $\WLO(L)$ in Eq. \eqref{eq2.48}
\begin{multline} \label{eq_def_WLOdisc}
\WLO^{disc}_{rig}(L)  :=  \lim_{s \to 0}
  \sum_{y \in I}\int\nolimits_{\sim} \biggl\{ \bigl( \prod_{x} 1^{(s)}_{\ct_{reg}}(B(x))
  \bigr) \Det^{disc}_{FP}(B)\\
\times \biggl[\int\nolimits_{\sim} \biggl(  \prod_{i=1}^m  \Tr_{\rho_i}\bigl( \Hol^{disc}_{R_i}(\Check{A}^{\orth} +  A^{\orth}_c, B)\bigr)  \biggr) \exp(iS^{disc}_{CS}(\Check{A}^{\orth}, B))
D\Check{A}^{\orth} \biggr] \\
 \times   \exp\bigl( - 2 \pi i k  \langle y, B(\sigma_0) \rangle \bigr)  \biggr\}
  \exp(i  S^{disc}_{CS}(A^{\orth}_c,B))    (DA^{\orth}_c \otimes  DB)
\end{multline}
where we use the notation $\int\nolimits_{\sim}  \cdots$ as defined in Definition \ref{def3.2}
in Sec. \ref{sec3} below
and where $\sigma_0$ is an arbitrary fixed point of $\face_0(q\cK)$
which does not lie in   $\bigcup_{i \le m}  \Image(R^i_{\Sigma})$.
Here  $R^i_{\Sigma}$ is the (``reduced'') $\Sigma$-projection of the closed simplicial ribbon $R_i$
(cf. Sec. 4.4.4 in \cite{Ha7a})
and we consider each $R^i_{\Sigma}$ as a map $[0,1] \times S^1 \to \Sigma$,
  cf.  Remark 4.3 in \cite{Ha7a}.

\medskip

Finally, we set\footnote{at this stage we do not yet claim that $\WLO^{disc}_{rig}(L)$ and $\WLO_{rig}(L)$
are actually well-defined}
\begin{equation} \label{eq4.42}
\WLO_{rig}(L) :=    \frac{\WLO^{disc}_{rig}(L)}{\WLO^{disc}_{rig}(\emptyset)}
\end{equation}
where   $\emptyset$ is the ``empty'' link\footnote{so
$\WLO^{disc}_{rig}(\emptyset)$ is a discrete analogue of the partition function
$Z(\Sigma \times S^1)$.
Explicitly,  $\WLO^{disc}_{rig}(\emptyset)$
is given by the expression which  we get from the RHS  of Eq. \eqref{eq_def_WLOdisc} after omitting the product
  $\prod_{i=1}^m  \Tr_{\rho_i}\bigl( \Hol^{disc}_{R_i}(\Check{A}^{\orth}+ A^{\orth}_c,   B)\bigr)$}.

\subsection{Two modifications}
\label{subsec4.10a}

As we discovered in \cite{Ha7a} we have to modify our original approach
if we want to obtain the correct values for $\WLO_{rig}(L)$.
 In order to do so we will now
make two modifications (Mod1) and (Mod2).
More precisely, we will redefine
$\WLO^{disc}_{rig}(L)$ according to the modifications (Mod1) and (Mod2).
  $\WLO_{rig}(L)$ will again be given by Eq. \eqref{eq4.42} (with the redefined version of $\WLO^{disc}_{rig}(L)$
  appearing on the RHS).

\subsubsection*{Modification (Mod1)}

 Let us now reconsider the question
of what a suitable  discrete analogue
$\Det^{disc}_{FP}(B)$ of the continuum expression
 $\Det_{FP}(B) = \det\bigl(1_{{\ck}}-\exp(\ad(B))_{| {\ck}}\bigr)$
 should be.
Above we made the ansatz
\begin{equation} \label{eq_ Det_disc_FP0} \Det^{disc}_{FP}(B) =
 \prod_{x \in \face_0(q\cK)}  \det\bigl(1_{{\ck}}-\exp(\ad(B(x)))_{| {\ck}}\bigr)
\end{equation}
We will now modify Eq. \eqref{eq_ Det_disc_FP0} and make instead the ansatz
\begin{equation} \label{eq_ Det_disc_FP1}
\Det^{disc}_{FP}(B) :=
 \prod_{x \in \face_0(q\cK)}  \det\nolimits^{1/2}\bigl(1_{{\ck}}-\exp(\ad(B(x)))_{| {\ck}}\bigr)
\end{equation}
where $\det\nolimits^{1/2}\bigl(1_{{\ck}}-\exp(\ad(\cdot))_{| {\ck}}\bigr): \ct \to \bR$
is any of the two {\it smooth} functions $f:\ct \to \bR$ fulfilling
$\forall b \in \ct: f(b)^2 = \det\bigl(1_{{\ck}}-\exp(\ad(b))_{| {\ck}}\bigr)$, and
which  is given explicitly by\footnote{here the sign  $+$ or $-$ on the RHS is of course independent of $b \in \ct$}
 \begin{equation} \label{eq_def_det1/2}
 \forall b \in \ct: \quad \det\nolimits^{1/2}\bigl(1_{{\ck}}-\exp(\ad({b}))_{|{\ck}}\bigr) =
 \pm  \prod_{{\alpha} \in {\cR_+}}  \bigl( 2  \sin( \pi  \langle \alpha, b \rangle) \bigr)
\end{equation}
 where $\cR_+$ on the RHS is the set of positive real roots associated to $(\cG,\ct)$ and the Weyl chamber
 $\CW$ fixed above. Without loss of generality we will assume that
the $+$ sign on the RHS of Eq. \eqref{eq_def_det1/2} holds.

\begin{remark} \label{rm_det_exponent} \rm
The inclusion of the exponent $1/2$ on the RHS of Eq. \eqref{eq_ Det_disc_FP1} is
necessary if we want to obtain the correct values for the WLOs.
 We will see later (cf. Appendix \ref{appJ.2} below) that --
 after making the transition to
the $BF_3$-theoretic setting as  explained in Sec. \ref{sec7} --
it may be possible to obtain a better understanding of the origin of this  exponent $1/2$. \par

Alternatively, one can simply bypass this point by rewriting the heuristic equation \eqref{eq2.48}
in a suitable way before discretizing it, cf. Sec. 2.5 and Sec. 3.6 in \cite{Ha9}.

\end{remark}

\subsubsection*{Modification (Mod2)}

In the following we will replace\footnote{so, in particular, $DB$ will denote the normalized Lebesgue
measure on $\cB_0(q\cK)$ where we have equipped $\cB_0(q\cK)$ with the scalar product induced
by the one on $\cB(q\cK)$} the space $\cB(q\cK) = C^0(q\cK,\ct)$  appearing on the RHS of Eq. \eqref{eq_def_WLOdisc}
by a certain subspace $\cB_0(q\cK)$ (chosen as naturally
as possible) with the property that
\begin{equation} \label{eq_obs1}
\ker(\pi \circ (d_{q\cK})_{| \cB_0(q\cK)}) = \cB_{c}(q\cK)
\end{equation}
holds where $\pi: C^1(q\cK,\ct) \to C^1(K,\ct)$ is the orthogonal projection w.r.t.
$\ll \cdot, \cdot \gg_{\cA_{\Sigma}(q\cK)}$ and where we have set
\begin{equation}
\cB_{c}(q\cK):= \{ B \in C^0(q\cK,\ct) \mid B \text{ constant}\}
\end{equation}
We remark that Eq. \eqref{eq_obs1} will play a crucial role in the proof of Theorem \ref{main_theorem}.
(We also remark that we cannot choose simply $\cB_0(q\cK) =  \cB(q\cK)$
because $\ker(\pi \circ d_{q\cK}) \neq \cB_{c}(q\cK)$).

\smallskip

The following choice is probably the best  when working with  simplicial ribbons in $\cK \times \bZ_N$ instead of
simplicial ribbons in $q\cK \times \bZ_N$ (cf. Remark \ref{rm_full_ribbons} above, and
Remark \ref{rm_full_ribbons2} and Sec. \ref{sec7} below).

\begin{choice} \label{example3} \rm
$\cB_0(q\cK):= \psi(\cB(\cK))$ where $\cB(\cK):= C^0(\cK,\ct)$ and where $\psi: \cB(\cK) \to \cB(q\cK)$
is the linear injection which associates to each $B \in  \cB(\cK)$ the extension $\bar{B} \in \cB(q\cK)$
given by
$$\bar{B}(x) = \mean_{y \in C(x)} B(y) \quad \quad \text{for all $x \in \face_0(q\cK)$} $$
with  ``mean'' referring to the arithmetic mean.
Above $C(x)$ denotes the set of all $y \in \face_0(\cK)$ which
 lie in the closure of the unique open cell of $\cK$  containing $x$.
\end{choice}

In the main part of the present paper
we will work with simplicial ribbons in   $q\cK \times \bZ_N$. In this case Choice \ref{example3} will not work
and we will   make  the following choice (motivated by Remark \ref{rm2.2} above):

\begin{choice} \label{example1} \rm

$\cB_0(q\cK) :=  \cB^{loc}_{\sigma_0}(q\cK) \cap \cB_{\aff}(q\cK)$
 with
\begin{align*} \cB^{loc}_{\sigma_0}(q\cK) & := \{ B \in \cB(q\cK) \mid B \text{ is constant on $U(\sigma_0) \cap \face_0(q\cK)$}  \}, \quad \\
 \cB_{\aff}(q\cK) & := \{ B \in \cB(q\cK) \mid B \text{ is affine on each $F \in \face_2(q\cK)$} \}
 \end{align*}
Here $  U(\sigma_0) \subset \Sigma$ is the union of those few $F \in \face_2(q\cK)$
which contain the point $\sigma_0$  and by ``$B$ is affine on $F$'' we mean that
\begin{equation} \label{eq_for_aff_B_0} B(p_1)+B(p_4) = B(p_2)+B(p_3)
\end{equation}
 holds where $p_1, p_2, p_3,p_4$ are the four vertices
of $F$ and numbered in such a way that $p_1$ is diagonal to $p_4$ and therefore
$p_2$ is diagonal to $p_3$.

\end{choice}

\begin{remark} \label{conv_newBqK} \rm
\begin{enumerate}
\item It is not difficult  to verify that $\cB_0(q\cK) =  \cB^{loc}_{\sigma_0}(q\cK) \cap \cB_{\aff}(q\cK)$ indeed satisfies Eq. \eqref{eq_obs1} above.

\item  Above we said that  we replace the space $\cB(q\cK)$ in Eq. \eqref{eq_def_WLOdisc} above by the subspace $\cB_0(q\cK)$.  In fact, in Sec. \ref{sec5} below many statements will hold
 (and many definitions make sense) for all $B \in \cB(q\cK)$ and not only for $B \in \cB_0(q\cK)$.
 Only in Sec. \ref{subsec5.2}  the use of the space $\cB_0(q\cK)$  will be essential.
 In the other parts of Sec. \ref{sec5} we can and will work with the full space $ \cB(q\cK)$.
\end{enumerate}
\end{remark}

For later use we will also define the space
$$  C^0_{\aff}(q\cK,\bR)  := \{ f \in C^0(q\cK,\bR) \mid f \text{ is affine on each $F \in \face_2(q\cK)$} \}$$
where ``$f$ is affine on $F$'' we mean that
\begin{equation} \label{eq_for_aff_B} f(p_1)+f(p_4) = f(p_2)+f(p_3)
\end{equation}
where $p_1, p_2, p_3,p_4$ are the four vertices
of $F$  numbered again as above.

\subsection{The main result}
\label{subsec4.10}

Recall that in  Sec. \ref{subsec4.9} above  we fixed a
 simplicial ribbon link $L= (R_1, R_2, \ldots, R_m)$
 in $q\cK \times \bZ_N$ with colors $(\rho_1, \rho_2, \ldots, \rho_m)$.
Let $\Lambda_+$ denote the set of dominant real weights  associated to the pair $(\cG,\ct)$
 and the Weyl chamber $\CW$.
For each $i \le m$ we denote by $\lambda_i \in \Lambda_+$ the highest weight of $\rho_i$.

\smallskip

We will now restrict ourselves to the special situation
where $L$ fulfills the following two conditions:
   \begin{description}
   \item[(NCP)'] The maps $R^i_{\Sigma}$ neither intersect each other nor themselves\footnote{
   i.e. these maps have pairwise disjoint
   images and each $R^i_{\Sigma}$, considered
   as a continuous map $[0,1] \times S^1 \to \Sigma$, is an embedding}
    \item[(NH)'] Each of the maps $R^i_{\Sigma}$, $i \le m$ is null-homotopic.
  \end{description}
   Here  $R^i_{\Sigma}$ is as in  Sec. \ref{subsec4.9} above
 and where we consider each $R^i_{\Sigma}$ as a map $[0,1] \times S^1 \to \Sigma$,
  cf.  Remark 4.3 in \cite{Ha7a}.

  \medskip

Recall that in \cite{Ha7a} we stated the following theorem which will be proven in Sec. \ref{sec5} below
(and recall also the comments we made in Remark 6.5, Remark 6.6, and Remark 6.7 in \cite{Ha7a})

\begin{theorem} \label{main_theorem}
Assume that the (colored) simplicial ribbon link $L= (R_1, R_2, \ldots, R_m)$ in $q\cK \times \bZ_N$ fixed in Sec. \ref{subsec4.9} above fulfills conditions (NCP)' and (NH)' above.
 Assume also that\footnote{the situation  $0 < k < \cg$ is not interesting since in this case
the set $\Lambda^k_+$ is empty, cf. Remark \ref{rm_app0} below.
 Accordingly,  $|L| = |\emptyset| = 0$. It turns out that we then also have $\WLO^{disc}_{rig}(L) = \WLO^{disc}_{rig}(\emptyset) = 0$} $k \ge \cg$ where $\cg$ is the dual Coxeter number of $\cG$
 and that $\lambda_i \in \Lambda^k_+$, $i \le m$, where $\Lambda^k_+$ is as in   part \ref{appB} of the Appendix below. Then  $\WLO_{rig}(L)$  is well-defined and we have
 \begin{equation} \label{eq_maintheorem}
 \WLO_{rig}(L) = \frac{|L|}{|\emptyset|}
 \end{equation}
 where $\emptyset$ is the ``empty link'' and
  where $|\cdot|$  is the  shadow invariant associated
 to  $\cG$ and $k$, cf. part \ref{appA} of the Appendix.
 \end{theorem}

\begin{remark} \label{rm_full_ribbons2} \rm As mentioned in Remark \ref{rm_full_ribbons} above,
instead of working with simplicial ribbons in $q\cK \times \bZ_N$ (``half ribbons'')
   one could try to work with simplicial ribbons in
     $\cK \times \bZ_N$ (``full ribbons'').
   This    would have several important  advantages, cf. Choice \ref{example3} above,
      Remark \ref{rm_Step5_full_ribbon}
   in Sec. \ref{subsec5.5} below, and   Remark \ref{rm_sec6_full_ribbon}
   in Sec. \ref{sec6} below.
 On the other hand the use of ``full ribbons''  instead of ``half ribbons''
   would also have an important disadvantage, cf. again Remark \ref{rm_Step5_full_ribbon}.
    This is why in Theorem \ref{main_theorem}
    we only consider the case of half ribbons. We will come back to the case of full ribbons in Sec. \ref{sec7}
    below and in \cite{Ha9}.
\end{remark}

\section{Oscillatory Gauss-type measures on Euclidean spaces}
\label{sec3}

In the present section we will  recall two  definitions introduced in Sec. 5 in \cite{Ha7a}
and then derive several elementary results which will play an important role
in the proof of Theorem \ref{main_theorem}.

\subsection{Basic Definitions}
\label{subsec3.1}

Let us fix a  Euclidean vector space  $(V, \langle \cdot, \cdot \rangle)$
and set $d:= \dim(V)$.

\begin{definition} \label{def3.1}
  An ``oscillatory
 Gauss-type measure'' on  $(V, \langle \cdot, \cdot \rangle)$
 is a  complex Borel measure $d\mu$ on $V$
 of the form
 \begin{equation} \label{eq3.1}
 d\mu(x) = \tfrac{1}{Z} e^{ - \tfrac{i}{2} \langle x - m, S (x-m) \rangle} dx
\end{equation}
with $Z \in \bC \backslash \{0\}$,
   $m \in V$, and where  $S$ is a  symmetric endomorphism of $V$  and
 $dx$  the normalized\footnote{i.e. unit hyper-cubes have volume $1$ w.r.t. $dx$}
 Lebesgue measure on $V$.
 Note that $Z$, $m$ and $S$ are uniquely determined by $d\mu$
 so we can use the notation $Z_{\mu}$, $m_{\mu}$ and $S_{\mu}$
 in order to denote these objects.
\begin{enumerate}
\item We call $d\mu$ ``centered''iff $m=0$.

\item  We call $d\mu$ ``degenerate'' iff $S$ is not invertible

\item We call $d\mu$ ``normalized'' iff
$Z=  \frac{(2 \pi)^{d/2}}{ \det^{\frac{1}{2}}(i  S')}$ where $S' :=
S_{| \ker(S)^{\orth}}$. (See Example \ref{obs1} below for the
definition
 of  $\det^{\frac{1}{2}}(i S)$ and a  motivation
for the term ``normalized'').
 \end{enumerate}
 \end{definition}

\begin{definition} \label{def3.2} Let  $d\mu$  be an oscillatory
 Gauss-type measure on  $(V, \langle \cdot, \cdot \rangle)$.
 A (Borel) measurable function
  $f: V \to \bC$ will be called improperly integrable w.r.t. $d\mu$
  iff\footnote{Observe that
$\int_{\ker(S_{\mu})}  e^{- \eps \|x\|^2} dx =
(\tfrac{\eps}{\pi})^{-n/2}$. In particular, the factor
$(\tfrac{\eps}{\pi})^{n/2} $ in Eq. \eqref{eq3.2} above  ensures
that also for degenerate oscillatory
 Gauss-type measure the improper integrals  $\int\nolimits_{\sim} 1 \ d\mu$  exists, cf.
 Example \ref{obs1} below}
 \begin{equation}\label{eq3.2} \int\nolimits_{\sim} f d\mu := \int\nolimits_{\sim} f(x)
   d\mu(x): =
    \lim_{\eps \to 0} (\tfrac{\eps}{\pi})^{n/2} \int f(x) e^{- \eps |x|^2} d\mu(x)
  \end{equation}
  exists. Here  we have set  $n:=\dim(\ker(S_{\mu}))$.
   Note that if $d\mu$ is non-degenerate we have $n=0$ so the factor $(\tfrac{\eps}{\pi})^{n/2}$
is then trivial.
 \end{definition}

Most of the time we will consider non-degenerate
 oscillatory Gauss-type measures, the exception being
 Proposition \ref{prop3.5} below. \par
Using a simple analytic continuation argument and the corresponding
explicit formulas for Gaussian probability measures   we can easily
 prove the existence
of $\int\nolimits_{\sim} f \ d\mu$ and compute the corresponding
value explicitly for a large class of functions $f$. Let us
illustrate this by looking at some simple examples:

\begin{example} \label{ex1} Consider the special case where $V=\bR$, where $\langle \cdot, \cdot \rangle$ is
 the  scalar product given by $\langle x, y \rangle = x  y$, and where
 $d\mu(x) =  \exp(i \langle x,x \rangle)dx = \exp(i x^2)dx$. Then the improper integrals
 $$ \int\nolimits_{\sim} 1 \ d\mu(x), \quad  \int\nolimits_{\sim} x \ d\mu(x), \quad
 \int\nolimits_{\sim} x^2 \ d\mu(x), \quad  \int\nolimits_{\sim} e^{cx} \ d\mu(x), \ c \in \bC$$
 exist and are given explicitly by
 \begin{itemize}
 \item  $\int\nolimits_{\sim} 1 \ d\mu(x) =  \sqrt{i \pi }  = \sqrt{\pi} \ e^{\frac{\pi}{4} i}$
 \item $\int\nolimits_{\sim} x \ d\mu(x) =  0$
 \item $\int\nolimits_{\sim} x^2 \ d\mu(x) =  \tfrac{i}{2} \ \sqrt{i \pi }  $
 \item $\int\nolimits_{\sim} e^{cx} \ d\mu(x) = e^{i \frac{c^2}{4}} \ \sqrt{i \pi }$
  \end{itemize}
 where $\sqrt{.}: \bC \backslash (-\infty,0) \to
\bC$ denotes the standard square root.

\smallskip

In order to show the existence  (and to compute the explicit value)
 of $ \int\nolimits_{\sim} 1 \ d\mu(x)$  we consider  the analytic function $F: \{z \mid
Re(z) > 0\} \to \bC$ given by $F(z):= \int \exp( - z x^2) dx$. According to a well-known formula
 we have $F(a) = \sqrt{\pi/a}$ for all $a \in (0,\infty)$. The
obvious uniqueness argument for analytic functions now implies that
$F(z) = \sqrt{\pi/z}$ for all $z \in \bC$
with $Re(z) > 0$.
 Thus $\int\nolimits_{\sim} 1 \ d\mu =  \lim_{\eps
\to 0} F(\eps -i ) = \lim_{\eps \to 0} \sqrt{\pi/(\eps - i)} =  \sqrt{i \pi } $.

\smallskip

The other three integrals can be dealt with in a similar way.
\end{example}

\noindent In the  next example  $(V, \langle \cdot, \cdot \rangle)$ is again an  arbitrary Euclidean space.

\begin{example} \label{obs1}
  Let $d\mu$ be a non-degenerate oscillatory
 Gauss-type measure on $(V, \langle \cdot, \cdot \rangle)$
  with  $S$, $m$, and $Z$ given as in Eq. \eqref{eq3.1}
\begin{enumerate}
\item We have\footnote{we remark that if $d\mu$ is degenerate then an analogous
statement will hold with $S$  replaced by $S':= S_{| \ker(S)^{\orth}}$}

 \begin{equation} \label{eq3.3} \int\nolimits_{\sim} 1 \ d\mu  = \frac{1}{Z}
 \frac{(2 \pi)^{d/2}}{\det^{\frac{1}{2}}(i S)}
 \end{equation}
where we have set
$\det^{\frac{1}{2}}(i S):= \prod_k \sqrt{ i
\lambda_k} =  e^{\frac{\pi i}{4} \sum_k \sgn(\lambda_k)} (\prod_k |\lambda_k|^{1/2})$
 where $(\lambda_k)_k$ are the (real) eigenvalues of the symmetric matrix
$S$.
In particular, $d\mu$ is
normalized in the sense of Definition \ref{def3.1} above iff
$\int\nolimits_{\sim} 1 \ d\mu = 1$.

\item  In the special case when $d\mu$ is normalized  we have for all $v, w \in V$
\begin{equation} \label{eq3.4}
 \int\nolimits_{\sim} \langle v, x \rangle \ d\mu(x)  =   \langle v, m \rangle  , \quad \quad
 \int\nolimits_{\sim} \langle v, x \rangle \langle w, x \rangle \ d\mu(x)
  = \tfrac{1}{i} \langle v, S^{-1}  w \rangle +  \langle v, m \rangle \langle w, m \rangle
\end{equation}
\end{enumerate}

\end{example}

We will not try to identify the largest possible class of functions
$f$ for which $\int\nolimits_{\sim} f \ d\mu$ exists.
 For our purposes the function algebra $\cP_{exp}(V)$ defined in the next
definition  will be sufficient.

\begin{definition} \label{def3.3} \begin{enumerate}
\item Let $W$ be a  finite-dimensional associative $\bR$-algebra (with the standard topology).
By $\cP_{exp}(V,W)$ we  will denote  the subalgebra of $\Map(V,W)$
which is generated by the affine
maps $\varphi:V \to W$ and their ``exponentials''
$\exp_W \circ \varphi$. Here $\exp_W:W \to W$ denotes the exponential map of $W$.

\item By $\cP_{exp}(V)$ we denote the subalgebra
 of $\Map(V,\bC)$
which is generated by the functions of the form $\theta \circ f$
with $f \in  \cP_{exp}(V,W)$ and $\theta \in \Hom_{\bR}(W,\bC)$
where $W$ is any\footnote{note that we do not keep $W$ fixed here}
  finite-dimensional associative $\bR$-algebra.
   \end{enumerate}
\end{definition}

Using  analytic continuation arguments, some explicit formulas for Gaussian probability measures,
and suitable growth estimates one can prove the  following result (which is easy to believe):
\begin{proposition} \label{prop3.1} Let $d\mu$ be a non-degenerate oscillatory
 Gauss-type measure on  $(V, \langle \cdot, \cdot \rangle)$.
Then for every  $f \in  \cP_{exp}(V)$ the  improper integral  $\int\nolimits_{\sim} f \ d\mu \in \bC$
exists.
\end{proposition}

\subsection{Three propositions}
\label{subsec3.2}

 Let us fix for a while
a  normalized non-degenerate oscillatory Gauss-type measure $d\mu$ on
 $(V, \langle \cdot,\cdot \rangle)$
 and  introduce  the  notation
 \begin{subequations} \label{eq3.5}
\begin{align} \bE_{\sim}[X] & := \int\nolimits_{\sim} X d\mu \in \bC \\
 \cov_{\sim}(X,X') & := \bE_{\sim}[X X'] - \bE_{\sim}[X] \bE_{\sim}[X'] \in \bC
 \end{align}
 \end{subequations}
 for   maps $X,X' \in \cP_{exp}(V)$
 (in  analogy to the case of (Gaussian or non-Gaussian) probability measures on $V$).

\begin{observation} \label{prop3.2}  \label{prop_obs2}
 Let $d\mu$  be as above
 and let $X_1, X_2, \ldots, X_n$ be a sequence  of  affine maps $V \to \bR$.
\begin{enumerate}
\item If $\bE_{\sim}[X_i]=0$ for every $i \le n$ then
\begin{equation} \label{eq3.6}
 \bE_{\sim}[ \prod_j X_j ]  =  \begin{cases}  \frac{1}{(n/2)! 2^{n/2}} \sum_{\sigma \in S_{n}}
  \prod_{i=1}^{n/2} \cov_{\sim}(X_{\sigma(2i-1)}, X_{\sigma(2i)}) & \text{ if $n$ is even }\\
  0 & \text{ if $n$ is odd } \end{cases}
\end{equation}
\item If   $\cov_{\sim}(X_i,X_j) = 0$ for all $i, j \le n$ with $i \neq j$
 we have
\begin{equation} \label{eq3.7}  \bE_{\sim}[ \prod_j X_j ] =   \prod_j \bE_{\sim}[  X_j ]
\end{equation}
\end{enumerate}
Here Eq. \eqref{eq3.6} follows from the analogous formula for the
moments of a Gaussian probability measure and a suitable analytic
continuation argument. Clearly, in the special case where
$\cov_{\sim}(X_i,X_j) = 0$ for $i, j \le n$ with $i \neq j$ Eq.
\eqref{eq3.6} reduces to $\bE_{\sim}[ \prod_j X_j ] = 0$. By
applying the latter equation to the subsequences of the sequence
$X'_1, X'_2, \ldots, X'_n$  given by  $X'_j:= X_j - \bE[X_j]$ we  arrive at Eq. \eqref{eq3.7}.
\end{observation}

 Let us now consider the case where
  $Y:V \to \bR$ is an affine map with\footnote{Note that
  the condition $\cov_{\sim}(Y,Y) = 0$ does not imply that $Y$ is a constant   map on $V$.
 This is in sharp contrast to the situation
 for (Gaussian or non-Gaussian) probability measures
 where the  relation $\cov(Y,Y) = 0$  always implies $Y= \bE[Y]$ $d\mu$-a.s.}
  $\cov_{\sim}(Y,Y) = 0$
and let us consider the trivial sequence $(X_i)_{i=1}^n$ where $X_i
= Y$ for each $i \le n$. Since $\cov(Y,Y)_{\sim} = 0$ we trivially
have $\cov_{\sim}(X_i,X_j) = 0$ for  $i \neq j$ (and  even for
$i=j$)  so according to Eq. \eqref{eq3.7}  we have $ \bE_{\sim}[ Y^n
] = \bE_{\sim}[Y]^n$,
 from which we conclude, for example, that
\begin{equation} \label{eq3.8}  \bE_{\sim}[ \exp(Y)] = \bE_{\sim}[ \sum_n  \tfrac{Y^n}{ n !}]
\overset{(*)}{=} \sum_n  \bE_{\sim}[\tfrac{Y^n}{ n !}] =   \sum_n
\frac{\bE_{\sim}[Y]^n}{ n !} = \exp(\bE_{\sim}[Y])
\end{equation}
In order to see that  step $(*)$ above holds observe that for every fixed $\eps > 0$ we have
$$ \int  \sum_n \biggl| \tfrac{Y(x)^n}{ n !}  e^{ - \eps |x|^2} \tfrac{1}{Z}
 e^{ - \tfrac{i}{2} \langle x - m, S (x-m) \rangle}\biggr| dx \le
\int  \sum_n  \tfrac{|Y(x)|^n}{ n !}  \tfrac{1}{Z} e^{ - \eps |x|^2}
 dx  < \infty$$ and we can therefore conclude  from the
dominated convergence theorem\footnote{applied to the positive measure $dx$ ``appearing'' in $d\mu$}  that
 $$\int  \sum_n  \tfrac{Y(x)^n}{ n
!} e^{ - \eps |x|^2}  d\mu(x) = \sum_n   \int  \tfrac{Y(x)^n}{ n !}
e^{ - \eps |x|^2}  d\mu(x)$$
 for each $\eps > 0$. Accordingly, in order to prove step $(*)$ it is enough to prove that
$\lim_{\eps \to 0}  \sum_n  \tfrac{ I(n,\eps)}{n!} =  \sum_n  \lim_{\eps \to 0} \tfrac{ I(n,\eps)}{n!}$
 where we have set $I(n,\eps):=\int Y(x)^n e^{ - \eps |x|^2}  d\mu(x)$.
The latter claim can easily be proven by computing the integrals  $I(n,\eps)$ explicitly.\par
 More generally, we obtain\footnote{Observe
that $\Phi(Y) \in \cP_{exp}(V)$ so the
 existence of the LHS of Eq. \eqref{eq3.9} is guaranteed
by Proposition \ref{prop3.1}}  for every  $\Phi \in \cP_{exp}(\bR)$
\begin{equation} \label{eq3.9}\bE_{\sim}[\Phi(Y)] =\Phi(\bE_{\sim}[Y])\end{equation}
 since every such $\Phi$ is necessarily entire analytic
and the coefficients $(c_n)_n$, given by $\Phi(x)=
\sum_{n=0}^{\infty} c_n x^n$ for all $x \in \bR$, have the property
that $c_n \overset{n \to \infty}{\longrightarrow} 0$ rapidly enough
so that
\begin{itemize}
\item[1.] we can again apply the dominated convergence theorem in a
similar way as above  and prove that the two limit procedures $\int \cdots dx$ and $\sum_n $
  can be interchanged
\item[2.] we can  prove again that the two limit procedures $\lim_{\eps \to 0}$ and   $\sum_n $ can be interchanged
\end{itemize}
Finally, we can  generalize Eq. \eqref{eq3.9} to the case where we have a  $\Phi \in \cP_{exp}(\bR^n)$, $n
\in \bN$, and where $(Y_k)_{k \le n}$ is  a sequence of affine maps
$V \to \bR$ such that $\cov_{\sim}(Y_i,Y_j) = 0$ holds for all $i,j
\le n$. We then arrive at the following result,  which will be the
key argument in Sec. \ref{subsec5.1} below.  In order to make the
application of Proposition \ref{prop3.3} in Sec.  \ref{subsec5.1}
 more transparent we avoid the  use of
 the notation  $\bE_{\sim}[\cdot]$ and $\cov_{\sim}(\cdot,\cdot)$
 from now on.
\begin{proposition} \label{prop3.3} Let $d\mu$ be a normalized
 non-degenerate oscillatory  Gauss-type measure on  $(V, \langle \cdot, \cdot \rangle)$
 and let $(Y_k)_{k \le n}$, $n \in \bN$, be  a sequence of affine maps $V \to \bR$ such that
\begin{equation} \label{eq3.10}  \int\nolimits_{\sim} Y_i Y_j d\mu  = \bigl( \int\nolimits_{\sim} Y_i d\mu  \bigr) \bigl(
\int\nolimits_{\sim} Y_j d\mu \bigr)
\end{equation}
 holds for all   $i, j \le n$.
Then for every  $\Phi \in \cP_{exp}(\bR^n)$   we have
\begin{equation} \label{eq3.11} \int\nolimits_{\sim}  \Phi((Y_k)_{k}) d\mu
 =\Phi\bigl(\bigl(  \int\nolimits_{\sim}  Y_k  d\mu \bigr)_{k} \bigr)
\end{equation}
\end{proposition}

\begin{remark} \label{rm3.1}  \rm In Sec.  \ref{subsec5.1} below we will actually apply a reformulation
of Proposition \ref{prop3.3} where the  sequence $(Y_k)_{k \le n}$,
$n \in \bN$, is replaced by a family $(Y^{i,a}_k)_{k \le n, i \le m,
a \le D}$ of affine maps fulfilling the obvious analogue of Eq.
\eqref{eq3.10} above and the function  $\Phi \in \cP_{exp}(\bR^n)$
is replaced by a function $\Phi \in \cP_{exp}(\bR^{m \times n \times
D})$.
\end{remark}

\begin{example} \label{ex2}
Consider the (non-degenerate normalized  centered) oscillatory
 Gauss-type measure
$d\mu(x):= \frac{1}{2 \pi} \exp(i \langle x_1,x_2 \rangle ) dx_1 dx_2$
on $V=\bR^2$. For every $f \in \cP_{exp}(\bR)$ we have
\begin{equation} \label{eq3.12}
\int_{\sim} f(x_1) d\mu(x) =  f(0)
\end{equation}
This follows by applying  Proposition \ref{prop3.3} with $\Phi = f$
and $Y_1(x):=x_1$. Observe that  Eq. \eqref{eq3.10} is indeed
fulfilled since according to Example \ref{obs1} we have
$\int\nolimits_{\sim} Y_1 d\mu =  \int\nolimits_{\sim} \langle x,e_1
\rangle  d\mu(x) = 0$ and $\int\nolimits_{\sim} Y_1 Y_1 d\mu =
\tfrac{1}{i} \langle e_1, S_{\mu}^{-1} e_1 \rangle = 0$ where we
have set
 $e_1:= \biggl(\begin{matrix} 1 \\ 0 \end{matrix}\biggr)$
 and used that $S_{\mu}= - \biggl(\begin{matrix} 0 && 1 \\ 1 && 0
\end{matrix}\biggr)$.
\end{example}

Using  a different argument one can easily prove that Eq.
\eqref{eq3.12}  holds for arbitrary continuous bounded functions
$f$. Moreover,  we can  include an additional ``exponential
factor'', and we can in fact consider more general
  oscillatory  Gauss-type measures $d\mu$, cf. Proposition \ref{prop3.5} below,
  which will play a key role in Sec.  \ref{subsec5.2} below.\par

  As a preparation for Proposition \ref{prop3.5} let
  us first consider the special case
  where the oscillatory  Gauss-type measure $d\mu$ is non-degenerate:

\begin{proposition} \label{prop3.4}
Assume that $V=V_1 \oplus V_2$ where $V_1$ and $V_2$
are two isomorphic subspaces of $V$ which are orthogonal to each other.
For each  $j=1,2$ we denote  the $V_j$-component of $x \in V$ by $x_j$.
Moreover, let   $d\mu$ be a  (non-degenerate  centered)
 normalized   oscillatory  Gauss-type measure on  $(V, \langle \cdot, \cdot \rangle)$
of the form $d\mu(x)= \tfrac{1}{Z} \exp(i \langle x_2,M x_1) dx$ for
some linear isomorphism $M: V_1 \to V_2$. Then for every bounded
continuous function $f:V_1 \to \bC$ and every fixed $v \in V_2$ we
have
\begin{equation} \label{eq3.13}
\int_{\sim} f(x_1) \exp(i \langle x_2,v \rangle ) d\mu(x) =  f(-
M^{-1} v)
\end{equation}
(In particular, the LHS of Eq. \eqref{eq3.13} exists; note that
the present situation  is not covered by Proposition \ref{prop3.1}
above).
\end{proposition}

\begin{proof}  Let $dx_1$ resp. $dx_2$ be the normalized Lebesgue measure on $V_1$ resp. $V_2$
(equipped with the scalar product induced by the one on $V$). We have
\begin{align} \label{eq3.14}
 &\int_{\sim} f(x_1) \exp(i \langle x_2,v \rangle ) d\mu(x) \nonumber \\
&   =   \frac{1}{Z} \lim_{\eps \to 0}  \int_{V_1} \int_{V_2}
    e^{- \eps (|x_1|^2 + |x_2|^2)}
f( x_1)  \exp(i \langle x_2,v + M x_1 \rangle ) dx_2 dx_1 \nonumber \\
&  =   \frac{(2 \pi)^{d_2}}{Z} \lim_{\eps \to 0}  \int_{V_1}    e^{- \eps |x_1|^2}
f( x_1)  \delta_{\eps}(v + M x_1) dx_1
\end{align}
where $d_2:= \dim(V_2) = d/2$ and where
  $\delta_{\eps}: V_2 \to \bR$ is given by
\begin{equation} \label{eq3.15}
\delta_{\eps}(w) :=   \tfrac{1}{(2 \pi)^{d_2}} \int_{V_2}  e^{- \eps
|x_2|^2}
 \exp(i \langle x_2, w \rangle ) dx_2
= \tfrac{1}{(2 \pi)^{d_2}}  \bigl( \tfrac{\pi}{\eps}\bigr)^{d_2/2}
 e^{- \tfrac{1}{2} \tfrac{|w|^2}{2\eps}} = \tfrac{1}{(4 \eps \pi)^{d_2/2}}
  e^{- \tfrac{|w|^2}{4\eps}}
\end{equation}
 for all $w \in V_2$.
Let us now fix an (arbitrary) isometry $\psi:V_2 \to V_1$.
Clearly, the pushforward  $(\psi)_* dx_2$ of $dx_2$ coincides with $dx_1$ so we have
\begin{equation} \label{eq3.16_0}
 \int_{V_1}    e^{- \eps |x_1|^2}
f( x_1)  \delta_{\eps}(v + M x_1) dx_1
= \int_{V_2}    e^{- \eps |\psi  x_2|^2}
f( \psi  x_2)  \delta_{\eps}(v + M \psi x_2) dx_2
\end{equation}
Making the change  of variable $ M \psi x_2 + v \to y_2$ on the RHS of Eq. \eqref{eq3.16_0}
we  obtain
\begin{equation} \label{eq3.16}
 \int_{V_1}    e^{- \eps |x_1|^2}
f( x_1)  \delta_{\eps}(v + M x_1) dx_1
= |\det(M\psi)|^{-1} \int_{V_2}    e^{- \eps | M^{-1}(y_2 - v)|^2}
f( M^{-1}(y_2 - v))  \delta_{\eps}(y_2) dy_2
\end{equation}
where  $dy_2$ is the normalized Lebesgue measure on $V_2$.
Since $(\delta_{\eps})_{\eps >0}$ is an ``approximation to the identity'' (i.e. converges
weakly to the Dirac distribution $\delta_0$)
 it is therefore clear\footnote{a formal proof of Eq. \eqref{eq3.17} can be obtained
after a suitable change of variable and the application of the dominated convergence
theorem, cf. the proof of Eq. \eqref{eq3.20} below which generalizes
Eq. \eqref{eq3.17}} that
\begin{equation} \label{eq3.17}
 \lim_{\eps \to 0}  \int_{V_1}    e^{- \eps |x_1|^2}
f( x_1)  \delta_{\eps}(v + M x_1) dx_1 = f(- M^{-1} v) \cdot
|\det(M\psi)|^{-1}
\end{equation}
 The assertion of the proposition
now follows from  Eq. \eqref{eq3.14}, Eq. \eqref{eq3.17}
and the following equation:
\begin{equation} \label{eq3.17_plus1}
\frac{(2 \pi)^{d_2}}{Z} =  \frac{(2 \pi)^{d/2}}{Z} \overset{(*)}{=} \det\nolimits^{\frac{1}{2}}(i S_{\mu})
\overset{(**)}{=}  \det(i S_{\mu})^{1/2}  \overset{(***)}{=} \det((M\psi)^t M\psi )^{1/2} = |\det(M\psi)|
 \end{equation}
 Here step $(*)$ follows from  Example \ref{obs1} and the
 assumption that $d\mu$ is normalized,
 step $(**)$ follows because\footnote{this can be seen, e.g.,
 from the equation  $J^{-1} S_{\mu} J = - S_{\mu}$ where
 $J:=  \left(\begin{matrix} -1 && 0 \\ 0 && 1
\end{matrix}\right)$,  $1$ denoting both  the identity in $\End(V_1)$ and $\End(V_2)$}
for each eigenvalue $\lambda$ of $S_{\mu}$ also $-\lambda$ is an
eigenvalue of $S_{\mu}$ and has the same multiplicity as $\lambda$, and
  step $(***)$ follows because
 after making the identification $V_2 \cong_{\psi} V_1$ we
 have $S_{\mu} = - \left(\begin{matrix} 0 && M \psi \\ (M \psi)^t && 0
\end{matrix}\right)$.
\end{proof}

\begin{convention} \label{conv3.1} \rm
 For a continuous function $f:V_0 \to \bC$ on a $d_0$-dimensional Euclidean space
$V_0$ we set
\begin{equation}
 \int^{\sim}_{V_0} f(x_0) dx_0 := \tfrac{1}{\pi^{d_0/2}} \lim_{\eps \to 0} \eps^{d_0/2}   \int_{V_0} e^{-\eps |x_0|^2} f(x_0) dx_0
\end{equation}
provided that the expression on the RHS  of the previous equation is well-defined.
Here $dx_0$ is the normalized Lebesgue measure on $V_0$.
\end{convention}

\begin{proposition} \label{prop3.5}
Assume that $V= V_0 \oplus V_1 \oplus V_2$ where $V_0$, $V_1$, $V_2$ are pairwise
orthogonal subspaces of $V$. For each $j=0,1,2$ we denote
 the $V_j$-component of $x \in V$ by  $x_j$.
Moreover, let $d\mu$ be a  (centered) normalized  oscillatory  Gauss-type measure on  $(V, \langle \cdot, \cdot \rangle)$ of the form $d\mu(x)=  \tfrac{1}{Z} \exp(i \langle x_2,M x_1) dx$
for some linear isomorphism $M: V_1 \to V_2$.
Then for every fixed $v \in V_2$
and every bounded uniformly  continuous function $F:V_0 \oplus V_1 \to \bC$
  the LHS of the following equation exists  iff the RHS  exists
and in this case we have
\begin{equation} \label{eq3.18}
\int_{\sim} F(x_0 + x_1) \exp(i \langle x_2,v \rangle ) d\mu(x) =  \int^{\sim}_{V_0} F(x_0  - M^{-1} v) dx_0
\end{equation}
where $dx_0$ is the normalized Lebesgue measure on $V_0$.
\end{proposition}

\begin{proof}  We set $d_j := \dim(V_j)$ for $j=0,1,2$.
 Similarly as in Eq. \eqref{eq3.14} and with $\delta_{\eps}: V_2 \to \bR$  as above we obtain
\begin{align} \label{eq3.19}
& \int_{\sim} F(x_0 + x_1) \exp(i \langle x_2,v \rangle ) d\mu(x) \nonumber \\
& \quad \quad = \lim_{\eps \to 0} \bigl( \tfrac{\eps}{\pi}\bigr)^{d_0/2}
 \int_{V_0} \biggr[   \int_{V_1} \biggl[  \int_{V_2}   e^{- \eps ( |x_0|^2 +  |x_1|^2 + |x_2|^2)}
F(x_0 + x_1)   \exp(i \langle x_2,v \rangle ) \times \nonumber \\
 &  \quad \quad  \quad \quad  \quad \quad \quad  \quad \quad  \quad \quad
  \times \tfrac{1}{Z} \exp(i \langle x_2,M x_1) dx_2 \biggr] dx_1 \biggr]
dx_0 \biggr] \nonumber \\
 &  \quad \quad  = \tfrac{1}{\pi^{d_0/2}} \  \lim_{\eps \to 0}  \eps^{d_0/2} \int_{V_0} dx_0 e^{- \eps |x_0|^2}
 \biggl[ \tfrac{(2\pi)^{d_2}}{Z} \int_{V_1}  dx_1  e^{- \eps |x_1|^2}
F(x_0 + x_1)  \delta_{\eps}(v + M x_1) \biggr].
\end{align}
 Let $\psi: V_2 \to V_1$ be  a fixed isometry.
From the assumption that  $d\mu$ was normalized
it follows  -- using the same argument as in Eq. \eqref{eq3.17_plus1} above --
that  $\frac{(2\pi)^{d_2}}{Z} = \det^{\frac{1}{2}}(i (S_{\mu})_{| V_1 \oplus V_2}) = |\det(M\psi)|$.
 Eq. \eqref{eq3.19} above and the equality just mentioned will therefore imply the
  assertion of the proposition
provided that we can show
\begin{equation} \label{eq3.20} \lim_{\eps \to 0} T(\eps) = 0
\end{equation}
where
\begin{equation}  T(\eps)  := \eps^{d_0/2} \int_{V_0} dx_0 e^{- \eps |x_0|^2}
 \biggl[F(x_0 - M^{-1}v) -   |\det(M\psi)| \int_{V_1}  dx_1  e^{- \eps |x_1|^2}
F(x_0 + x_1)  \delta_{\eps}(v + M x_1)  \biggr]
\end{equation}
In order to prove \eqref{eq3.20}
recall that  $\psi_*(dx_2)=dx_1$
and that $1 = \tfrac{1}{\pi^{d_2/2}} \int_{V_2}  e^{- |y_2|^2} dy_2$
so we obtain
\begin{align}  \label{eq3.21} T(\eps) & =  \eps^{d_0/2} \int_{V_0} dx_0 e^{- \eps |x_0|^2}
\biggl[ \tfrac{1}{\pi^{d_2/2}} \int_{V_2} dy_2    e^{- |y_2|^2}  F(x_0 - M^{-1}v)    \nonumber \\
 & \quad \quad \quad \quad \quad \quad \quad \quad \quad \quad
  -  \ |\det(M\psi)| \int_{V_2} dx_2  e^{- \eps |\psi x_2|^2}
F(x_0 + \psi x_2) \tfrac{1}{(4 \eps \pi)^{d_2/2}}
  e^{- \tfrac{|v+ M \psi x_2|^2}{4\eps}}   \biggr] \nonumber \\
 & \overset{(*)}{=} \frac{1}{\pi^{d_2/2}}
 \int_{V_0 \oplus V_2} dy_0 dy_2 e^{- |y_0|^2 - |y_2|^2}  \biggl[F(\tfrac{y_0}{\sqrt{\eps}} - M^{-1}v) - \gamma_{\eps}(y_2)
F(\tfrac{y_0}{\sqrt{\eps}} - M^{-1}v + \sqrt{\eps} 2 M^{-1} y_2)  \biggr]
\end{align}
In Step $(*)$  we have made  the changes of
variable $\sqrt{\eps} x_0 \to y_0 $ and  $\tfrac{1}{
\sqrt{\eps}} \tfrac{1}{2}(v + M \psi x_2) \to y_2$  and  we have set
$\gamma_{\eps}(y_2):=  e^{- \eps |M^{-1}(2 \sqrt{\eps} y_2 -
v)|^2}$. Relation \eqref{eq3.20} now follows by applying the dominated
convergence theorem to the last expression in  Eq. \eqref{eq3.21}  and taking
into account that for all fixed $y_0 \in V_0$ and $y_2 \in V_2$ we have
\begin{align*}
& \lim_{\eps \to 0} \bigl[F(\tfrac{y_0}{\sqrt{\eps}} - M^{-1}v ) - \gamma_{\eps}(y_2)
F(\tfrac{y_0}{\sqrt{\eps}}- M^{-1}v  + \sqrt{\eps} 2 M^{-1} y_2)  \bigr]\\
& \quad  = \lim_{\eps \to 0} \bigl[F(\tfrac{y_0}{\sqrt{\eps}} - M^{-1}v  ) (1- \gamma_{\eps}(y_2)) \bigr]\\
& \quad   \quad \quad   \quad  \quad + \lim_{\eps \to 0}  \gamma_{\eps}(y_2)\bigl[F(\tfrac{y_0}{\sqrt{\eps}} - M^{-1}v  ) -
F(\tfrac{y_0}{\sqrt{\eps}} - M^{-1}v  + \sqrt{\eps} 2M^{-1} y_2)  \bigr]  = 0 + 0 = 0
\end{align*}
since $\lim_{\eps \to 0} \gamma_{\eps}(y_2) = 1$
and, by assumption, $F$ is bounded and
uniformly continuous.
\end{proof}

The following remark will be useful in Sec.  \ref{subsec5.2}
 and Sec.  \ref{subsec5.4} below.

\begin{remark}\label{rm_last_sec3}  \rm
If $\Gamma$ is a lattice in $V$
and  $f:V \to \bC$ a $\Gamma$-periodic continuous function then  $\int^{\sim}_{V} f(x) dx$ exists and we have
 \begin{equation}  \label{eq2_lastrmsec3}
\int^{\sim}_{V} f(x) dx = \frac{1}{vol(Q)} \int_{Q}  f(x) dx
 \end{equation}
 with  $Q :=\{\sum_i x_i e_i \mid 0 \le x_i \le 1 \forall i \le d\}$
 where $(e_i)_{i \le d}$ is  any fixed basis of the lattice $\Gamma$
 and where  $vol(Q)$ denotes the volume of $Q$.
 Observe that Eq. \eqref{eq2_lastrmsec3} implies
 \begin{equation}  \label{eq1_lastrmsec3}  \forall y \in   V: \quad
 \int^{\sim}_{V} f(x) dx =  \int^{\sim}_{V} f(x + y) dx
 \end{equation}
\end{remark}

\section{Proof of Theorem \ref{main_theorem}}
\label{sec5}

Recall that Theorem \ref{main_theorem} states that in the special situation described above
 $\WLO_{rig}(L)$ is well-defined and has the value $|L|/|\emptyset|$.
  In the following we will concentrate on the ``computational half'' of this statement.
  That  $\WLO_{rig}(L)$ is  well-defined in the first place will also become clear during
  the computations\footnote{in order to check well-definedness we should, of course,
reverse the order of our considerations/computations:
 we first check that the expressions appearing in Step 6 are well-defined.
 Based on this we can verify  that also the expressions in Step 5 must be well-defined and so on
 until we arrive at the expressions in Step 1}
   even though we will rarely make explicit statements in this   directions.

\begin{convention} \label{conv_sim} \rm
 In the following $\sim$ will denote equality up to a  multiplicative non-zero constant.  This ``constant'' may depend  on $G$, $N$, $\cK$,  and $k$
 but it will never depend on the (simplicial ribbon) link $L$.
   \end{convention}

\setcounter{subsection}{-1}
\subsection{Some preparations}
\label{subsec5.0}

\subsubsection*{\it a) Computation of $\det\bigl(L^{(N)}(B)\bigr)$}

Recall that in Sec. \ref{sec4} above we used the notation $L^{(N)}(B)$ both for the linear operator
$\cA^{\orth}(K) \to \cA^{\orth}(K)$ and the restriction of this operator
to the invariant subspace  $\Check{\cA}^{\orth}(K)$. From now on the notation
$L^{(N)}(B)$ will always refer to the restricted operator.

\begin{proposition} \label{obs2}  For $B \in \cB(q\cK)$ we have
\begin{equation} \label{eq_determinante}
\det\bigl(L^{(N)}(B)\bigr)   =  N^{d} \prod_{\bar{e} \in  \face_0(K_1 | K_2)}
  \det\bigl(1_{{\ck}}- \exp(\ad(B(\bar{e})))_{| {\ck}}\bigr)^2
\end{equation}
 where $d:= \dim(\cA^{\orth}(K))$.
In particular, if
\begin{equation}B \in \cB_{reg}(q\cK) := \{ B \in  \cB(q\cK) \mid
B(x) \in \ct_{reg} \text{ for all $x \in \face_0(q\cK)$}\}
\end{equation}
 then we have $\det\bigl(L^{(N)}(B)\bigr) \neq 0$.
\end{proposition}

According to  Eqs. \eqref{def_LN}, \eqref{eq_LN_ident1}, \eqref{eq_LN_ident2}
in order to prove Proposition \ref{obs2} it will be enough to prove Lemma \ref{lem_obs1} below.

\smallskip

Recall the definition of the two linear operators
$\hat{L}^{(N)}(b)$ and  $\Check{L}^{(N)}(b)$ on   $\Map(\bZ_N,\cG)$ for fixed $b \in {\ct}$, cf. Eqs. \eqref{eq_def_LOp} above.\par
In the following we will consider the restriction of each of these two operators
to the orthogonal\footnote{w.r.t. the obvious scalar product on $\Map(\bZ_N,\cG)$} complement of its kernel\footnote{for $b \in \ct_{reg}$, which is the case relevant for us,
 we have $\ker( \hat{L}^{(N)}(b))=\ker( \Check{L}^{(N)}(b)) =  \Map_c(\bZ_N,\ct) $
 with  $\Map_c(\bZ_N,\ct)  := \{ f \in \Map(\bZ_N,\cG) \mid \text{ $f$ is a constant function taking values in $\ct$ }\} \cong \ct$. The orthogonal complement
 $\Map'(\bZ_N,\cG)$ of $\Map_c(\bZ_N,\ct)$
is given by $\Map'(\bZ_N,\cG)  = \{ f \in \Map(\bZ_N,\cG) \mid \sum_{t \in \bZ_N} f(t) \in \ck \}$}.
The restrictions will again be denoted by  $\hat{L}^{(N)}(b)$ and  $\Check{L}^{(N)}(b)$.

\begin{lemma} \label{lem_obs1}
We have\footnote{the $-$ sign in Eqs. \eqref{eq_obs2.3a}
and \eqref{eq_obs2.3b} holds iff both $r$ and $N-1$ are odd}
 \begin{align}
 \label{eq_obs2.3a} \det(\hat{L}^{(N)}(b))  & =  \pm \det\bigl(1_{{\ck}}-\exp(\ad({b}))_{|
{\ck}}\bigr) \cdot N^{d},\\
 \label{eq_obs2.3b} \det(\Check{L}^{(N)}(b))  & =  \pm \det\bigl(1_{{\ck}}-\exp(\ad({b}))_{|
{\ck}}\bigr) \cdot N^{d},
\end{align}
where $d := \dim(\Map(\bZ_N,\cG)) =  N \dim(\cG)$, $r:= \dim(\ct)$.
\end{lemma}
 \begin{proof}  Let us prove Eq. \eqref{eq_obs2.3a}.  The proof of Eq. \eqref{eq_obs2.3b} is similar.
    First observe that
   \begin{equation}  \label{eq1_in_obs3.8}
 \det(\hat{L}^{(N)}(b)) = \det( N( \tau_1 e^{\ad(b)/N} -1)) = N^{d'} \det( \tau_1  - e^{-\ad(b)/N})
 \end{equation}
 where $d' :=  \dim(\Map'(\bZ_N,\cG)) = d - \dim(\ct) $ and where
 we have used that  $e^{\ad(b)/N}$ is orthogonal.
   The  complexified operator
  $( \tau_1  - e^{-\ad(b)/N}) \otimes \id_{\bC}$ is  diagonalizable
  with eigenvalues
  \begin{align*} \lambda_{k,{\alpha}} & := e^{ \frac{2 \pi i k}{N}} - e^{\frac{-\alpha(b)}{N}}  , \quad \quad  \text{ for each } k \in \{1,2,\ldots,N\}, {\alpha} \in {\cR}_{\bC}\\
   \mu_{k,a} & :=    e^{ \frac{2 \pi i k}{N}}-1  \quad \quad \text{ for each }  k \in  \{1,2,\ldots,N-1\},  a \in \{1,2, \ldots, r\}
  \end{align*}
     where  ${\cR}_{\bC}$ denotes the set of complex roots  of
${\cG}$ w.r.t. ${\ct}$ (cf. part \ref{appB} of the Appendix).
Using the two polynomial equations
 $x^N-1= \prod_{k=0}^{N-1} (x-e^{ \frac{2 \pi i k}{N}}) = (-1)^N \prod_{k=0}^{N-1} (e^{ \frac{2 \pi i k}{N}} - x) $
and $x^{N-1} + x^{N-2} + \ldots + 1= \prod_{k=1}^{N-1}
 (x-e^{ \frac{2 \pi i k}{N}}) = (-1)^{N-1} \prod_{k=1}^{N-1}
 (e^{ \frac{2 \pi i k}{N}} - x) $ we therefore  obtain
\begin{align*}& \det( \tau_1  - e^{-\ad(b)/N})
= \det\nolimits_{\bC}\bigl( (\tau_1  - e^{-\ad(b)/N}) \otimes \id_{\bC}\bigr) \\
&  =  \biggl( \prod_{{\alpha} \in {\cR}_{\bC}}  \prod_k
  \bigl\{ e^{ \frac{2 \pi i k}{N}} - e^{\frac{-\alpha(b)}{N}} \bigr\}
 \biggr)
\biggl( \prod_{a =1}^r  \bigl\{ \prod_{k \neq 0} ( e^{ \frac{2 \pi i k}{N}} - 1) \bigr\} \biggr)\\
& =   \biggl(
\prod_{{\alpha} \in {\cR}_{\bC}} (-1)^{N} \bigl\{ e^{-\alpha(b)}- 1\bigr\} \biggr)
 \biggl( \prod_{a =1}^r (-1)^{N-1} \bigl\{ N \bigr\} \biggr) = (-1)^{r (N-1)} N^{r} \prod_{{\alpha} \in {\cR}_{\bC}} \bigl\{ e^{-\alpha(b)}- 1\bigr\}
\end{align*}
The assertion now follows by combining the last equation with Eq. \eqref{eq1_in_obs3.8}
above and by taking into account the relations  $d = d'  + r$
and  $\prod_{{\alpha} \in {\cR}_{\bC}}  (e^{-{\alpha}({b})}-1)   =  \prod_{{\alpha} \in {\cR}_{\bC}}  (1-e^{{\alpha}({b})})=  \det\bigl((1_{{\ck}}-\exp(\ad({b}))_{|{\ck}}) \otimes \id_{\bC}\bigr) =
      \det\bigl(1_{{\ck}}-\exp(\ad({b}))_{|{\ck}}\bigr)$.
 \end{proof}

\subsubsection*{\it b) Some consequences of conditions (NCP)' and (NH)'}

Let $L=(R_1,R_2,\ldots, R_m)$ be the simplicial ribbon link in $q\cK \times \bZ_N$
fixed in Sec. \ref{subsec4.9} above.
Recall that each  $R_i = (F_k)_{k \le n_i}$,  $n_i \in \bN$,
 ``induces''\footnote{$l_i$ and $l'_i$ are just the two loops on the boundary of $R_i$}
  two simplicial loops  $l_i$ and $l'_i$ in $q\cK \times \bZ_N$.
In the following we use the short notation $l^i_{\Sigma} := (l_i)_{\Sigma}$,
$l^{'i}_{\Sigma} := (l'_i)_{\Sigma}$, $l^i_{S^1} := (l_i)_{S^1}$, and
$l^{'i}_{S^1} := (l'_i)_{S^1}$ for the $\Sigma$- or $S^1$-projections of these loops
and  we will often consider  $l^j_{\Sigma}$ and $l^{'j}_{\Sigma}$
 as (piecewise smooth) maps $S^1 \to \Sigma$.
Clearly,  $\arc(l^{i}_{\Sigma})$ and $\arc(l^{'i}_{\Sigma})$ can then be considered as
subsets of $\Sigma$.

\smallskip

It is not difficult to see that the two conditions (NCP)' and (NH)' on our simplicial ribbon link
$L$ imply the following conditions:

\begin{description}

\item[(FC1)] For all $i,j \le m$ we have
$\arc(l^i_{\Sigma})  \cap \arc(l^{'j}_{\Sigma}) = \emptyset$
and we also have $\arc(l^i_{\Sigma})  \cap \arc(l^{j}_{\Sigma}) = \emptyset$ and
$\arc(l^{'i}_{\Sigma})  \cap \arc(l^{'j}_{\Sigma}) = \emptyset$ if $i \neq j$.

\item[(FC2)] For each $i \le m$ the open region
 $O_i \subset \Sigma$  ``between''\footnote{more precisely, $O_i$
 is the interior of  the subset $\Image(R^i_{\Sigma})$ of $\Sigma$}  $\arc(l^i_{\Sigma})$ and  $\arc(l^{'i}_{\Sigma})$
does not contain an element of $\face_0(q\cK)$.

\item[(FC3)] For each $i \le m$ and every $F \in \face_2(q\cK)$ with $F \subset  \Image(R^i_{\Sigma})$
exactly one of the four sides of (the tetragon) $F$  will lie on $\arc(l^i_{\Sigma})$
and exactly one side will lie on $\arc(l^{'i}_{\Sigma})$

\item[(FC4)] $l^i_{S^1} = l'^i_{S^1}$ is fulfilled for each $i \le m$

 \end{description}

In order to simplify the notation in Secs \ref{subsec5.1}-\ref{subsec5.3}   below
we will set
 $$n := \max_{i \le m} n_i$$
 and we will  extend each simplicial loop $l^i_{\Sigma}$, $l^{'i}_{\Sigma}$, $l^i_{S^1}$,
$l^{'i}_{S^1}$ to a simplicial loop with ``length'' $n$ in a trivial way, i.e.
by ``adding'' $n -n_i$ empty edges. For the extended simplicial loops we will use the same notation.

\smallskip

Finally we set, for each $i \le m$ and $k \le n$
\begin{equation}\bar{l}^{i(k)}_{\Sigma}  := \pi( l^{i(k)}_{\Sigma}), \quad \quad \text{ and } \quad \quad
  \bar{l}^{'i(k)}_{\Sigma}  := \pi( l^{'i(k)}_{\Sigma})
   \end{equation}
  where $\pi:C_1(q\cK) \to C_1(K) (\subset C_1(q\cK)) $ is the orthogonal projection.
  Observe that for each $i \le m$ we  have
\begin{subequations}  \label{eq_FC5}
\begin{align}
 & \forall k \le n: \bar{l}^{i(k)}_{\Sigma} \in C_1(K_1) \quad  \text{ and }  \quad
                \forall k \le n: \bar{l}^{'i(k)}_{\Sigma} \in C_1(K_2)
 \end{align}
 or
 \begin{align}  & \forall k \le n: \bar{l}^{i(k)}_{\Sigma} \in C_1(K_2) \quad \text{ and } \quad
  \forall k \le n: \bar{l}^{'i(k)}_{\Sigma} \in C_1(K_1)
\end{align}
\end{subequations}

From (FC1) and Eqs \eqref{eq_FC5} it easily follows that
 for all  $i_1,i_2 \le m$, and $k_1, k_2 \le n$
  we have\footnote{or, more precisely, both
 $\star_K \bar{l}^{i_1(k_1)}_{\Sigma} \neq  \bar{l}^{'i_2(k_2)}_{\Sigma}$
 and $\star_K \bar{l}^{i_1(k_1)}_{\Sigma} \neq - \bar{l}^{'i_2(k_2)}_{\Sigma}$       etc.}
\begin{equation} \label{eq_FC6} \star_K  \bar{l}^{i_1(k_1)}_{\Sigma} \neq \pm \bar{l}^{'i_2(k_2)}_{\Sigma},
\quad \quad  \star_K  \bar{l}^{i_1(k_1)}_{\Sigma} \neq \pm \bar{l}^{i_2(k_2)}_{\Sigma}, \quad \quad
\star_K  \bar{l}^{'i_1(k_1)}_{\Sigma} \neq \pm \bar{l}^{'i_2(k_2)}_{\Sigma}
\end{equation}
 provided that $\bar{l}^{i_1(k_1)}_{\Sigma} \neq 0$ and $\bar{l}^{'i_1(k_1)}_{\Sigma} \neq 0$.
Here $\star_K$ is the linear isomorphism on $C_1(K) = C^1(K,\bR)$
which is defined exactly in the same way as the operator
$\star_K$ on $\cA_{\Sigma}(K) = C^1(K,\cG)$ which we introduced in Sec. \ref{subsec4.2} above.

\subsection{Step 1: Performing the $\int\nolimits_{\sim}  \cdots
\exp(iS^{disc}_{CS}(\Check{A}^{\orth},B))  D\Check{A}^{\orth} $
integration in Eq. \eqref{eq_def_WLOdisc}}
\label{subsec5.1}

\begin{lemma} \label{lem1} Under the assumptions  on the simplicial ribbon link $L=(R_1,\ldots,R_m)$
 made above we have for every fixed $A^{\orth}_c \in  \cA^{\orth}_c(K)$ and $B \in \cB_{reg}(q\cK)$
\begin{equation}  \label{eq5.1}\int\nolimits_{\sim}   \prod_i  \Tr_{\rho_i}\bigl( \Hol^{disc}_{R_i}(\Check{A}^{\orth} + A^{\orth}_c,
  B)\bigr)  \exp(iS^{disc}_{CS}(\Check{A}^{\orth},B))  D\Check{A}^{\orth} = Z^{disc}_B \prod_i  \Tr_{\rho_i}\bigl(
\Hol^{disc}_{R_i}(A^{\orth}_c,   B)\bigr)
\end{equation}
where $Z^{disc}_B := \int_{\sim} \exp(iS^{disc}_{CS}(\Check{A}^{\orth},B))  D\Check{A}^{\orth}$
\end{lemma}

\begin{proof} Let  $A^{\orth}_c \in \cA^{\orth}_c(K)$ and   $B \in \cB_{reg}(q\cK)$ be as in the assertion of the lemma.
 In order to prove the lemma we will apply Proposition \ref{prop3.3} (and Remark \ref{rm3.1} above)
 to the special situation where (cf. Convention \ref{conv_EucSpaces} in Sec. \ref{subsec4.8})
 \begin{itemize}
 \item   $V= \Check{\cA}^{\orth}(K)$,
\item $d\mu= d\nu^{disc}_B$ with  $d\nu^{disc}_B = \tfrac{1}{Z^{disc}_B} \exp(iS^{disc}_{CS}(\Check{A}^{\orth},B))  D\Check{A}^{\orth}$,

\item  $(Y^{i,a}_{k})_{i \le m,k \le n,a \le \dim(\cG)}$  is the
 family   of maps $Y^{i,a}_{k}:  \Check{\cA}^{\orth}(K) \to \bR$ given by\footnote{observe that
 in view of condition (FC4) above the RHS of Eq. \eqref{eq5.4} is very closely related
 to the RHS of Eq. \eqref{eq4.21} in Sec. \ref{subsec4.3}}
\begin{multline}  \label{eq5.4}
Y^{i,a}_k(\Check{A}^{\orth})  :=  \biggl\langle T_a,
 \bigl( \Check{A}^{\orth}(\start l^{i(k)}_{S^1})\bigr)(\tfrac{1}{2} \ l^{i(k)}_{\Sigma}+  \tfrac{1}{2} \ l^{'i(k)}_{\Sigma}) +
 A^{\orth}_c(\tfrac{1}{2} \ l^{i(k)}_{\Sigma}+  \tfrac{1}{2} \ l^{'i(k)}_{\Sigma}) \\
   +  \bigl( \tfrac{1}{2} B(\start l^{i(k)}_{\Sigma}) + \tfrac{1}{2} B(\start l^{'i(k)}_{\Sigma}) \bigr) \cdot dt^{(N)}(l^{i(k)}_{S^1})\biggr\rangle
\end{multline}
where $(T_a)_{a \le \dim(\cG)}$ is an arbitrary
$\langle \cdot, \cdot  \rangle$-ONB of $\cG$ (which will be kept fixed in the following), and

\item $\Phi: \bR^{m \times n \times \dim(\cG)} \to \bC$  is given by
\begin{equation}  \label{eq5.5} \Phi((x^{i,a}_k)_{i,k,a}) =  \prod_{i=1}^m  \Tr_{\rho_i}(\prod_{k=1}^n \exp(\sum_{a=1}^{\dim(\cG)} T_a x^{i,a}_k))
\quad \quad \text{for all $(x^{i,a}_k)_{i,k,a} \in \bR^{m \times n
\times \dim(\cG)}$}
\end{equation}
\end{itemize}
 Observe that
\begin{enumerate}
\item $d\nu^{disc}_{B}$ is a  well-defined normalized  non-degenerate centered oscillatory Gauss-type measure.
   Since by assumption $B \in \cB_{reg}(q\cK)$ this follows from Eq.  \eqref{eq_SCS_expl_discb} and Proposition \ref{prop4.2} in Sec. \ref{subsec4.2} above and from   Proposition \ref{obs2} in Sec. \ref{subsec5.0}.\par

   For later use let us  mention that  according to  Example \ref{obs1} above
  and Eq. \eqref{eq_determinante} in Proposition \ref{obs2} above we have
\begin{multline} \label{eq_def_ZBdisc}
 Z^{disc}_B = \int_{\sim} \exp(iS^{disc}_{CS}(\Check{A}^{\orth},B))  D\Check{A}^{\orth} \\
\sim   \det(L^{(N)}(B))^{-1/2} \sim \prod_{\bar{e} \in  \face_0(K_1 | K_2)}
  \det\bigl(1_{{\ck}}- \ad(B(\bar{e}))_{| {\ck}}\bigr)^{-1}
\end{multline}

\item $\Phi \in \cP_{exp}(\bR^{m \times n \times \dim(\cG)})$  since
\begin{multline}  \label{eq5.6}
\Phi((x^{i,a}_k)_{i,k,a}) =  \prod_i  \Tr_{\rho_i}(\prod_k
\exp(\sum_a T_a x^{i,a}_k))
  =\prod_i  \Tr_{\End(V_i)} \bigl( \prod_k  \rho_i(\exp( \sum_a T_a x^{i,a}_k )) \bigr) \\
 =  \prod_i  \Tr_{\End(V_i)} \bigl( \prod_k \exp_{\End(V_i)}( \sum_a  \bigl(
 (\rho_i)_* T_a \bigr) x^{i,a}_k ) \bigr)
 \end{multline}
where   $\exp_{\End(V_i)}$ is  exponential map of the associative algebra\footnote{observe that
the vector spaces underlying $\End(V_i)$ and $\gl(V_i)$ coincide}
$\End(V_i)$ and  $(\rho_i)_* : \cG \to \gl(V_i)$, for $i \le m$, is the Lie algebra representation
 induced by $\rho_i$.

\item For all $i \le m$, $k \le n$, $a \le \dim(\cG)$ we have
\begin{equation}  \label{eq5.7}
\int\nolimits_{\sim}  Y^{i,a}_k \ d\nu^{disc}_{B} = Y^{i,a}_k(0)
\end{equation}
In order to see this let us introduce
$j:\bZ_N \to C_1(q\cK)$
by
$$j(t):= \begin{cases}
\tfrac{1}{2} \ l^{i(k)}_{\Sigma}+  \tfrac{1}{2} \ l^{'i(k)}_{\Sigma} & \text{ if } t =  \start l^{i(k)}_{S^1}\\
 0 &  \text{ if } t \neq  \start l^{i(k)}_{S^1}
\end{cases}$$
and set $\Check{j}_a:= p(T_a j)$ where  $p:   \cA^{\orth}(q\cK) \to  \Check{\cA}^{\orth}(K)$
is the $\ll  \cdot, \cdot\gg_{\cA^{\orth}(q\cK)}$-orthogonal projection
onto the subspace $\Check{\cA}^{\orth}(K)$ of  $\cA^{\orth}(q\cK)$
and where $T_a j := j \otimes T_a \in \cA^{\orth}(q\cK) \cong \Map(\bZ_N, C_1(q\cK)) \otimes \cG$.
 Then
\begin{multline*}
Y^{i,a}_k(\Check{A}^{\orth}) - Y^{i,a}_k(0) \\
= \langle T_a , \bigl( (\Check{A}^{\orth})(\start l^{i(k)}_{S^1})\bigr)(\tfrac{1}{2} \ l^{i(k)}_{\Sigma}+  \tfrac{1}{2} \ l^{'i(k)}_{\Sigma}) \rangle
  =  \ll  \Check{A}^{\orth} , T_a j\gg_{\cA^{\orth}(q\cK)}
  =  \ll  \Check{A}^{\orth} ,\Check{j}_a \gg_{\cA^{\orth}(q\cK)}
\end{multline*}
On the other hand since  $d\nu^{disc}_{B}$ is centered
  Example \ref{obs1} above
   implies that $\int_{\sim}
 \ll  \cdot , \Check{j}_a\gg_{\cA^{\orth}(q\cK)} d\nu^{disc}_B = 0$.
 Since $d\nu^{disc}_{B}$ is also normalized  we obtain Eq. \eqref{eq5.7}.

\item For all $i,i' \le m$, $k,k' \le n$, $a,a' \le \dim(\cG)$ we have
\begin{equation}  \label{eq5.8} \int\nolimits_{\sim}  Y^{i,a}_k Y^{i',a'}_{k'} \ d\nu^{disc}_{B}
 =   \int\nolimits_{\sim}  Y^{i,a}_k \ d\nu^{disc}_{B}
 \int\nolimits_{\sim} Y^{i',a'}_{k'} \ d\nu^{disc}_{B}
 \end{equation}
 This follows from Eq. \eqref{eq5.7} above and
\begin{align}  \label{eq5.9}
& \int\nolimits_{\sim}  (Y^{i,a}_k - Y^{i,a}_k(0))
 (Y^{i',a'}_{k'} - Y^{i',a'}_{k'}(0)) d\nu^{disc}_{B} =
 \int\nolimits_{\sim}  \ll \cdot , \Check{j}_a \gg_{\cA^{\orth}(q\cK)} \ll \cdot ,
 \Check{j}'_{a'} \gg_{\cA^{\orth}(q\cK)} d\nu^{disc}_{B} \nonumber \\
& \quad \overset{(*)}{\sim}  \quad
\ll \Check{j}_a, \bigl(\star_{K} L^{(N)}(B)  \bigr)^{-1}
\Check{j}'_{a'} \gg_{\cA^{\orth}(q\cK)} \overset{(**)}{=} 0
\end{align}
where $\Check{j}_a$ is as in point iii) above and where
$\Check{j}'_{a'}$ is defined in a completely analogous way with
$i$, $k$, and $a$ replaced by $i'$, $k'$, and $a'$.
  Here step $(*)$ follows from  Example \ref{obs1} above
and step $(**)$ follows from the inequalities \eqref{eq_FC6} appearing at the end
of Sec. \ref{subsec5.0} above.
\end{enumerate}
Thus the assumptions of Proposition \ref{prop3.3} above are fulfilled
 and we obtain
\begin{multline}  \label{eq5.11} \frac{1}{Z^{disc}_B} \int\nolimits_{\sim}
 \prod_i  \Tr_{\rho_i}\bigl( \Hol^{disc}_{R_i}(\Check{A}^{\orth} + A^{\orth}_c,
  B)\bigr)   \exp(iS^{disc}_{CS}(\Check{A}^{\orth},B))
D\Check{A}^{\orth}\\
\overset{(+)}{=} \int\nolimits_{\sim}  \prod_i  \Tr_{\rho_i}(\prod_k \exp(\sum_a T_a Y^{i,a}_k))
d\nu^{disc}_{B} =  \int\nolimits_{\sim}  \Phi((Y^{i,a}_k)_{i,k,a}) \ d\nu^{disc}_{B} \overset{(*)}{=}
 \Phi((\int\nolimits_{\sim} Y^{i,a}_k \ d\nu^{disc}_{B})_{i,k,a}) \\
\overset{(**)}{=}  \Phi((Y^{i,a}_k(0))_{i,k,a}) =  \prod_i \Tr_{\rho_i}( \prod_k
\exp( \sum_a T_a Y^{i,a}_k(0))) = \prod_i \Tr_{\rho_i}\bigl(
\Hol^{disc}_{R_i}(A^{\orth}_c,   B)\bigr)
\end{multline}
Here  step $(+)$ follows from the definitions  and condition (FC4) in Sec. \ref{subsec5.0} above, step $(*)$ follows from Proposition \ref{prop3.3} and Remark \ref{rm3.1}  above, and step $(**)$ follows
from Eq. \eqref{eq5.7} above.
\end{proof}

Using Lemma \ref{lem1} and taking into account the implication
  $$ \prod_{x} 1^{(s)}_{\ct_{reg}}(B(x))
\neq 0 \quad \Rightarrow \quad B \in \cB_{reg}(q\cK)$$
for all $s>0$  we now obtain from Eq. \eqref{eq_def_WLOdisc}
\begin{multline} \label{eq5.14}
\WLO^{disc}_{rig}(L)\\
=  \lim_{s \to 0} \sum_{y \in I}
\int\nolimits_{\sim} \biggl\{
 \bigl( \prod_{x} 1^{(s)}_{\ct_{reg}}(B(x)) \bigr)
  \prod_i  \Tr_{\rho_i}\bigl(\Hol^{disc}_{R_i}( A^{\orth}_c, B)\bigr) \Det^{disc}(B)  \biggr\}\\
  \times \exp\bigl( -  2 \pi i k   \langle y, B(\sigma_0) \rangle \bigr)
  \exp(i  S^{disc}_{CS}(A^{\orth}_c,B))    (DA^{\orth}_c \otimes  DB)
\end{multline}
where we have set
\begin{equation} \label{eq_def_Detdisc}
\Det^{disc}(B) := \Det^{disc}_{FP}(B) Z^{disc}_B
 \end{equation}

\subsection{Step 2: Performing the $\int\nolimits_{\sim}  \cdots
 \exp(i  S^{disc}_{CS}(A^{\orth}_c,B))    (DA^{\orth}_c \otimes  DB)$-integration in \eqref{eq5.14}}
\label{subsec5.2}

With the help of Proposition \ref{prop3.5} above  let us now evaluate
the $\int\nolimits_{\sim}  \cdots
 \exp(i  S^{disc}_{CS}(A^{\orth}_c,B))    (DA^{\orth}_c \otimes  DB)$-integral appearing
 in Eq.  \eqref{eq5.14}. In order to
do so we first rewrite the integrand in Eq. \eqref{eq5.14} in
such a way that it assumes the form of the integrand on the LHS  of
the formula appearing in Proposition \ref{prop3.5} above.
 In order to achieve this we will now exploit
the fact that  all of the
 remaining fields $A^{\orth}_c$ and  $B$  in Eq.  \eqref{eq5.14}
take values in the {\it Abelian} Lie algebra  $\ct$.
For  fixed  $A^{\orth}_c$ and $B$ we can therefore rewrite
 $\Hol^{disc}_{R_i}(A^{\orth}_c,B)$ as an exponential of a sum, namely as
\begin{equation}  \label{eq5.16} \Hol^{disc}_{R_i}( A^{\orth}_c, B)
  = \exp\bigl(  \Phi_i(B) +   \sum_k  A^{\orth}_c(\tfrac{1}{2} \ l^{i(k)}_{\Sigma}+  \tfrac{1}{2} \ l^{'i(k)}_{\Sigma})    \bigr)
\end{equation}
(cf. Eq. \eqref{eq4.21}  above)
where we  have set for each $i \le m$
\begin{equation} \label{eq5.17} \Phi_i(B)  :=  \sum_k  \bigl( \tfrac{1}{2} B(\start l^{i(k)}_{\Sigma})
 + \tfrac{1}{2} B(\start l^{'i(k)}_{\Sigma}) \bigr) \cdot dt^{(N)}(l^{i(k)}_{S^1})
\end{equation}

Moreover, since $\Hol^{disc}_{R_i}( A^{\orth}_c, B) \in T$ we can
replace in Eq. \eqref{eq5.14}  the characters $\chi_i:=
\Tr_{\rho_i}$, $i \le m$,   by their restrictions $\chi_{i|T}$. But
$\chi_{i|T}$ is just a linear combination of global weights, more
precisely, for every $b \in \ct$ we have
\begin{equation} \label{eq5.18}
\Tr_{\rho_i}(\exp(b)) = \chi_{i|T}(\exp(b)) = \sum_{\alpha \in
\Lambda} m_{\chi_i}(\alpha) e^{2 \pi i \langle \alpha, b \rangle}
\end{equation}
where $m_{\chi_i}(\alpha)$  the multiplicity of $\alpha \in \Lambda$ as a weight
in $\chi_i$
(here $\Lambda \subset \ct^* \cong \ct$ denotes
 the lattice of the real weights associated to the pair $(\cG,\ct)$, cf. part \ref{appB} of the Appendix below).
Combining Eqs. \eqref{eq5.16} -- \eqref{eq5.18}  we obtain\footnote{here we are a bit sloppy and use the letter
$i$ both for  the multiplication index and the imaginary unit}
\begin{align} \label{eq5.19}
& \prod_i  \Tr_{\rho_i}\bigl( \Hol^{disc}_{R_i}( A^{\orth}_c,   B) \bigr) \nonumber \\
& = \prod_i  \biggl(  \sum_{\alpha_i \in \Lambda} m_{\chi_i}(\alpha_i) \cdot \exp(2 \pi i \langle \alpha_i,
\Phi_i(B) \rangle )  \cdot
\exp\bigl(  2 \pi i   \sum_k \langle \alpha_i,
A^{\orth}_c(\tfrac{1}{2} \ l^{i(k)}_{\Sigma}+  \tfrac{1}{2} \ l^{'i(k)}_{\Sigma})  \rangle \bigr) \biggr) \nonumber \\
 & =   \sum_{\alpha_1, \alpha_2, \ldots, \alpha_m \in \Lambda} \bigl( \prod_i  m_{\chi_i}(\alpha_i) \bigr)
\bigl( \prod_i \exp(2 \pi i \langle \alpha_i, \Phi_i(B) \rangle )  \bigr)
 \exp\bigl(  2 \pi i  \ll  A^{\orth}_c ,  \sum_{i} \alpha_i \cdot {\mathfrak l}^{i}_{\Sigma} \gg_{\cA^{\orth}(q\cK)}  \bigr)
\end{align}
where  we have set
\begin{equation} \label{eq_mathfrak_l_def}
 {\mathfrak l}^{i}_{\Sigma} :=
  \sum_k  \tfrac{1}{2} ( l^{i(k)}_{\Sigma} + l^{'i(k)}_{\Sigma})  \in C_1(q\cK)
 \end{equation}

Let us now set  for each    $s > 0$, $y \in I$,
$(\alpha_i)_i := (\alpha_1, \alpha_2, \ldots, \alpha_m) \in \Lambda^m$,
and $B \in \cB(q\cK)$:

\begin{multline}  \label{eq5.22}
F^{(s)}_{(\alpha_i)_i,y}(B)  :=  \bigl( \prod_{x}
1^{(s)}_{\ct_{reg}}(B(x)) \bigr)
 \bigl( \prod_i \exp(2 \pi i \langle \alpha_i, \Phi_i(B) \rangle )  \bigr) \Det^{disc}(B) \\
 \times   \exp\bigl( - 2 \pi i k    \langle y, B(\sigma_0)  \rangle \bigr)
 \end{multline}
Then, according to Eq. \eqref{eq5.19},
 we can rewrite  Eq. \eqref{eq5.14} as
\begin{multline} \label{eq5.21} \WLO^{disc}_{rig}(L)
= \lim_{s \to 0} \sum_{(\alpha_i)_i \in \Lambda^m}  \bigl( \prod_i  m_{\chi_i}(\alpha_i) \bigr)
 \sum_{y \in I}\\
 \times \int_{\sim} F^{(s)}_{(\alpha_i)_i,y}(B)
  \exp\bigl(  2 \pi i  \ll  A^{\orth}_c ,  \sum_{i} \alpha_i \cdot {\mathfrak l}^{i}_{\Sigma}  \gg_{\cA^{\orth}(q\cK)}  \bigr) \exp(i  S^{disc}_{CS}(A^{\orth}_c,B))    (DA^{\orth}_c \otimes  DB)
\end{multline}
 (Observe that there are only finitely many $(\alpha_i)_i \in \Lambda^m$
for which the product  $\prod_i  m_{\chi_i}(\alpha_i)$ does not vanish.
Accordingly, the summation  $\sum_{(\alpha_i)_i \in \Lambda^m}  \bigl( \prod_i  m_{\chi_i}(\alpha_i) \bigr) \cdots $
above is a finite and we can interchange it with  $\sum_{y \in I}$).

\smallskip

Let us now fix for a while $s>0$, $y \in I$, and  $(\alpha_i)_i \in \Lambda^m$
and evaluate the corresponding $\int_{\sim} \cdots $-integral
in Eq. \eqref{eq5.21}. In order to do so we will
apply Proposition \ref{prop3.5} above
to the special situation where

\begin{itemize}
\item $V  := \cA^{\orth}_c(K) \oplus  \cB_0(q\cK)$ (cf. Convention \ref{conv_EucSpaces}),

\item $d\mu  := d\nu^{disc}$ where
\begin{multline*} d\nu^{disc} := \tfrac{1}{Z^{disc}} \exp(iS^{disc}_{CS}(A^{\orth}_c,B))
(DA^{\orth}_c  \otimes DB)\\
  = \tfrac{1}{Z^{disc}} \exp(  i
 \ll  A^{\orth}_c,  -   2 \pi  k (\star_K \circ \pi \circ d_{q\cK}) B \gg_{\cA^{\orth}(q\cK)}  ) (DA^{\orth}_c \otimes DB)
\end{multline*}
where\footnote{recall Eq. \eqref{eq_SCS_expl_discc}
 and observe that  $ \ll  \star_K A^{\orth}_c,  d_{q\cK}  B \gg_{\cA^{\orth}(q\cK)}
  =  \ll  \star_K A^{\orth}_c,  \pi( d_{q\cK}  B) \gg_{\cA^{\orth}(q\cK)}
 = - \ll   A^{\orth}_c, \star_K ( \pi( d_{q\cK}  B)) \gg_{\cA^{\orth}(q\cK)}$}
  we have set $Z^{disc}:= \int_{\sim} \exp(i  S^{disc}_{CS}(A^{\orth}_c,B))    (DA^{\orth}_c \otimes  DB)$
and where  $\pi:\cA_{\Sigma,\ct}(q\cK) \to  \cA_{\Sigma,\ct}(K) \cong \cA^{\orth}_c(K)$ is the orthogonal
 projection  (cf. also Convention \ref{conv1neu} above),

\item   $V_1 := \ker(\star_K \circ
\pi \circ d_{q\cK})^{\orth}  = \ker(\pi \circ d_{q\cK})^{\orth} \overset{(*)}{=} (\cB_c(q\cK))^{\orth} \subset \cB_0(q\cK)$
where in $(*)$ we used  Eq. \eqref{eq_obs1} in Sec. \ref{subsec4.10a} above.
(Here and in the next paragraph   $d_{q\cK}$ is a short notation for  $(d_{q\cK})_{| \cB_0(q\cK)}$).

\item  $V_2:= \Image(\star_K \circ \pi \circ d_{q\cK}) \subset  \cA^{\orth}_c(K)$

\item  $V_0 := \cB_c(q\cK) \oplus (V_2)^{\orth}$

 \smallskip

where $(V_2)^{\orth}$ is the orthogonal complement of $V_2$ in $\cA^{\orth}_c(K)$
(cf. Convention \ref{conv_EucSpaces}).

\item  $F := F^{(s)}_{(\alpha_i)_i,y} \circ p$
where $p: V_0 \oplus V_1 =  \cB_0(q\cK) \oplus (V_2)^{\orth} \to \cB_0(q\cK)$
is the obvious projection.

\item $v :=  2 \pi   \sum_i \alpha_i \cdot  {\mathfrak l}^{i}_{\Sigma}$
\end{itemize}

\noindent The following remarks show that  the assumptions of Proposition \ref{prop3.5}  above are indeed fulfilled:
\begin{enumerate}

\item      $d\nu^{disc}$ is  a normalized centered oscillatory Gauss type measure on  $\cA^{\orth}_c(K) \oplus \cB_0(q\cK)$  which has the form as in Proposition \ref{prop3.5}
    with $V_0$, $V_1$, and $V_2$ given as above and where
     $M: V_1 \to V_2$  is the well-defined   linear isomorphism given by
   \begin{equation} \label{eq_def_M} M  =  -   2 \pi  k (\star_K \circ \pi \circ d_{q\cK})_{| V_1}
   \end{equation}

\item The function  $F$ is bounded and uniformly continuous

 \item In order to see that $v$ is an element of $V_2$
  we consider the linear map $m_{\bR}:  C^0(q\cK,\bR) \to  C^1(q\cK,\bR)$
 which is given by
 $$m_{\bR} := \star_K \circ \pi \circ  d_{q\cK}$$
  where
 $d_{q\cK}:C^0(q\cK,\bR) \to  C^1(q\cK,\bR)$,
 $\pi: C^1(q\cK,\bR) \to C^1(K,\bR)$,   and  $\star_{K}:C^1(K,\bR) \to  C^1(K,\bR)$
 are the ``real analogues'' of the three maps appearing on the RHS of Eq. \eqref{eq_def_M} above.
   From Lemma \ref{lem2_pre} in Sec. \ref{subsec5.3} below it follows that
for each $ {\mathfrak
l}^{i}_{\Sigma} \in C_1(q\cK) = C^1(q\cK,\bR)$
 there is a $f_i \in C^0_{\aff}(q\cK,\bR)$
 such that
\begin{equation} \label{eq5.28}
  {\mathfrak l}^{i}_{\Sigma} =   m_{\bR} \cdot f_i
 \end{equation}
 holds.   Moreover,    $f_i$ can be chosen to be constant on $U(\sigma_0) \cap \face_0(q\cK)$.
   In the following we will assume that $f_i$ is chosen in this way.
   (According to Lemma \ref{lem2_pre} below this determines $f_i$ uniquely up to an additive
   constant   which will be fixed later using a suitable normalization
   condition, see Eq. \eqref{eq5.28''} and  Eq. \eqref{eq5.28'}    below).
      Clearly, this implies that
\begin{equation} \label{eq_def_v} v =  2 \pi   \sum_i \alpha_i \cdot  {\mathfrak l}^{i}_{\Sigma} =
   2 \pi   \sum_i \alpha_i \cdot m_{\bR} \cdot f_i   = (\star_K \circ \pi \circ d_{q\cK}) \cdot \bigl( 2 \pi   \sum_i \alpha_i \cdot  f_i\bigr)
   \end{equation}
   is indeed an element of $V_2$.

\end{enumerate}

\noindent
Applying Proposition \ref{prop3.5}  above  to the present situation we therefore obtain
\begin{align} \label{eq5.33} & \tfrac{1}{ Z^{disc}} \int_{\sim} F^{(s)}_{(\alpha_i)_i,y}(B)
  \exp\bigl(  2 \pi i \ll  A^{\orth}_c ,  \sum_{i} \alpha_i \cdot {\mathfrak l}^{i}_{\Sigma}  \gg_{\cA^{\orth}(q\cK)}  \bigr) \exp(iS^{disc}_{CS}(A^{\orth}_c,B)) (DA^{\orth}_c  \otimes DB)   \nonumber \\
  & = \int_{\sim} F(B)
  \exp\bigl(   i \ll  A^{\orth}_c , v  \gg_{\cA^{\orth}(q\cK)}  \bigr)
    d\nu^{disc}(A^{\orth}_c,B)   \nonumber \\
 & \quad \quad \overset{(+)}{\sim}     \int^{\sim}_{V_0}
  F( x_0 - M^{-1} v)  dx_0  \overset{(*)}{\sim}
   \int^{\sim}_{\ct}    F^{(s)}_{(\alpha_i)_i,y}(b - M^{-1} v) db
\end{align}
where the two improper integrals $\int^{\sim} \cdots dx_0$ and
$\int^{\sim} \cdots db$ are defined\footnote{that the last (and therefore also the first)
of these two improper integrals is in fact well-defined follows from Remark \ref{rm_last_sec3} above  and a ``periodicity argument'', which will be given in  Step 4 below}
 according to Convention \ref{conv3.1} in Sec. \ref{sec3} above.
In Step $(+)$ we have applied Proposition \ref{prop3.5}.
Step $(*)$ above follows because $F( x_0 - M^{-1} v)$ depends only on the $\cB_{c}(q\cK)$-component of
$x_0 \in V_0 =  \cB_{c}(q\cK) \oplus (V_2)^{\orth}$
and because $\cB_{c}(q\cK) = \{B \in \cB(q\cK) \mid B \text{ is constant } \} \cong  \ct  $.

\medskip

Recall that $f_i$ as introduced above is uniquely determined
up to an additive constant. We can fix this constant by
demanding that the normalization condition
 \begin{equation} \label{eq5.28''} \sum_{x \in \face_0(q\cK)} f_i(x) = 0
\end{equation}
is fulfilled. With this normalization condition we obtain
 $    \sum_i \alpha_i \cdot f_i  \in V_1 = (\cB_c(q\cK))^{\orth}$.
 From this and Eq. \eqref{eq_def_v} above  we see that
\begin{equation} \label{eq5.34} M^{-1} v =    - \tfrac{1}{k} \sum_i \alpha_i \cdot f_i
\end{equation}
\noindent  Combining  Eqs. \eqref{eq5.22}, \eqref{eq5.21}, \eqref{eq5.33}, and  \eqref{eq5.34}
 we obtain
\begin{multline} \label{eq5.35}
\WLO^{disc}_{rig}(L)
    \sim \lim_{s \to 0} \sum_{(\alpha_i)_i \in \Lambda^m}  \bigl( \prod_i  m_{\chi_i}(\alpha_i) \bigr)
 \sum_{y \in I} \\
 \times      \int^{\sim}_{\ct} db \ \biggl[ \exp\bigl( - 2 \pi i k    \langle y, B(\sigma_0)  \rangle \bigr)  \bigl( \prod_{x}
1^{(s)}_{\ct_{reg}}(B(x)) \bigr) \\
\times  \bigl( \prod_i \exp(2 \pi i \langle \alpha_i, \Phi_i(B) \rangle )  \bigr) \Det^{disc}(B)
   \biggr]_{| B  = b + \tfrac{1}{k}      \sum_i \alpha_i    f_i }
    \end{multline}
   Apart from the remaining
$\int^{\sim}_{\ct} \cdots db$-integration
 (which will be taken care of in Step 4 below) we have now completed
the evaluation of the  $\int\nolimits_{\sim}  \cdots
\exp(iS^{disc}_{CS}(A^{\orth}_c,B)) (DA^{\orth}_c  \otimes DB)$-integral in  Eq.
\eqref{eq5.14}.

\begin{remark}   \label{rm4.1} \rm
 It is not difficult to see\footnote{this follows from Eq. \eqref{eq1_lastrmsec3}
 in Remark \ref{rm_last_sec3} above
 and the periodicity properties of the integrand in
 $\int^{\sim}_{\ct} \cdots db$ (for fixed $y$ and $\alpha_1, \ldots, \alpha_m$),
 cf.  Step 4 below} that
Eq. \eqref{eq5.35} also holds   if we (re)define $f_i$
using  the following normalization condition (instead of the normalization condition \eqref{eq5.28''} above):
 \begin{equation} \label{eq5.28'}  f_i(\sigma_0)=0
\end{equation}
\end{remark}

Since condition \eqref{eq5.28'} is technically more convenient than \eqref{eq5.28''}
we will use the latter normalization condition in  the following, i.e. we assume that
$f_i$ is defined as in \eqref{eq5.28} above in combination with \eqref{eq5.28'}.

\subsection{Step 3: Some simplifications}
\label{subsec5.3}

The next lemma will prove a claim made in Sec. \ref{subsec5.2} above
and it will also allow us to simplify the RHS of  Eq. \eqref{eq5.35} above.

\begin{lemma} \label{lem2_pre}
Assume that $ {\mathfrak l}^i_{\Sigma}  \in C_1(q\cK) \cong C^1(q\cK,\bR)$ is as in Eq. \eqref{eq_mathfrak_l_def} above. Then we have:
\begin{enumerate}
\item  There is a $f \in C^0_{\aff}(q\cK,\bR)$ such that
$ {\mathfrak l}^i_{\Sigma}  = m_{\bR} \cdot f $
and such that  $f$ is  constant on $U(\sigma_0) \cap \face_0(qK)$ (cf. the end of Sec. \ref{subsec4.10a} above).
Moreover, these properties determine $f$ uniquely up to an additive constant.

\item If $f$ is as in part i) of the present lemma then the map $f : \face_0(q\cK) \ni \sigma \mapsto f(\sigma) \in \bR$  is constant on $\arc(l^j_{\Sigma}) \cap \face_0(q\cK)$ and on
 $\arc(l^{'j}_{\Sigma}) \cap \face_0(q\cK) $ for all $j \le m$.
\end{enumerate}
\end{lemma}

\begin{proof}   From conditions (NCP)' and (NH)' on the simplicial ribbon link $L$ it follows that
 there are  three connected components $C_0$, $C_1$ and $C_2$ of $\Sigma \backslash (\arc(l^i_{\Sigma}) \cup \arc(l^{'i}_{\Sigma}))$. In the following we assume that these three connected components
 are given as in  Fig. \ref{fig_step3} below.

 \begin{figure}[h]
\begin{center}
  \includegraphics[height=6 cm,width=9 cm]{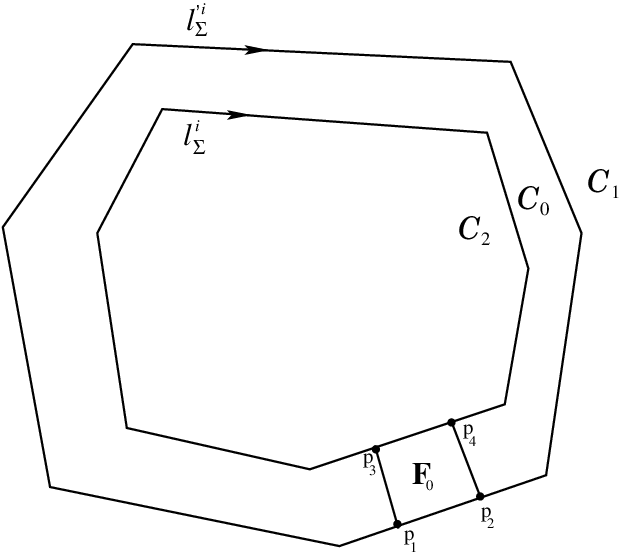}
\caption{}  \label{fig_step3}
\end{center}
\end{figure}

It follows from  Condition (FC2) in Sec. \ref{subsec5.0} that $ \face_0(qK) \cap C_0 = \emptyset$,
or, in other words, that
$\face_0(qK) \subset  \overline{C_1} \cup \overline{C_1}$.
Accordingly, the map  $f: \face_0(q\cK) \to \bR$ given by
\begin{equation} \label{eq_def_f}
f(p):=
 \begin{cases} c & \text{ if } p \in \overline{C_1} \\
c \pm 1 & \text{ if } p \in \overline{C_2} \\
\end{cases}
\end{equation}
for all $p \in \face_0(q\cK)$ is well-defined.
Here $c \in \bR$ is an arbitrary constant which will be kept fixed in the following
and  the sign  $\pm$ is ``$+$'' if for any $t \in [0,1]$ in which  $l^i_{\Sigma}$ is differentiable
the  normal vector $n = \star (\tfrac{d}{dt} l^i_{\Sigma}(t))$
``points into''  $C_0$ and ``$-$'' otherwise.

\medskip

Let us first verify that $f \in C^0_{\aff}(q\cK,\bR)$.
Let $F \in \face_2(q\cK)$ and let $p_1,p_2,p_3,p_4$ be the four vertices of $F$
enumerated such that $p_1$ is diagonal to $p_4$ (and therefore $p_2$ is diagonal to $p_3$).
We have to show that $f(p_1)+f(p_4) = f(p_2)+f(p_3)$  (cf. Eq. \eqref{eq_for_aff_B}).
If $F \subset \overline{C_1}$ or $F \subset \overline{C_2}$
this is obvious. If $F \subset \overline{C_0}$ (for example $F=F_0$ where $F_0$ is as in Fig. \ref{fig_step3}  above)
then Condition (FC3) implies that  $f(p_1)+f(p_4) = c + (c \pm 1) =  f(p_2)+f(p_3)$. \par

Recall from Sec. \ref{subsec4.9} that $\sigma_0 \in \face_0(q\cK)$ was chosen such that
$\sigma_0 \notin \Image(R^i_{\Sigma})$.
This implies that $\sigma_0 \in C_1$ or $\sigma_0 \in C_2$
and from the definition of $U(\sigma_0)$ we obtain
$U(\sigma_0) \subset \overline{C_1}$ or $U(\sigma_0) \subset \overline{C_2}$
so Eq. \eqref{eq_def_f} implies  that $f$ is constant on $U(\sigma_0) \cap \face_0(qK)$.\par
 The uniqueness part of the assertion follows by combining the definition of $m_{\bR}$
 with   the real analogue
 of Eq. \eqref{eq_obs1} in Sec. \ref{subsec4.10a} above and
 the fact that  $\star_{K}:C^1(K,\bR) \to  C^1(K,\bR)$
  is a bijection.

\smallskip

In order to conclude the proof of part i) of  Lemma \ref{lem2_pre}
we have to show that
\begin{equation} \label{eq_lemma3_crucial}
 {\mathfrak l}^i_{\Sigma}  = m_{\bR} \cdot f = \star_K(\pi(d_{q\cK}f))
 \end{equation}
Observe first that
$(d_{q\cK} f)(e) = 0$ unless the interior of $e \in \face_1(q\cK)$ is contained in $C_0$.
In the latter case we have  $(d_{q\cK} f)(e) =  \pm \sgn(e)$
where the sign $\pm$  is the same as in Eq. \eqref{eq_def_f} above and
where $\sgn(e) = 1$ if the (oriented) edge  $e$   ``points from'' the region $C_1$ to the region $C_2$ and $\sgn(e) = -1$ otherwise. \par

Next observe that for every $e \in \face_1(q\cK)$ whose interior is contained in $C_0$
there exists an index $k \le n$  such that\footnote{from the definition of $q\cK$ it follows
that for every edge $e$ in $q\cK$ exactly one endpoint is in
$\face_0(K_1 | K_2)$ and the other endpoint is either in $\face_0(K_1)$ or  in $\face_0(K_2)$}
\begin{equation} \label{eq_star_pi_e}  \star_K (\pi(\pm \sgn(e) \cdot e)) =
\begin{cases}  \bar{l}^{i(k)}_{\Sigma} & \text{ if $e$ has an endpoint in $\face_0(K_2)$} \\
  \bar{l}^{'i(k)}_{\Sigma} & \text{ if $e$ has an endpoint in $\face_0(K_1)$} \\
\end{cases}
\end{equation}
and that  the map\footnote{here $0$ is the zero element of $C_1(K)$}
$$ \psi:\{ e  \in \face_1(q\cK) \mid \text{the interior of $e$ is contained in $C_0$} \}  \to
 ( \{  \bar{l}^{i(k)}_{\Sigma} \mid  k \le n \} \cup \{  \bar{l}^{'i(k)}_{\Sigma} \mid  k \le n \}) \backslash \{0\} $$
given by $\psi(e) =  \star_K (\pi(\pm \sgn(e) \cdot e))$ for all $e \in \dom(\psi)$ is a bijection.

\smallskip

Finally, observe that the sum of the elements of the finite subset\footnote{in order to avoid confusion
recall that according to the standard convention
for sets we have $\{a,a\} = \{a\}$, $\{a,a,b,b\} = \{a,b\}$ etc.
Since $\bar{l}^{i(2k-1)}_{\Sigma} = \bar{l}^{i(2k)}_{\Sigma}$ for $k \le n/2$ this means that
$ \{  \bar{l}^{i(k)}_{\Sigma} \mid  k \le n \} =  \{  \bar{l}^{i(k)}_{\Sigma} \mid  k \le n, k \text{ odd} \}$,
which implies  that the sum of the elements of  $\{  \bar{l}^{i(k)}_{\Sigma} \mid  k \le n \}$
equals  $\sum_{k \le n, k \text{ odd}} \bar{l}^{i(k)}_{\Sigma} = \tfrac{1}{2} \sum_{k=1}^n \bar{l}^{i(k)}_{\Sigma}$.
A similar remark applies to the sum of the set of the elements of  $\{  \bar{l}^{'i(k)}_{\Sigma} \mid  k \le n \}$}
$ \{  \bar{l}^{i(k)}_{\Sigma} \mid  k \le n \} \cup \{  \bar{l}^{'i(k)}_{\Sigma} \mid  k \le n \}$
of $C_1(K) \subset C_1(q\cK)$ equals  $ {\mathfrak l}^i_{\Sigma}$,  cf. Eq. \eqref{eq_mathfrak_l_def} above.
From this  Eq. \eqref{eq_lemma3_crucial} follows.

\medskip

It remains  to show part ii) of  Lemma \ref{lem2_pre}.
Recall that according to Eq. \eqref{eq_def_f},  $f$ is constant on $\face_0(qK) \cap \overline{C_1}$ and
 also on $\face_0(qK) \cap \overline{C_2}$.
 From condition (FC1) and condition (FC2) it follows that for each $j \neq i$  the set
$\arc(l^j_{\Sigma})$ (considered as a subset of $\Sigma$)
 either lies entirely in $C_1$
or entirely in $C_2$.
Similarly,  it follows that for each $j \neq i$ the set
$\arc(l^{'j}_{\Sigma})$ either lies entirely in $ C_1$
or in $C_2$.
Since we also have $\arc(l^i_{\Sigma}) = \partial C_2 \subset \overline{C_2}$
and $\arc(l^{'i}_{\Sigma}) = \partial C_1 \subset \overline{C_1}$
 part ii) of  Lemma \ref{lem2_pre} now follows.
\end{proof}

 \begin{corollary} \label{obs3}  Let
 $B \in \cB(q\cK)$
   be of the form
\begin{equation} \label{eq5.37}
  B= b + \tfrac{1}{k}  \sum_{i=1}^m \alpha_i  f_i
 \end{equation}
with $b \in \ct$, $\alpha_i \in \Lambda$ and where $f_i$ is given by
Eq. \eqref{eq5.28} above in combination with \eqref{eq5.28'}.
 Then the map $\face_0(q\cK) \ni \sigma \mapsto B(\sigma) \in \ct$  is constant on
  $\arc(l^j_{\Sigma}) \cap \face_0(q\cK)$ and on
 $\arc(l^{'j}_{\Sigma}) \cap \face_0(q\cK) $ for all $j \le m$.
 \end{corollary}

Let us set
\begin{align} \label{eq5.36} \sigma_i & :=  \start l^{i(1)}_{\Sigma}   \in
\face_0(q\cK), \quad \quad
 \sigma'_i  :=  \start l^{'i(1)}_{\Sigma}   \in \face_0(q\cK)
\end{align}
According to Corollary  \ref{obs3}  we have
\begin{subequations}
\begin{equation} \label{eq5.38} B(\sigma_i) =
B(\start l^{i(k)}_{\Sigma}) \quad \forall k \le n
\end{equation}
\begin{equation}  B(\sigma'_i) =  B(\start l'^{i(k)}_{\Sigma}) \quad \forall k \le n
\end{equation}
\end{subequations}
for every $B$ of the form in Eq. \eqref{eq5.37}.
  Next observe that
\begin{equation} \label{eq5.39}  \eps_i:= \wind(l^{i}_{S^1})   =  \sum_k dt^{(N)}(l^{i(k)}_{S^1})
\end{equation}
 where $ \wind(l^{i}_{S^1})$ is the winding number of $l^{i}_{S^1}$.\par

 Combining  Eq. \eqref{eq5.35} and Eq. \eqref{eq5.17}   with
 Eqs. \eqref{eq5.38} -- \eqref{eq5.39} and taking into account
the normalization condition \eqref{eq5.28'} appearing at the end of Sec. \ref{subsec5.2}
   we obtain
\begin{multline} \label{eq5.41}
 \WLO^{disc}_{rig}(L) \sim \lim_{s \to 0}\sum_{(\alpha_i)_i \in \Lambda^m}  \bigl( \prod_i  m_{\chi_i}(\alpha_i) \bigr) \sum_{y \in I}  \\
 \times  \biggl( \int^{\sim}_{\ct} db \   \biggl[
 \exp\bigl( - 2 \pi i k  \langle y, b  \rangle \bigr)
   \bigl( \prod_{x} 1^{(s)}_{\ct_{reg}}(B(x)) \bigr) \\
  \times  \bigl(   \prod_{\iprime} \exp( \pi i \eps_{\iprime}  \langle \alpha_{\iprime},
 B(\sigma_{\iprime}) + B(\sigma'_{\iprime})  \rangle )
\Det^{disc}(B)  \bigr)  \biggr]_{| B  = b + \tfrac{1}{k}
    \sum_i \alpha_i f_i } \biggr)
\end{multline}
\noindent Setting
 \begin{multline} \label{eq5.49}
 F^{(s)}_{(\alpha_i)_i}(b) :=  \bigl[ \bigl(  \prod_{x}
1^{(s)}_{\ct_{reg}}(B(x)) \bigr) \\
 \times \bigl(  \prod_{\iprime} \exp(   \pi i \eps_{\iprime}  \langle \alpha_{\iprime},
 B(\sigma_{\iprime}) + B(\sigma'_{\iprime})  \rangle )  \Det^{disc}(B)  \bigr) \bigr]_{| B  = b + \tfrac{1}{k}
   \sum_i \alpha_i f_i }
\end{multline}
 we can rewrite    Eq. \eqref{eq5.41} as
\begin{equation}  \label{eq5.50} \WLO^{disc}_{rig}(L)
 \sim \lim_{s \to 0} \sum_{(\alpha_i)_i \in \Lambda^m}
 \bigl( \prod_i  m_{\chi_i}(\alpha_i) \bigr)   \sum_{y \in I}
  \int^{\sim}_{\ct} db \   e^{ - 2\pi  i k \langle y,  b \rangle}   F^{(s)}_{(\alpha_i)_i}(b)
\end{equation}

\subsection{Step 4: Performing the  remaining limit procedures
 $\int^{\sim} \cdots db$,  $\sum_{y \in I}$,   and $s \to 0$
 in Eq. \eqref{eq5.50}}
\label{subsec5.4}

\noindent  i) Let us first rewrite the  $\int^{\sim} \cdots db$ integral.
The crucial observation is that for fixed $y \in I$, $s > 0$, and $(\alpha_i)_i \in \Lambda^m$
the function
$\ct \ni b \mapsto e^{ - 2\pi  i k \langle y,  b \rangle}   F^{(s)}_{(\alpha_i)_i}(b)
\in \bC$ is invariant under all translations of the form
$b \mapsto b + x$ where $x \in I = \ker(\exp_{| \ct}) \cong \bZ^{\dim(\ct)}$.

Indeed, for all $b \in \ct$ and $x \in I$ we have
\begin{subequations}  \label{eq5.51}
\begin{align}
1^{(s)}_{\ct_{reg}}(b + x) & = 1^{(s)}_{\ct_{reg}}(b)\\
e^{ 2\pi  i \eps \langle \alpha,  b + x \rangle} & = e^{ 2\pi  i \eps \langle \alpha,  b \rangle}
\quad \text{  for all $\alpha \in \Lambda$, $\eps \in \bZ$} \\
\det\nolimits^{1/2}(1_{\ck} - \exp(\ad(b + x))_{| \ck})  & =  \det\nolimits^{1/2}(1_{\ck} - \exp(\ad(b))_{| \ck})\\
e^{ - 2\pi  i k \langle y,  b + x \rangle} & = e^{ - 2\pi  i k \langle y,  b\rangle} \quad \text{
 for all $y \in I$}
\end{align}
\end{subequations}

    The second  equation  follows because the assumption that
 $G$ is simply-connected implies that
\begin{equation} \label{eq5.62} I =  \Gamma
\end{equation}
where $\Gamma \subset \ct$ is the lattice generated by the real coroots and, by definition,
$\Lambda$ is the lattice dual to $\Gamma$.
 The first of these four equations   follows from the assumption in Sec. \ref{subsec4.6}
 that $1^{(s)}_{\ct_{reg}}$ is invariant under $\cW_{\aff}$ (and using again Eq. \eqref{eq5.62}).
 The third  equation follows because (cf. Eq. \eqref{eq_def_det1/2} above)
\begin{align*}
&  \det\nolimits^{1/2}\bigl(1_{{\ck}}-\exp(\ad({b+x}))_{|{\ck}}\bigr) =
 \prod_{{\alpha} \in {\cR_+}}  \bigl( 2  \sin( \pi  \langle \alpha, b +x \rangle) \bigr) \\
& = (-1)^{\sum_{\alpha \in \cR_+} \langle \alpha,x\rangle }   \prod_{{\alpha} \in {\cR_+}}  \bigl( 2  \sin( \pi  \langle \alpha, b \rangle) \bigr)\\
& \overset{(*)}{=}  \prod_{{\alpha} \in {\cR_+}}  \bigl( 2  \sin( \pi  \langle \alpha, b \rangle) \bigr)
= \det\nolimits^{1/2}\bigl(1_{{\ck}}-\exp(\ad({b}))_{|{\ck}}\bigr)
\end{align*}
where in step $(*)$ we used $\sum_{\alpha \in \cR_+} \langle \alpha,x\rangle  = 2 \langle \rho, x \rangle
\in 2 \bZ$ since $\rho \in \Lambda$.
 Finally, in order to see that the fourth equation holds, observe
that because of \eqref{eq5.62}  it is enough to show that
\begin{equation} \label{eq_CartanMatrix} \langle \Check{\alpha}, \Check{\beta} \rangle \in \bZ \quad \text{ for all coroots  $\Check{\alpha}, \Check{\beta}$}
\end{equation}
According to the general theory of semi-simple Lie algebras  we have
 $2\tfrac{\langle \Check{\alpha}, \Check{\beta} \rangle}{\langle \Check{\alpha}, \Check{\alpha} \rangle}
\in \bZ$. Moreover,  there are at most two different (co)roots lengths
and  the quotient between the square lengths of the long and short coroots is either 1, 2, or 3.
Since  the normalization of $\langle
\cdot,  \cdot \rangle$ was chosen such that
$\langle \Check{\alpha}, \Check{\alpha} \rangle = 2$ holds
if $\Check{\alpha}$ is a short coroot we therefore have
$\langle \Check{\alpha}, \Check{\alpha} \rangle/2 \in \{1,2,3\}$
and \eqref{eq_CartanMatrix} follows.

  \medskip

From Eqs. \eqref{eq5.49},  \eqref{eq5.51},  and Eq. \eqref{eq5.66} below
 we conclude that
 $\ct \ni b \mapsto e^{ - 2\pi  i k \langle y,  b \rangle}   F^{(s)}_{(\alpha_i)_i}(b)
\in \bC$ is indeed $I$-periodic
and we can therefore apply  Eq. \eqref{eq2_lastrmsec3} in Remark \ref{rm_last_sec3} above  and  obtain
\begin{equation}  \label{eq5.57}
  \int^{\sim} db \
 e^{ - 2\pi  i k \langle y,  b \rangle}  F^{(s)}_{(\alpha_i)_i}(b)
 \sim  \int_{Q} db \  e^{ - 2\pi  i k \langle y,  b \rangle}  F^{(s)}_{(\alpha_i)_i}(b)
\end{equation}
where on the RHS  $\int_{Q} \cdots db$ is now an ordinary integral
and where  we have set
\begin{equation}  \label{eq5.54}
Q:= \{ \sum_i \lambda_i e_i \mid \lambda_i \in (0,1) \text{ for all $i \le m$}   \} \subset \ct,
\end{equation}
Here  $(e_i)_{i \le m}$ is an (arbitrary) fixed basis of $I$.\par

According to Eq. \eqref{eq5.57} we can now rewrite Eq. \eqref{eq5.50} as
\begin{equation}  \label{eq5.50b} \WLO^{disc}_{rig}(L)
 \sim \lim_{s \to 0} \sum_{(\alpha_i)_i \in \Lambda^m}
 \bigl( \prod_i  m_{\chi_i}(\alpha_i) \bigr)   \sum_{y \in I}
  \int_Q db \   e^{ - 2\pi  i k \langle y,  b \rangle}   F^{(s)}_{(\alpha_i)_i}(b)
\end{equation}

\medskip

\noindent ii) We can now perform the infinite sum $\sum_y$
and the $\int \cdots db$-integral in Eq. \eqref{eq5.50b}: \par

First recall  that, due to Eq. \eqref{eq5.62} above and the definition of $\Lambda$,
 $\Lambda$ is dual to $I$.
According to the (rigorous)  Poisson summation formula for distributions we therefore have
\begin{equation} \label{eq5.61}
 \sum_{y \in I}   e^{ - 2\pi  i k \langle y,  b \rangle}
  = c_{\Lambda} \sum_{x \in \tfrac{1}{k} \Lambda} \delta_x(b)
  \end{equation}
where $\delta_x$ is the  delta distribution in $x \in \ct$ and $c_{\Lambda}$ a constant depending on the lattice $\Lambda$.
Let us now apply  Eq. \eqref{eq5.61} to the RHS  of Eq. \eqref{eq5.50b} above.
In order to see that this is possible note first that not only
 $F^{(s)}_{(\alpha_i)_i}$ is smooth
 but also the product $1_{Q}  F^{(s)}_{(\alpha_i)_i}$
 because  $\partial Q \subset \ct \backslash \ct_{reg}$
and because $F^{(s)}_{(\alpha_i)_i}$ vanishes on an open neighborhood of the set $\ct \backslash \ct_{reg}$
(cf.  the condition $\supp(1^{(s)}_{\ct_{reg}}) \subset \ct_{reg}$ in Sec. \ref{subsec4.6}
  so, according to the definition of  $F^{(s)}_{(\alpha_i)_i}$ and Eq. \eqref{eq5.28'},
   there is a factor $1^{(s)}_{\ct_{reg}}(b)$ appearing in $F^{(s)}_{(\alpha_i)_i}(b)$).
  Moreover, since $Q$ is bounded   $1_{Q}  F^{(s)}_{(\alpha_i)_i}$
 has compact support.  Thus we can
indeed apply Eq. \eqref{eq5.61} to the RHS  of Eq. \eqref{eq5.50b} above and
  we then obtain
\begin{equation}\label{eq5.63} \WLO^{disc}_{rig}(L)  \sim \lim_{s \to 0} \sum_{(\alpha_i)_i \in \Lambda^m}
\bigl( \prod_i  m_{\chi_i}(\alpha_i) \bigr)
 \sum_{b \in \tfrac{1}{k} \Lambda} 1_{Q}(b)  F^{(s)}_{(\alpha_i)_i}(b)
 \end{equation}

\noindent iii) Finally, let us also perform the $s \to 0$ limit.
Taking into account that $1^{(s)}_{\ct_{reg}} \to 1_{\ct_{reg}}$
 pointwise we  obtain from Eq.  \eqref{eq5.63} and Eq. \eqref{eq5.49} after  the change of variable
$b \to k b =: \alpha_0$
\begin{align} \label{eq5.64} & \WLO^{disc}_{rig}(L)   \sim \sum_{\alpha_0, \alpha_1, \ldots, \alpha_m \in \Lambda}
 1_{k Q}(\alpha_0)
 \bigl( \prod_{i=1}^m  m_{\chi_i}(\alpha_i) \bigr)     \nonumber \\
&  \quad \quad \times \bigl[ \bigl( \prod_{x}
1_{\ct_{reg}}(B(x)) \bigr)  \prod_{\iprime=1}^m \exp(    \pi i \eps_{\iprime} \langle \alpha_{\iprime},
 B(\sigma_{\iprime}) + B(\sigma'_{\iprime})  \rangle )
\Det^{disc}(B)  \bigr]_{| B  =  \tfrac{1}{k} ( \alpha_0 +    \sum_i \alpha_i f_i)}
\end{align}

\subsection{Step 5: Rewriting $\Det^{disc}(B)$ in Eq. \eqref{eq5.64}}
\label{subsec5.5}

In Steps 1--4 we have reduced the original ``path integral''
expression for $\WLO^{disc}_{rig}(L)$ to a ``combinatorial expression'', i.e. an
  expression which does not involve any limit procedure. Let
us now have a closer look at $\Det^{disc}(B)$, cf. Eq. \eqref{eq_def_Detdisc} in Step 1 above.

\smallskip

From assumptions (NCP)' and (NH)' above it follows
 that the set  of connected components of  $\Sigma \backslash
  \bigl(\bigcup_j  \arc(l^j_{\Sigma}) \bigr)$
  has exactly $m+1$ elements, which we will denote by
$Y_0, Y_1, \ldots Y_{m}$ in the following.
Moreover,  it follows that
also the set  of connected components of  $\Sigma \backslash \bigcup_j \Image(R^j_{\Sigma}) =
 \Sigma \backslash \bigcup_j \bigl(O_j \cup \arc(l^j_{\Sigma}) \cup \arc(l^{'j}_{\Sigma}) \bigr)$
  has exactly $m+1$ elements, which we will denote by
$Z_0, Z_1, \ldots Z_{m}$.
(Here $O_j \subset \Sigma$,  $j \le m$, denotes  the open ``region between''
$\arc(l^j_{\Sigma})$ and $\arc(l^{'j}_{\Sigma})$, cf. condition (FC2) in Sec. \ref{subsec5.0} above.)\par

In the following we assume without loss of generality that
the numeration of the $Z_i$, $i \le m$, was chosen such that
$$Z_i \subset Y_i \quad \forall i \in \{0,1,\ldots, m\}$$
holds. For later use we remark that\footnote{here $\overline{Z_i}$ denotes the closure of  $Z_i$}
 $\{ \overline{Z_i} \mid 0 \le i \le m\} \cup \{ O_i \mid i \le m\}$,
is a partition of $\Sigma$.
Thus  condition (FC2) implies  that
\begin{equation} \label{eq_disj_Ver}
\face_0(q\cK) = \bigsqcup\nolimits_{i=0}^m (\face_0(q\cK) \cap \overline{Z_i})
\end{equation}

Observe also that
\begin{equation} \label{eq_disj_Ver2}
\face_0(q\cK) = \face_0(K_1) \sqcup \face_0(K_1|K_2) \sqcup \face_0(K_2)
\end{equation}
(Here and in the following we use the notation  $S \sqcup T$ for the disjoint union of two sets $S$ and $T$).

\smallskip

In the following we assume that  $B \in \cB(q\cK)$  is of the form
\begin{equation}  \label{eq5.37b}
  B=  \tfrac{1}{k}\bigl( \alpha_0 + \sum_i \alpha_i f_i\bigr),
  \quad \text{with  $\alpha_0, \ldots, \alpha_m \in \Lambda$}
 \end{equation}
and with $f_i$ as in Eq. \eqref{eq5.28} above in combination with \eqref{eq5.28'}.

\begin{lemma} \label{lem2a} If $B \in \cB(q\cK)$ is of the form \eqref{eq5.37b}
then the restriction of $B: \face_0(q\cK) \to \ct$
to $\face_0(q\cK) \cap \overline{Z_i}$ is constant for each $i$.
\end{lemma}
\begin{proof} This lemma is a generalization of Corollary  \ref{obs3} in Sec. \ref{subsec5.3} above
and follows easily from (the proof of) Lemma \ref{lem2_pre} above.
\end{proof}

 In the following we  set
 $B(Y_i):=B(Z_i) :=  B(x)$ for any $x \in Z_i \cap \face_0(q\cK)$
(according to  Lemma \ref{lem2a} the value of $B(Y_i) = B(Z_i)$ does not depend
 on the choice of $x \in Z_i \subset Y_i$).

\begin{lemma} \label{lem2} For every $B \in \cB(q\cK)$
of the form \eqref{eq5.37b} and fulfilling  $\prod_{x}
1_{\ct_{reg}}(B(x)) \neq 0$   we have
\begin{align} \label{eq5.70}
\Det^{disc}(B) &  \sim \prod_{i=0}^m
  \det\nolimits^{1/2}\bigl(1_{\ck}-\exp(\ad(B(Y_i)))_{|\ck}\bigr)^{\chi(Y_i)}
 \end{align}
 where $\chi(Y_i)$  is the Euler characteristic of $Y_i$.
 \end{lemma}
\begin{proof}
From the definition of $ \Det^{disc}(B)$ in Eq. \eqref{eq_def_Detdisc}, from Eq. \eqref{eq_def_ZBdisc}
and  Eq. \eqref{eq_ Det_disc_FP1}  in Sec. \ref{sec4} above
it follows that
\begin{equation} \label{eq5.66} \Det^{disc}(B) \sim  \frac{ \prod_{x \in \face_0(K_1)}
\det\nolimits^{1/2}\bigl(1_{\ck}-\exp(\ad(B(x)))_{|\ck}\bigr) \prod_{x \in
\face_0(K_2)} \det\nolimits^{1/2}\bigl(1_{\ck}-\exp(\ad(B(x)))_{|\ck}\bigr)}{
\prod_{x \in \face_0(K_1|K_2)}
\det\nolimits^{1/2}\bigl(1_{\ck}-\exp(\ad(B(x)))_{|\ck}\bigr)}
\end{equation}
(observe that the expression on the RHS  is well-defined since by assumption $\prod_{x}
1_{\ct_{reg}}(B(x)) \neq 0$, which implies that the denominator  is non-zero).\par

According to Lemma \ref{lem2a} and
Eqs. \eqref{eq_disj_Ver} and \eqref{eq_disj_Ver2}
 it is enough to prove that for each $i \in \{0,1,\ldots, m\}$ we have
\begin{equation} \label{eq_chi1} \chi(Y_i) = \# ( \face_0(K_1) \cap \overline{Z_i}) - \# ( \face_0(K_1|K_2) \cap \overline{Z_i})
+ \# ( \face_0(K_2) \cap \overline{Z_i})
\end{equation}
Clearly, $\chi(Y_i)=\chi(\overline{Y_i})$
where $\overline{Y_i}$ is  the closures of $Y_i$.
Moreover,  $\overline{Y_i}$ is a  subcomplex of the CW complex $K_1 = \cK = (\Sigma,\cC)$
so setting $Cell_p(\overline{Y_i})  :=
 \{ \sigma \in Cell_p(K_1) \mid \sigma \subset \overline{Y_i}\}$
where $Cell_p(K_1)$ is the set of (open) p-cells of $K_1$   we obtain
\begin{equation}  \label{eq_chi2} \chi(\overline{Y_i}) = \sum_{p=0}^2 (-1)^p \# Cell_p(\overline{Y_i}) \overset{(*)}{=}
\# ( \face_0(K_1) \cap \overline{Y_i}) - \# ( \face_0(K_1|K_2) \cap \overline{Y_i})
+ \# ( \face_0(K_2) \cap \overline{Y_i})
\end{equation}
(step $(*)$ follows by taking into account the natural 1-1-correspondences
$Cell_0(K_1) \leftrightarrow \face_0(K_1)$,  $Cell_1(K_1) \leftrightarrow \face_0(K_1 | K_2)$,
and $Cell_2(K_1) \leftrightarrow \face_0(K_2)$).\par

In order to complete the proof of Lemma \ref{lem2} it is therefore
enough to show that the RHS  of Eq. \eqref{eq_chi1} and the RHS  of Eq. \eqref{eq_chi2} coincide.\par

In order to see this observe that  for each $0 \le i \le m$
there is $J \subset \{0,1,\ldots,m\}$ such that
$$\overline{Y_i} = \overline{Z_i} \sqcup  \bigsqcup_{j \in J} \bigl( O_j
\sqcup \arc(l^{'j}_{\Sigma}) \bigr)$$
so our claim follows from
\begin{subequations}
\begin{multline}
\# ( \face_0(K_1) \cap \arc(l^{'j}_{\Sigma}) ) - \# ( \face_0(K_1|K_2) \cap \arc(l^{'j}_{\Sigma}))
+ \# ( \face_0(K_2) \cap \arc(l^{'j}_{\Sigma})) \\
 =  - \# ( \face_0(K_1|K_2) \cap \arc(l^{'j}_{\Sigma})) +  \# ( \face_0(K_2) \cap \arc(l^{'j}_{\Sigma}) )
= 0  \end{multline}
and from
\begin{equation}
\# ( \face_0(K_1) \cap O_j) - \# ( \face_0(K_1|K_2) \cap  O_j)
+ \# ( \face_0(K_2) \cap  O_j) = \# \emptyset -   \# \emptyset +  \# \emptyset = 0
\end{equation}
(cf.   condition (FC2) and Eq. \eqref{eq_disj_Ver2}).
\end{subequations}
\end{proof}

\begin{lemma} \label{lem3} For every $B \in \cB(q\cK)$
of the form \eqref{eq5.37b} we have
\begin{equation} \label{eq_lem3}
 \prod_{x \in \face_0(q\cK)} 1_{\ct_{reg}}(B(x)) = \prod_{i=0}^m 1_{\ct_{reg}}(B(Y_i))
\end{equation}
\end{lemma}
\begin{proof} The assertion follows from Lemma \ref{lem2a} and Eq. \eqref{eq_disj_Ver} above.
\end{proof}

\noindent Combining  Eq. \eqref{eq5.64}  with Lemma
\ref{lem2} and Lemma \ref{lem3} we arrive at
\begin{align}  \label{eq5.89} \WLO^{disc}_{rig}(L)  &  \sim
 \sum_{\alpha_0, \alpha_1, \ldots, \alpha_m \in \Lambda} 1_{k Q}(\alpha_0)
\bigl( \prod_{i=1}^m  m_{\chi_i}(\alpha_i) \bigr) \nonumber \\
&  \quad \quad \quad \times \biggl[ \prod_{i=0}^m 1_{\ct_{reg}}(B(Y_i))
 \det\nolimits^{1/2}\bigl(1_{\ck}-\exp(\ad(B(Y_i)))_{|\ck}\bigr))^{\chi(Y_i)}\nonumber \\
&  \quad \quad \quad \times   \prod_{\iprime=1}^m \exp(   \pi i \eps_{\iprime}  \langle \alpha_{\iprime},
 B(\sigma_{\iprime}) + B(\sigma'_{\iprime})  \rangle )  \biggr]_{|
B  =  \tfrac{1}{k} ( \alpha_0 +    \sum_i \alpha_i f_i)}
\end{align}

\begin{remark} \label{rm_Step5_full_ribbon} \rm
What would  happen if we had worked with ``full ribbons'' $\bar{R}_j$ instead of ``half ribbons'' $R_j$
(cf. Remark \ref{rm_full_ribbons} and Remark \ref{rm_full_ribbons2} above)? In this case
the RHS of Eq. \eqref{eq_mathfrak_l_def} above and therefore also  the elements $f_i$ of $C^0(q\cK,\bR)$
given by Eq. \eqref{eq5.28}
would have different values, which would make it necessary to  redefine
 the  sets $Z_i$ appearing  above in a suitable way\footnote{recall
 that the ``old'' sets $Z_i$ are the connected components of $\Sigma \backslash \bigcup_j \Image(R^j_{\Sigma})$
 where each $R^j_{\Sigma}$ is the (reduced) projection of the ``half ribbon'' $R_j$
 in $q\cK \times \bZ_{\bN}$.  The ``new'' sets $Z_i$ will be the connected components of
 $\Sigma \backslash \bigcup_j \Image(\bar{R}^j_{\Sigma})$
 where each $\bar{R}^j_{\Sigma}$ is the (reduced) projection of the ``full ribbon'' $\bar{R}_j$
 in $\cK \times \bZ_{\bN}$}.

\begin{enumerate}
\item  Lemma \ref{lem2} above would then still be true. In fact,
the proof would be simpler. Even more importantly, the structure of the
proof of the new version of Lemma \ref{lem2} (or rather, its $BF$-theoretic analogue)
  is exactly what is needed when trying to obtain a result like Eq. \eqref{eq_maintheorem_BF} below\footnote{recall that we do not expect that Theorem  \ref{main_theorem} can  be generalized
successfully to the case of general ribbon links unless we make a transition to the $BF_3$-theoretic
 point of view, cf. the beginning of Sec. \ref{subsec7.1} below}
for general ribbon links.

\item On the other hand,  for the ``new'' definition of the sets $Z_i$ mentioned above,
Eq. \eqref{eq_disj_Ver} would  no longer
 be true and  Lemma \ref{lem3} would no longer hold unless we insert additional indicator functions
 on the RHS of Eq. \eqref{eq_lem3}.
  It is still possible that -- by exploiting suitable algebraic identities -- one can
   recover Eq. \eqref{eq5.89} after all (in spite of the additional indicator functions
    appearing in Eq. \eqref{eq_lem3}).   If Eq. \eqref{eq5.89}  cannot be recovered, which is likely, then one can bypass this complication by using an additional regularization procedure.
    Since at the moment it is not clear whether such an additional regularization procedure
    is indeed necessary or not we decided to work only with half ribbons until now.
     We will begin to work with  full ribbons in  Sec. \ref{sec7} below.
  \end{enumerate}

\end{remark}

\subsection{Step 6: Comparison of $\WLO_{rig}(L)$ with the shadow invariant $|L|$}
\label{subsec5.6}

From the computations in Sec. 5 in \cite{HaHa} it follows that the RHS  of Eq. \eqref{eq5.89} above
coincides with the shadow invariant $|L|$ (associated to $\cG$ and $k$) up to a multiplicative constant.
For the convenience of the reader we will briefly sketch this derivation.
In the following we will use the notation of part \ref{appB} and \ref{appA}  of the Appendix.

\smallskip

For $\alpha_1, \ldots, \alpha_m \in \Lambda$  and
 $\alpha_0 \in \Lambda \cap  k Q$
set
 \begin{equation} \label{eq5.90} B  :=  \tfrac{1}{k}(\alpha_0 +   \sum_i \alpha_i f_i)
\end{equation}
and introduce  the function
$\varphi: \{Y_0,Y_1, \ldots, Y_m\} \to  \Lambda$ by
\begin{equation} \label{eq5.91} \varphi(Y):=kB(Y) - \rho \quad \quad \forall Y \in \{ Y_0,Y_1, \ldots, Y_m \}
\end{equation}
One can show that then (cf. Sec. 5 in \cite{HaHa})
 \begin{subequations} \label{eq5.92}
\begin{align}
\det\nolimits^{1/2}\bigl(1_{\ck}-\exp(\ad(B(Y)))_{|\ck}\bigr) & \sim \dim(\varphi(Y)) \\
\prod_{\iprime} \exp(  \pi i \eps_{\iprime}  \langle \alpha_{\iprime},
 B(\sigma_{\iprime}) + B(\sigma'_{\iprime})  \rangle )& = \prod_Y
 \exp(\tfrac{\pi i}{{k}} \langle \vf(Y),\vf(Y) +2\r \rangle )^{\gleam(Y)}
\end{align}
\end{subequations}

Let $P$ be the unique Weyl alcove which is contained in the Weyl chamber $\CW $ fixed in Sec. \ref{subsec2.1}
and which has $0 \in \ct$ on its boundary.
Moreover, let $\Lambda^k_+$ and  $\cW_k  \cong \cW_{\aff}$ be as in
part \ref{appB} of the Appendix and let $col(L)=
(\Lambda^k_+)^{\{Y_0,Y_1, \ldots, Y_m\}} $ (the set of ``area
colorings''). From  the relation
$\Lambda^k_+ = \Lambda \cap (P k - \rho)$,
the bijectivity of the map $\theta:P \times \cW_{\aff} \ni (b,\sigma) \mapsto \sigma(b) \in  \ct_{reg}$
and the fact that for a suitable finite subset $W$ of $\cW_{\aff} (\cong \cW_k)$ we have $\theta(P \times W) = Q \cap \ct_{reg}$  it follows  that there is a natural  1-1-correspondence
 between the set $col(L) \times W \times \cW_k^{\{Y_1, \ldots, Y_m\}}$
 and the set of those $B$ which are of the form in Eq. \eqref{eq5.90}
above (with $\alpha_0 \in \Lambda \cap  k Q$ and $\alpha_1, \ldots, \alpha_m \in \Lambda$)
 and which have the extra property that $\prod_{Y}
1_{\ct_{reg}}(B(Y)) =1$.\par

Using this and Eq. \eqref{eqA.7} below plus a suitable
 symmetry argument based on  the group $\cW_k \cong \cW_{\aff}$ (cf. the proof of Theorem 5.1 in \cite{HaHa})
one then arrives at\footnote{the multiplicities  $m_{\chi_i}(\alpha_i)$, $i \le m$,
appearing in Eq. \eqref{eq5.89}  lead to the fusion coefficients  $N_{ \mu \nu}^{\lambda}$
appearing in Eq. \eqref{eq5.93} below, cf. the RHS of Eq. \eqref{eqA.7}}
 \begin{multline} \label{eq5.93}
\WLO^{disc}_{rig}(L) \sim \sum_{\vf\in col(L)} \biggl( \prod_{i}
N_{\gamma(l_i)
\varphi(Y^+_{i})}^{\varphi(Y^-_{i})} \biggr) \\
\times \biggl(  \prod_Y \dim(\vf(Y))^{ \chi(Y)}
  \exp(\tfrac{\pi i}{k} \langle \vf(Y),\vf(Y) +2\r \rangle )^{\gleam(Y)} \biggr) = |L|
  \end{multline}

\bigskip

If we  apply Eq. \eqref{eq5.93} to the empty link $\emptyset$ instead of $L$
and take into account that\footnote{this follows easily from Eq. \eqref{eqA.8} below
after taking into account that the set $\Lambda^k_+$ is not empty if  $k \ge \cg$, cf. Remark \ref{rm_app0}} $|\emptyset| \neq 0$ if $k \ge \cg$
and that the symbol $\sim$  denotes equality up to a multiplicative non-zero constant
independent of $L$ we see that also $\WLO^{disc}_{rig}(\emptyset) \neq 0$ for $k \ge \cg$.
Accordingly, $\WLO_{rig}(L)$
is then well-defined and Eq. \eqref{eq5.93} implies that indeed
$$\WLO_{rig}(L) = \frac{\WLO^{disc}_{rig}(L)}{\WLO^{disc}_{rig}(\emptyset)} = \frac{|L|}{ |\emptyset|}$$

\section{A brief comment regarding general simplicial ribbon links}
\label{sec6}

Let us make some  comments
regarding the question if it is possible
to generalize the computations above to general simplicial ribbon links.
The crucial step will be the evaluation of the integral
\begin{equation} \label{eq_Outlook1}
 \int\nolimits_{\sim}   \prod_{i=1}^m  \Tr_{\rho_i}\bigl( \Hol^{disc}_{R_i}(\Check{A}^{\orth} + A^{\orth}_c,   B)\bigr) \exp(iS^{disc}_{CS}(\Check{A}^{\orth},B))
  D\Check{A}^{\orth}
\end{equation}
 for given $B$ and $A^{\orth}_c$, cf. Eq. \eqref{eq5.1}  above.
 For general simplicial ribbon links this is considerably more difficult
than in Sec. \ref{subsec5.1} above.
The good news is that the evaluation of the integral \eqref{eq_Outlook1}
 can be reduced
to the computation of the ``2-clusters''
\begin{equation} \label{eq_2cluster} \int\nolimits_{\sim}  \bigl\{ \rho_i(\exp(\sum_a T_a Y^{i,a}_{k})) \otimes \rho_{i'}(\exp(\sum_{a'} T_{a'} Y^{i',a'}_{k'})) \bigr\} \exp(iS^{disc}_{CS}(\Check{A}^{\orth},B))
  D\Check{A}^{\orth} \text{ \ $\in$ }  \End(V_i) \otimes  \End(V_{i'})
\end{equation}
 for the few $i,i' \le m$ and $k,k ' \le n$ for which
$\star_K \bar{l}^{i(k)}_{\Sigma} = \pm \bar{l}^{'i'(k')}_{\Sigma}$.
Here  $Y^{i,a}_{k}$ and $Y^{i',a'}_{k'}$ are as in Eq. \eqref{eq5.4} above
 and  $V_i$, $V_{i'}$ are the representation spaces of $\rho_i$ and $\rho_{i'}$, cf. Sec. \ref{subsec5.1}.
The integral in \eqref{eq_Outlook1} above can be expressed by these ``2-clusters''
by a similar formula as Eq. (6.4) in  \cite{Ha2} (cf. also Sec. 5.3 in \cite{Ha4} and \cite{Lab,CCFM}).\par

The explicit formula for  $\WLO_{rig}(L)$ for general $L$ which one obtains in this way
 should again be a sum over the set of  ``area
colorings'' $\vf$, but this time every summand
will contain an extra factor involving a product $\prod_{x \in
V(L)} \cdots $.
One could hope that this factor  coincides
with the factor $|L|_4^{\vf}$ (cf. part \ref{appA} of the Appendix
for the notation used here).

\smallskip

In order to evaluate the chances for this being the case
we can\footnote{in fact, the heuristic equation Eq. \eqref{eq2.48}
was only derived for simply-connected compact groups $G$ and therefore does not include the
case $G=U(1)$ or $G=U(1) \times U(1)$. However, it is not difficult to see that  for $G=U(1) \times U(1)$ and
$(k_1,k_2) = (k, -k)$ an analogue of Eq. \eqref{eq_WLO_BF}  below can be derived,
cf. Remark \ref{rm_7.2a} below.}
consider the case of Abelian structure group $G=U(1)$.
The computations are then  analogous to those appearing in
 the continuum setting in Secs 5.1 and 6.1 in \cite{Ha4} (which led
 to the correct result). However, in these computations there is one crucial difference
 in comparison to the computations in the continuum computations:
 there are several factors of  $1/2$, coming from the RHS  of Eq. \eqref{eq4.21} above,
 which ``spoil'' the final result. So ultimately
  we do {\it not} recover the (correct)  expressions which appeared in the  continuum setting.
This complication can be resolved\footnote{observe that in contrast to the RHS of Eq. \eqref{eq4.21} above
not all of the summands in the exponential on the RHS of Eq. \eqref{eq_Hol_disc_BF}
 below are multiplied with a weight factor $1/2$}
by  making the  transition to the ``$BF_3$-theory point of view''.

\begin{remark} \label{rm_sec6_full_ribbon} \rm As mentioned in Sec. \ref{subsec5.5} above (cf. also Sec. \ref{sec8} below) there is something else we have to do in order to have a reasonable chance of finding a generalization
of Theorem \ref{main_theorem} to the case of general ribbon links. Instead of working with ``half ribbons'' we should
rather work with ``full ribbons''. In Remark \ref{rm_Step5_full_ribbon} above we gave already one argument in favor of this claim. Here is another argument:\par

Let $L=(R_1,R_2, \ldots,R_m)$ be a simplicial ribbon link in $q\cK \times \bZ_N$ which does {\it not} fulfill condition (NCP)' above. More precisely, we assume that  for some $i \le m$ and $i' \le m$ the $\Sigma$-projections of $R_{i}$ and $R_{i'}$ intersect each other in a 2-face $F \in \face_2(q\cK)$.
 Let us study  the corresponding ``2-cluster'' given   as in Eq. \eqref{eq_2cluster} above.
Observe that the $\Sigma$-projections of the boundary loops of  $R_{i}$ and $R_{i'}$
intersect in four points\footnote{recall from Sec. \ref{subsec4.3}  that  each of the two half ribbons $R_{i}$ and $R_{i'}$  induces a pair of loops which will contribute  to  $\Hol^{disc}_{R_j}(A^{\orth},B)$
 (with $j = i,i'$) and therefore to the 2-cluster; the $\Sigma$-projections
    of these two pairs of loops intersect in four points, which are
    simply the four vertices of the intersection 2-face $F$ mentioned above} which contribute  to the corresponding 2-cluster.
 However, these four intersection points are not treated equally.
 Two of these intersection points do {\it not} ``detect'' an interaction\footnote{more precisely,
 the corresponding ``covariance expression'', ie the term analogous to the expression
 in Eq. \eqref{eq5.9} in Sec. \ref{subsec5.1} above will vanish}
  while the two other intersection points do.
This asymmetry is already a strong indication that things will go wrong and we cannot
expect to obtain the correct result
when evaluating the 2-cluster explicitly.\par

Now, if we work with full ribbons then there will be nine intersection points
contributing to each 2-cluster because to each of the two intersecting full ribbons we have three associated loops
(cf. Remark \ref{rm_full_ribbons} in Sec. \ref{subsec4.3} above).
These nine intersection points are again not treated equally.
Five  intersection points do not detect an interaction (one intersection point inside the intersection 2-face
 $F \in \face_2(\cK)$ and four on its four vertices)  but four intersection points do detect an interaction
(namely, the four points in the middle of each edge of $F$).
This time we have a symmetric situation and can be more optimistic that we obtain the correct result
when evaluating the 2-cluster explicitly in the $BF_3$-setting.
\end{remark}

\section{Transition to the ``$BF_3$-theory point of view''}
\label{sec7}

\subsection{Motivation}
\label{subsec7.1}

 The  simplicial program for Abelian CS theory (cf.
Sec. 3 in \cite{Ha7a}) was completed successfully
by D.H. Adams, see \cite{Ad0,Ad1}. A crucial step in \cite{Ad0,Ad1} was
the transition to the ``BF-theory point of view'', which can be divided into
two steps, namely ``field doubling''\footnote{which can come in the form of ``group doubling'' (see Step 1 below)
or ``base manifold doubling'';
observe that the word ``doubling'' is slightly misleading because it ignores
a sign change: we have $k_2 = - k_1$ where $k_1, k_2$ are as in the first paragraph of ``Step 1'' below}
 followed by  a suitable linear change of variables, cf. part \ref{appC'} of the Appendix below.\par
Adams' results  seem to suggest that
-- if one wants to have a chance of carrying out the simplicial program successfully also
for Non-Abelian CS theory -- then a similar strategy
will have to be used.\par

So far we have worked with the original CS point of view
because this helped us to reduce the lengths of many formulas
and because  for Theorem \ref{main_theorem} (which deals only with a special class of simplicial ribbon links $L$)
the  original CS point of view is sufficient.
On the other hand,  the Abelian ``test situation''
which we considered  at the end of Sec. \ref{sec6}
showed us that we can not expect to obtain correct results within the original CS point of view
when  dealing with  general simplicial ribbon links.
This is why from now on we will work with the ``$BF_3$-theory point of view''.

\subsection{The ``$BF_3$-theory point of view''}
\label{subsec7.2}

 In the following we will   make the   transition
 from non-Abelian CS theory in the torus gauge to the corresponding ``$BF_3$-theory point of view''
 at a heuristic level.

\subsection*{Step 1: ``Group doubling''}

 Let us now consider   the version of Eq. \eqref{eq2.48} in  the special case where
 $G=\tilde{G} \times \tilde{G}$ where $\tilde{G}$ is a simple, simply-connected
 compact Lie group and  where $(k_1,k_2) $ fulfills $k_1 = - k_2$,
 cf. Remark \ref{rm2.1} above.   We set $k:= k_1 = -k_2$.\par

  For simplicity,  let us consider the special case   where each of the representations $\rho_i$
   appearing in Eq. \eqref{eq2.48} is of the form
$\rho_i(\tilde{g}_1,\tilde{g}_2)= \tilde{\rho}_i(\tilde{g}_1)$, $\tilde{g}_1, \tilde{g}_2 \in \tilde{G}$,
for a some $\tilde{G}$-representation  $\tilde{\rho}_i$.
In this situation we should have\footnote{the first ``$\sim$'' follows from a short
heuristic computation}
 \begin{equation} \label{eq2.51} \WLO(L)  \sim \WLO_{\tilde{G}}(L) \overline{\WLO_{\tilde{G}}(\emptyset)} \sim
|L| \cdot \overline{|\emptyset|} \overset{(*)}{=}  |L| \cdot |\emptyset|
\end{equation}
 where $\WLO_{\tilde{G}}(L)$ on the RHS  is
 defined as  $\WLO(L)$  in Sec. \ref{subsec2.1b} for the group $\tilde{G}$
 instead of $G$ and where $|\cdot|$ is now the shadow invariant for  $\tilde{\cG}$ and $k$.
   In step $(*)$ we used the fact that $|\emptyset|$ is a real number.

Let us fix a maximal torus $\tilde{T}$ of $\tilde{G}$ and set
$T = \tilde{T} \times \tilde{T}$.
Let $\cB$,  $\cA^{\orth}$,
$\Check{\cA}^{\orth}$,   $\cA^{\orth}_c$, and
$\ll \cdot, \cdot \gg_{\cA^{\orth}}$  be defined  as in Sec. \ref{subsec2.1} and Sec. \ref{subsec2.2} for the
group $G=\tilde{G} \times \tilde{G}$.

\smallskip

In the following  $B_1, B_2$ (resp. $A^{\orth}_1$ and $A^{\orth}_2$) will denote the two
components of $B \in C^{\infty}(\Sigma,\ct) = C^{\infty}(\Sigma,\tilde{\ct}) \oplus C^{\infty}(\Sigma,\tilde{\ct})$
(resp. $A^{\orth} \in  C^{\infty}(S^1,\cA_{\Sigma,\cG}) = C^{\infty}(S^1,\cA_{\Sigma,\tilde{\cG}})
\oplus C^{\infty}(S^1,\cA_{\Sigma,\tilde{\cG}})$).
Moreover, we denote by
 $\tilde{I}$  the kernel of $\exp_{|\tilde{\ct}}:\tilde{\ct} \to \tilde{T}$.

\subsection*{Step 2: Linear change of variable}

As we  explain in  part \ref{appC'} of the Appendix,
CS theory with group $G= \tilde{G} \times
 \tilde{G}$ and  $(k_1,k_2)= (k,-k)$
 is equivalent to  $BF_3$-theory with group $\tilde{G}$ and ``cosmological constant'' $\Lambda$
 given by  $\Lambda = \tfrac{1}{k^2}$.
 More precisely, at the heuristic level, these two theories are related by a
 simple linear change of variables\footnote{oberve that $\kappa=\tfrac{1}{k}$
 in  Eqs. \eqref{eqB.5} and \eqref{eqB.8'}},
 cf. Eqs. \eqref{eqB.5} or Eqs.  \eqref{eqB.8'}  in part \ref{appC'} of the Appendix
  depending on whether we are dealing with the non-gauge fixed path integral
   or the path integral in the torus gauge.

\medskip

In order to simplify the notation a bit (and to avoid the appearance
of multiple $k$-factors)
we will work with the following simplified change of variable  $A^{\orth}   \to \tilde{A}^{\orth}$,
 $B \to  \tilde{B}$
instead of the one in Eq. \eqref{eqB.8'}:
 \begin{subequations} \label{eqB.8}
\begin{align}
 \label{eqB.8a}
 \tilde{A}^{\orth}  & :=  \bigl(\tfrac{ A^{\orth}_1 +
A^{\orth}_2}{2},  \tfrac{A^{\orth}_1 - A^{\orth}_2}{2}\bigr),\\
\label{eqB.8c} \tilde{B} & := \bigl(\tfrac{B_1 + B_2}{2},  \tfrac{B_1 - B_2}{2} \bigr)
\end{align}
\end{subequations}

By applying this  linear change of variable to the RHS  of  Eq. \eqref{eq2.48}
(in the special case $G=\tilde{G} \times \tilde{G}$, $(k_1,k_2) = (k,-k)$)
 we  arrive at\footnote{here we have used that -- according to the assumption made at the beginning of ``Step 1'' --
 we have $\rho_i(\tilde{g}_1,\tilde{g}_2)= \tilde{\rho}_i(\tilde{g}_1)$, for all $\tilde{g}_1, \tilde{g}_2 \in \tilde{G}$,  which implies  $\Tr_{\rho_i}\bigl( \Hol_{l_i}(A^{\orth},B)\bigr)  =
 \Tr_{\rho_i}\bigl( \Hol_{l_i}(A^{\orth}_1,B_1),\Hol_{l_i}(A^{\orth}_2,B_2) \bigr) =
  \Tr_{\tilde{\rho}_i}\bigl(\Hol_{l_i}(A^{\orth}_1,B_1)\bigr)$}
\begin{multline} \label{eq_WLO_BF}  \WLO(L)
 \sim \sum\limits_{(y_1,y_2) \in \tilde{I} \times \tilde{I}}  \int_{\tilde{\cA}^{\orth}_c \times \tilde{\cB}}
 1_{C^{\infty}(\Sigma,\tilde{\ct}_{reg} \times \tilde{\ct}_{reg} )}((\tilde{B}_1 + \tilde{B}_2,\tilde{B}_1 - \tilde{B}_2))  \Det_{FP}((\tilde{B}_1 + \tilde{B}_2,\tilde{B}_1 - \tilde{B}_2)) \\
 \times   \biggl[ \int_{\Check{\tilde{\cA}}^{\orth}} \prod_i  \Tr_{\tilde{\rho}_i}\bigl(
 \Hol_{l_i}((\Check{\tilde{A}}^{\orth}+ \tilde{A}^{\orth}_c)_1+(\Check{\tilde{A}}^{\orth}+ \tilde{A}^{\orth}_c)_2,
  \tilde{B}_1 + \tilde{B}_2)\bigr)
  \exp(i \bS(\Check{\tilde{A}}^{\orth},\tilde{B}) ) D\Check{\tilde{A}}^{\orth} \biggr] \\
 \times \exp\bigl( - 2\pi i k  \langle (y_1,y_2), ((\tilde{B}_1 + \tilde{B}_2)(\sigma_0),(\tilde{B}_1 - \tilde{B}_2)(\sigma_0)) \rangle \bigr) \bigr)
 \exp(i  \bS(\tilde{A}^{\orth}_c, \tilde{B})) (D\tilde{A}^{\orth}_c \otimes D\tilde{B})
\end{multline}
where for reasons of notational consistency,
 we have written $\Check{\tilde{\cA}}^{\orth}$
 instead of $\Check{{\cA}}^{\orth}$, $\tilde{\cA}^{\orth}_c$ instead of $\cA^{\orth}_c$,  and  $\tilde{\cB}$
 instead of  ${\cB}$ and  we have set
 \begin{align*}
\bS(\Check{\tilde{A}}^{\orth},\tilde{B}) & := S_{CS}(\Check{A}^{\orth},B)\\
 \bS(\tilde{A}^{\orth}_c,\tilde{B})  & := S_{CS}(A^{\orth}_c,B)
\end{align*}
More explicitly, we have\footnote{the first  appearance of $S_{CS}$ on the RHS  of the following
equation is a short hand for $S_{CS}\bigl(M,\tilde{G} \times \tilde{G},(k_1,k_2)\bigr)$
while the other appearances are a shorthand for
$S_{CS}\bigl(M,\tilde{G},k_1\bigr)  = S_{CS}\bigl(M,\tilde{G},- k_2\bigr)$.}
 \begin{align}  \label{eq_tildeS_CS_explizit}
  \bS(\Check{\tilde{A}}^{\orth},\tilde{B}) &  =  S_{CS}((\Check{\tilde{A}}^{\orth}_1+\Check{\tilde{A}}^{\orth}_2,\Check{\tilde{A}}^{\orth}_1-\Check{\tilde{A}}^{\orth}_2),
 (\tilde{B}_1+\tilde{B}_2, \tilde{B}_1-\tilde{B}_2)) \nonumber \\
& = S_{CS}(\Check{\tilde{A}}^{\orth}_1+\Check{\tilde{A}}^{\orth}_2,\tilde{B}_1+\tilde{B}_2) -
S_{CS}(\Check{\tilde{A}}^{\orth}_1-\Check{\tilde{A}}^{\orth}_2,\tilde{B}_1-\tilde{B}_2) \nonumber \\
& = S_{CS}(\Check{\tilde{A}}^{\orth}_1+\Check{\tilde{A}}^{\orth}_2 + (\tilde{B}_1+\tilde{B}_2)dt) -
S_{CS}(\Check{\tilde{A}}^{\orth}_1-\Check{\tilde{A}}^{\orth}_2 + (\tilde{B}_1-\tilde{B}_2)dt) \nonumber \\
& =   \pi k \ll (\Check{\tilde{A}}^{\orth}_1,\Check{\tilde{A}}^{\orth}_2), \biggl( \begin{matrix}
  \star   \ad( \tilde{B}_2) &&  \star
 \bigl(\tfrac{\partial}{\partial t} + \ad(\tilde{B}_1) \bigr)\\
  \star  \bigl(\tfrac{\partial}{\partial t} + \ad(\tilde{B}_1) \bigr)
  && \star   \ad(\tilde{B}_2)
\end{matrix} \biggr) \cdot (\Check{\tilde{A}}^{\orth}_1,\Check{\tilde{A}}^{\orth}_2) \gg_{\Check{\tilde{\cA}}^{\orth}}
\end{align}
and
 \begin{align}  \label{eq_tildeS_CS_explizit2}
  \bS(\tilde{A}^{\orth}_c,\tilde{B}) &  =  S_{CS}(((\tilde{A}^{\orth}_c)_1+(\tilde{A}^{\orth}_c)_2,(\tilde{A}^{\orth}_c)_1-(\tilde{A}^{\orth}_c)_2),
 (\tilde{B}_1+\tilde{B}_2,
\tilde{B}_1-\tilde{B}_2)) \nonumber \\
& = \ldots  \nonumber \\
& =  4 \pi k \ll \star \cdot
 (( \tilde{A}^{\orth}_c)_2,(\tilde{A}^{\orth}_c)_1),
  (d \tilde{B}_1,d \tilde{B}_2) \gg_{\cA_{\Sigma,\tilde{\ct} \oplus \tilde{\ct}}}
\end{align}
where $\Check{\tilde{A}}^{\orth} = (\Check{\tilde{A}}^{\orth}_1,\Check{\tilde{A}}^{\orth}_2)$,
$\tilde{A}^{\orth}_c = ((\tilde{A}^{\orth}_c),(\tilde{A}^{\orth}_c)_2)$,
and   $\tilde{B} = (\tilde{B}_1,\tilde{B}_2)$.

\begin{remark} \rm \label{rm_7.2a}
Recall that above we have assumed that $\tilde{G}$ is a simple, simply-connected compact
(and therefore non-Abelian) Lie group.
For sake of completeness (and in view of the discussion at the end of Sec. \ref{sec6} above
and Remark \ref{rm_7.2b} and  Remark \ref{rm_7.3} below)  let us mention that,
 in fact, one can derive (an analogue of) Eq. \eqref{eq_WLO_BF}  also  if $\tilde{G}$ is an Abelian compact Lie group.
Of course, in this case the RHS  of Eq. \eqref{eq_WLO_BF} simplifies drastically\footnote{we then have
 $\tilde{T}=\tilde{G}$, $\tilde{\ct}_{reg}=\tilde{\ct}$,  $\ad(\tilde{B}_j)=0$,
 $\Det_{FP}((\tilde{B}_1 + \tilde{B}_2,\tilde{B}_1 - \tilde{B}_2))$ is a constant function,
   and the sum $\sum_{y_1,y_2}$ is trivial.
    So we obtain:  $\WLO(L)  \sim   \int  \bigl[ \int \prod_i  \Tr_{\tilde{\rho}_i}\bigl(
 \Hol_{l_i}((\Check{\tilde{A}}^{\orth}+ \tilde{A}^{\orth}_c)_1+  (\Check{\tilde{A}}^{\orth}+ \tilde{A}^{\orth}_c)_2,
  \tilde{B}_1 + \tilde{B}_2)\bigr)
  \exp(i \bS(\Check{\tilde{A}}^{\orth},\tilde{B}) ) D\Check{\tilde{A}}^{\orth} \bigr] \exp(i  \bS(\tilde{A}^{\orth}_c, \tilde{B})) (D\tilde{A}^{\orth}_c \otimes D\tilde{B})$}.
\end{remark}

\begin{remark} \rm \label{rm_7.2b}
It might seem surprising
that the second of the aforementioned two steps,
i.e. the linear change of variable, really makes an essential difference.
Clearly, the original  heuristic path integral and the heuristic
 path integral after the application of the change of variable are equivalent.
However, once the problem of discretizing the
corresponding  path integral  is considered the difference really matters.
A detailed look at \cite{Ad0,Ad1}
will convince the  reader that this is indeed the case at least in the Abelian situation.\par

That a linear change of variables is useful also for the  discretization
 of non-Abelian CS (with doubled group)  is  less obvious.
Observe, for example, that there is a  $\star$-operator on the main diagonal of the $2 \times 2$-matrix
appearing in Eq. \eqref{eq_tildeS_CS_explizit} above.  Because of this we cannot hope
to be able to find a discretized version of the path integral on the RHS  of Eq. \eqref{eq_WLO_BF}
where each of the two components $\Check{\tilde{A}}^{\orth}_1$ and $\Check{\tilde{A}}^{\orth}_2$ ``lives'' either on
 $K_1 \times \bZ_N$ or on $K_2 \times \bZ_N$.
 Instead, each component $\Check{\tilde{A}}^{\orth}_1$ and $\Check{\tilde{A}}^{\orth}_2$
  must be implemented in a ``mixed'' fashion (which is what we did in the discretization approach
   of the present paper, cf. Sec. \ref{sec4}).
 This is a crucial difference compared to the Abelian situation where
it was indeed possible to find a non-mixed discretization
for the relevant simplicial fields.
This difference is one of the reasons
why we decided to postpone the transition to the $BF_3$-theory point of view
until now.
\end{remark}

\subsection{Simplification of some of the notation in Sec. \ref{subsec7.2}}
\label{subsec7.2b}

Before we discretize the expression on the RHS  of Eq. \eqref{eq_WLO_BF}
let us first simplify the notation somewhat:\par

Firstly,  we will
 drop the  $\tilde{\ }$-signs appearing in the previous subsection,
for example will write $G$ instead of $\tilde{G}$ , $\cB$ instead of $\tilde{\cB}$,
$I$ instead of $\tilde{I}$
and so on. Clearly, we then have
\begin{subequations}
 \begin{align}
 \cB &  =C^{\infty}(\Sigma,\ct \oplus \ct) \\
  \cA^{\orth} & =   C^{\infty}(S^1,\cA_{\Sigma,\cG \oplus \cG})\\
 \Check{\cA}^{\orth} & = \{ A^{\orth} \in  \cA^{\orth} \mid \int A^{\orth}(t) dt
  \in \cA_{\Sigma,\ck \oplus \ck} \}\\
 \cA^{\orth}_c & = \{ A^{\orth} \in  \cA^{\orth} \mid \text{ $A^{\orth}$ is constant and
 $\cA_{\Sigma,\ct \oplus \ct}$-valued} \}
    \end{align}
  \end{subequations}
 Moreover, we will set
 \begin{align*} B_{\pm} & := B_1 \pm B_2 \quad \text{ for $B=(B_1,B_2) \in \cB$}\\
 A^{\orth}_{\pm} & := A^{\orth}_1 \pm A^{\orth}_2 \quad \text{for $A^{\orth}=(A^{\orth}_1,A^{\orth}_2) \in \cA^{\orth}$}
 \end{align*}
  and use the notation  $y_+$ instead of $y_1$ and  $y_-$ instead of $y_2$.
 Then   we can rewrite Eq. \eqref{eq_WLO_BF} in the following way
 \begin{multline} \label{eq_WLO_BF_Eversion}  \WLO(L)
 \sim \sum_{y_+,y_- \in I}  \int_{ \cA^{\orth}_c \times {\cB}}
 \biggl\{ \prod_{\pm} \bigl[ 1_{C^{\infty}(\Sigma,{\ct}_{reg}}(B_{\pm})  \Det_{FP}(B_{\pm}) \bigr] \\
 \times   \biggl[ \int_{\Check{{\cA}}^{\orth}} \prod_i  \Tr_{{\rho}_i}\bigl(
 \Hol_{l_i}((\Check{A}^{\orth} + A^{\orth}_c)_+,  B_+ )\bigr)
  \exp(i \bS(\Check{{A}}^{\orth},{B}) ) D\Check{{A}}^{\orth} \biggr] \\
 \times \exp\bigl( - 2\pi i k  \sum_{\pm} \langle y_{\pm}, B_{\pm}(\sigma_0)  \rangle  \bigr)
   \biggr\}  \exp(i  \bS(A^{\orth}_c, B)) (DA^{\orth}_c \otimes DB)
\end{multline}
where  $\prod_{\pm} \cdots$ (resp. $\sum_{\pm} \cdots$) is the obvious two term product (resp. sum).
Above  we have set  (cf. Eq. \eqref{eq_tildeS_CS_explizit} and Eq. \eqref{eq_tildeS_CS_explizit2} above)
 \begin{align} \label{eq7.9}
  \bS(\Check{A}^{\orth},B) & :=  \pi k \ll (\Check{A}^{\orth}_1,\Check{A}^{\orth}_2) , \star \  \biggl(
 \begin{matrix}
   \ad( {B}_2) &&
 \tfrac{\partial}{\partial t} + \ad({B}_1) \\
 \tfrac{\partial}{\partial t} + \ad({B}_1)
  &&   \ad({B}_2)
\end{matrix} \biggr) \cdot (\Check{A}^{\orth}_1,\Check{A}^{\orth}_2) \gg_{\cA^{\orth}} \\
\label{eq7.10} \bS(A^{\orth}_c,B) & :=
4 \pi k \ll \star  \cdot
 (( A^{\orth}_c)_2,(A^{\orth}_c)_1),
  (d B_1,d B_2) \gg_{\cA^{\orth}}
\end{align}
where    $\star: \cA^{\orth} \to \cA^{\orth}$ is the Hodge star operator
induced by the auxiliary Riemannian metric $\mathbf g$.

\begin{remark} \rm
 Observe that for each $l \in \{l_1,l_2, \ldots, l_m\}$  we have
  (cf. Eqs. \eqref{eq_Hol_heurist} and \eqref{eq_Hol_heurist_abbr} in Sec. \ref{subsec2.1} above
  and  Eq. (5.15) in \cite{Ha7a})
 \begin{multline} \label{eq_Hol_cont_BF}
\Hol_{l}(A^{\orth}_+,  B_+ )
 =\lim_{n \to \infty} \prod_{k=1}^n \exp\biggl(\tfrac{1}{n} \bigl(A^{\orth}_+
  +  B_+ dt \bigr)(l'(t))   \biggr)_{| t=k/n}\\
 =  \lim_{n \to \infty} \prod_{k=1}^n \exp\biggl( \bigl( A^{\orth}_1(l_{S^1}(t))( \tfrac{1}{n} l'_{\Sigma}(t)) + A^{\orth}_2(l_{S^1}(t))(\tfrac{1}{n} l'_{\Sigma}(t))\\
   +  B_1(l_{\Sigma}(t)) dt(\tfrac{1}{n} l'_{S^1}(t)) + B_2(l_{\Sigma}(t)) dt(\tfrac{1}{n} l'_{S^1}(t))   \bigr)   \biggr)_{| t=k/n}
\end{multline}
for $A^{\orth} \in \cA^{\orth}$ and $B \in \cB$.
\end{remark}

\subsection{Discretization of Eq. \eqref{eq_WLO_BF_Eversion}}
\label{subsec7.3}

We will now sketch how -- using  a suitable
discretization of the expression on the RHS  of Eq. \eqref{eq_WLO_BF_Eversion} --
a rigorous definition of  $\WLO(L)$ appearing in Eq. \eqref{eq_WLO_BF_Eversion} can be obtained.
 (In part \ref{appJ} of the Appendix we will sketch an alternative
way of discretizing the RHS  of Eq. \eqref{eq_WLO_BF_Eversion}).

\smallskip

In view of what we have learned in Sec. \ref{subsec5.5}  (cf. Remark \ref{rm_Step5_full_ribbon}) we will
now work with simplicial ribbons in $\cK \times \bZ_N$ (= ``full ribbons'')
 instead of simplicial ribbons in $q\cK \times \bZ_N$ (= ``half ribbons'').

\smallskip

We set
 \begin{subequations} \label{eq_basic_spaces_discr_BF}
    \begin{align}
 \cB(q\cK)  & :=  C^0(q\cK,\ct \oplus \ct) \\
  \cA_{\Sigma}(q\cK)  & :=  C^1(q\cK,\cG \oplus \cG) \\
  \cA^{\orth}(q\cK) & := \Map(\bZ_N, \cA_{\Sigma}(q\cK))
 \end{align}
\end{subequations}
 and introduce the scalar product $\ll \cdot, \cdot \gg_{\cA^{\orth}(q\cK)}$  on $\cA^{\orth}(q\cK)$  in an analogous way as in Sec. \ref{subsec4.0} above.\par

For technical reasons  we will again
introduce the  subspaces  $\cA_{\Sigma}(K)$ and $\cA^{\orth}(K)$
of $\cA_{\Sigma}(q\cK)$ and $\cA^{\orth}(q\cK)$ given by\footnote{Here and in the following
we use again the notation $\cA_{\Sigma,V}(K) := C^1(K_1,V) \oplus C^1(K_2,V)$ for a finite-dimensional
real vector space $V$}
 \begin{subequations}
\begin{align}
\cA_{\Sigma}(K) & :=\cA_{\Sigma, \cG \oplus \cG}(K)  \subset \cA_{\Sigma}(q\cK)\\
\cA^{\orth}(K) & := \Map(\bZ_N,\cA_{\Sigma}(K)) \subset \cA^{\orth}(q\cK)
\end{align}
\end{subequations}

Moreover, we will introduce the subspace
$\cB_0(q\cK):= \psi(\cB(\cK))$ of $\cB(q\cK)$
where $\cB(\cK):= C^0(\cK,\ct \oplus \ct)$ and where
$\psi: \cB(\cK) \to \cB(q\cK)$  is given exactly like in Choice \ref{example3} in Sec. \ref{subsec4.10a}
above with $\ct$ replaced by $\ct \oplus \ct$.

\smallskip

As in Sec. \ref{subsec4.2} we have a
well-defined  operator $\star_{K}: \cA^{\orth}(K) \to \cA^{\orth}(K)$
and as in Sec. \ref{subsec4.0} we have a  decomposition
$$ \cA^{\orth}(K) = \Check{\cA}^{\orth}(K) \oplus  \cA^{\orth}_c(K)$$
where
 \begin{subequations}
    \begin{align}
 \label{eq_basic_spaces_discr_BFb}      \Check{\cA}^{\orth}(K) & := \{ A^{\orth} \in \cA^{\orth}(K)
       \mid \sum_{t \in \bZ_N}  A^{\orth}(t) \in \cA_{\Sigma, \ck \oplus \ck}(K)\} \\
   \label{eq_basic_spaces_discr_BFc}   \cA^{\orth}_c(K) & := \{ A^{\orth} \in \cA^{\orth}(K)
       \mid \text{ $A^{\orth}(\cdot)$ is constant and $ \cA_{\Sigma,\ct \oplus \ct}(K)$-valued}\} \cong \cA_{\Sigma,\ct \oplus \ct}(K)
\end{align}
 \end{subequations}

 Moreover, we set
 \begin{align*} B_{\pm} & := B_1 \pm B_2 \quad \text{ for $B=(B_1,B_2) \in \cB(q\cK)$}\\
 A^{\orth}_{\pm} & := A^{\orth}_1 \pm A^{\orth}_2 \quad \text{ for $A^{\orth}=(A^{\orth}_1,A^{\orth}_2) \in \cA^{\orth}(K)$}
 \end{align*}
\noindent
 As the discrete analogues of Eq. \eqref{eq7.9} and  Eq. \eqref{eq7.10} above we now take\footnote{Recall that  $L^{(N)}(B_0)$ is a discrete ``approximation'' of $\partial_t + \ad(B_0)$ so $\tfrac{1}{2} (L^{(N)}(B_{+}) + L^{(N)}(B_{-}))$
is a discrete analogue of $\tfrac{1}{2} ( \partial_t + \ad(B_+) +\partial_t + \ad(B_-)) =\tfrac{1}{2} ( \partial_t + \ad(B_1 + B_2) +\partial_t + \ad(B_1 - B_2)) =
\partial_t +  \ad(B_1)$;  similarly $\tfrac{1}{2} (L^{(N)}(B_{+}) - L^{(N)}(B_{-}))$
is a discrete approximation   of  $\tfrac{1}{2} ( \partial_t + \ad(B_1 + B_2) - (\partial_t + \ad(B_1 - B_2))) =
\ad(B_2)$}
  \begin{align} \label{eq_7.12}
   \bS^{disc}(\Check{A}^{\orth},B) & :=  \pi k \ll (\Check{A}^{\orth}_1,\Check{A}^{\orth}_2) ,\star_{K} R^{(N)}(B)  \cdot (\Check{A}^{\orth}_1,\Check{A}^{\orth}_2) \gg_{\cA^{\orth}(q\cK)} \\
 \bS^{disc}(A^{\orth}_c,B) & :=   4 \pi k \ll \star_K \cdot
 (( A^{\orth}_c)_2,(A^{\orth}_c)_1),
  ( d_{q\cK} B_1, d_{q\cK} B_2) \gg_{\cA^{\orth}(q\cK)}
\end{align}
where
\begin{equation} \label{eq_RNB-def}
R^{(N)}(B) := \biggl( \begin{matrix} \tfrac{L^{(N)}(B_{+}) - L^{(N)}(B_{-})}{2}
   &&  \tfrac{L^{(N)}(B_{+}) + L^{(N)}(B_{-})}{2}  \\
      \tfrac{L^{(N)}(B_{+}) + L^{(N)}(B_{-})}{2}  &&
       \tfrac{L^{(N)}(B_{+}) - L^{(N)}(B_{-})}{2} \end{matrix} \biggr)
\end{equation}

\smallskip

As mentioned above we
 will now work with ``full ribbons'', i.e. closed simplicial ribbons $R$  in $\cK \times \bZ_N$.
Recall from Remark \ref{rm_full_ribbons} that each such $R$
induces three simplicial loops $l^+ = (l^{+(k)})_{k \le n}$, $l^- = (l^{-(k)})_{k \le n}$, and
     $l = (l^{(k)})_{k \le n}$, $n \in \bN$, in $q\cK \times \bZ_N$.
As the discrete analogue $\Hol^{disc}_{R}(A^{\orth}_+,  B_+ ) $
 of the continuum expression $\Hol_{l}(A^{\orth}_+,  B_+ )$
 in Eq. \eqref{eq_Hol_cont_BF} above we now take
\begin{multline}  \label{eq_Hol_disc_BF}
\Hol^{disc}_{R}(A^{\orth}_+,  B_+ )
:=  \prod_{k=1}^n \exp\biggl(  \bigl( \sum_{\pm} \tfrac{1}{2}  A^{\orth}_1(\start l^{\pm(k)}_{S^1})(l^{\pm(k)}_{\Sigma}) \bigr) +  A^{\orth}_2(\start l^{(k)}_{S^1})  (l^{(k)}_{\Sigma})  \\
  +   \bigl( \sum_{\pm} \tfrac{1}{2}  B_1(\start l^{\pm(k)}_{\Sigma})  dt^{(N)}(l^{\pm(k)}_{S^1}) \bigr)
  +    B_2(\start l^{(k)}_{\Sigma})   dt^{(N)}(l^{(k)}_{S^1})   \biggr)
\end{multline}

In view of Eq. \eqref{eq4.21_full} above and   the list of replacements  in Sec. 5.4.1 in \cite{Ha7a}
the ansatz in Eq. \eqref{eq_Hol_disc_BF} is quite natural.
The only point that requires an explanation\footnote{one could, of course,
ask the analogous question with respect to the field components
$B_1$ and $B_2$. However, the latter question can easily be avoided
since it turns out  that the value of  $\WLO_{rig}(L)$ as defined
in Eq. \eqref{eq_WLO_BF_norm} below does not change if in Eq. \eqref{eq_Hol_disc_BF} we replace the  expression
$\bigl( \sum_{\pm} \tfrac{1}{2}  B_1(\start l^{\pm(k)}_{\Sigma})  dt^{(N)}(l^{\pm(k)}_{S^1}) \bigr)
  +    B_2(\start l^{(k)}_{\Sigma})   dt^{(N)}(l^{(k)}_{S^1})$ by the symmetric expression
 $\sum_{j=1}^2 \bigl[ \bigl( \sum_{\pm} \tfrac{1}{4}  B_j(\start l^{\pm(k)}_{\Sigma})  dt^{(N)}(l^{\pm(k)}_{S^1}) \bigr)  +  \tfrac{1}{2} B_j(\start l^{(k)}_{\Sigma}) dt^{(N)}(l^{(k)}_{S^1}) \bigr]$.}
  is why the field component $A^{\orth}_2$ ``interacts'' only with the loop $l$ while
 $A^{\orth}_1$ ``interacts'' with the two  loops $l^+$ and $l^-$.
 We will give this explanation in  ``Observation 1'' in part \ref{appJ.1} of the Appendix.

\begin{remark} \label{rm_7.3} \rm
 In the special case where $G$ is Abelian (cf. Remark \ref{rm_7.2a} above)
    there are several     simplifications.
    For example, in the special case where the (complex)
 representation  $\rho$ of $G$ is irreducible (and therefore 1-dimensional) we  have from Eq. \eqref{eq_Hol_disc_BF}
 \begin{multline*}  \Tr_{\rho}\bigl( \Hol^{disc}_{l}(A^{\orth}_+,  B_+ )  \bigr)
=  \rho\biggl( \prod_{k=1}^n \exp\biggl( \sum_{\pm} \tfrac{1}{2} A^{\orth}_1(\start l^{\pm(k)}_{S^1})(l^{\pm(k)}_{\Sigma})
    +     \tfrac{1}{2}  B_1(\start l^{\pm(k)}_{\Sigma})  dt^{(N)}(l^{\pm(k)}_{S^1}) \biggr) \biggr)\\
  \times  \rho\biggl(  \prod_{k=1}^n \exp\biggl(  A^{\orth}_2(\start l^{(k)}_{S^1})  (l^{(k)}_{\Sigma}) +
    B_2(\start l^{(k)}_{\Sigma})   dt^{(N)}(l^{(k)}_{S^1}) \biggr) \biggr)
\end{multline*}
We can work with a generalization of the last expression
where the second ``$\rho$'' appearing above  is replaced by another finite-dimensional representation $\rho'$ of $G$.
 By doing so we obtain a kind\footnote{Let us emphasize that we do not get a strict
analogue: in the case of Abelian $G$ our approach is less general than \cite{Ad0,Ad1}.
 We remark  that for an Abelian $G$ one actually can construct  a strict torus gauge analogue of the approach in \cite{Ad0,Ad1}  by modifying our approach in a suitably way. However, this modified approach will
 not be useful for dealing with the case of non-Abelian $G$} of torus gauge ``analogue'' of
the approach in \cite{Ad0,Ad1}.
\end{remark}

As the discrete version for the two expressions $\Det_{FP}(B_{\pm})$ appearing in Eq. \eqref{eq_WLO_BF_Eversion}
we choose again (as in  Eq. \eqref{eq_ Det_disc_FP1} in Sec. \ref{subsec4.10a} above)
\begin{align}  \label{eq_def_DetFPdisc_BF}
 \Det^{disc}_{FP}(B_{\pm}) & :=
   \prod_{x \in \face_0(q\cK)}  \det\nolimits^{1/2}\bigl(1_{{\ck}}-\exp(\ad(B_{\pm}(x)))_{| {\ck}}\bigr)
  \end{align}

The remaining steps  for discretizing the RHS  of  Eq. \eqref{eq_WLO_BF_Eversion}
 can be carried out easily (as in  in Secs \ref{subsec4.6}--\ref{subsec4.9} above).
 There is only one exception: since we are now working with ``full ribbons''
 it will be necessary to use an additional regularization procedure, cf. Remark \ref{rm_Step5_full_ribbon} in Sec. \ref{subsec5.5} above.

 \smallskip

 We then arrive at a rigorous version $\WLO^{disc}_{rig}(L)$
 of $\WLO(L)$ and its normalization
\begin{equation} \label{eq_WLO_BF_norm} \WLO_{rig}(L):= \tfrac{\WLO^{disc}_{rig}(L)}{\WLO^{disc}_{rig}(\emptyset)}
\end{equation}
In view of the heuristic formula Eq. \eqref{eq2.51}
we expect that for $k \ge \cg$ we have\footnote{Here $|\cdot|$ is the shadow invariant for  $\cG$ and $k$;
recall that we write the Lie algebra $\tilde{\cG}$ now simply as $\cG$, i.e. without the $\sim$}
 \begin{equation} \label{eq_maintheorem_BF} \WLO_{rig}(L) = \frac{|L|  |\emptyset|}{|\emptyset|  |\emptyset|}=\frac{|L|}{|\emptyset|}
\end{equation}
and in contrast to the situation in Theorem \ref{main_theorem}
there is now a reasonable chance that Eq. \eqref{eq_maintheorem_BF} even holds
for general simplicial ribbon links $L$ in $\cK \times \bZ_N$.

\section{Discussion \& Outlook}
\label{sec8}

In   \cite{Ha7a} we proposed a  rigorous ``simplicial'' realization
 of the torus-gauge-fixed non-Abelian CS-path integral for manifolds $M$ of the form $M=\Sigma \times S^1$.
In the present paper we proved the  main result of \cite{Ha7a}, Theorem \ref{main_theorem} above,
which deals with a special class of simplicial ribbon links.
During the proof of Theorem \ref{main_theorem} and in Sec \ref{sec6} it became clear that
in order to have a reasonable chance of generalizing our computations successfully to the case
of general simplicial ribbon links it seems to be necessary to make (at least) the following two modifications
of our approach:
 \begin{enumerate}

\item we should make the ``transition to the $BF$-theory point of view''
  (cf. ``Step 1'' and ``Step 2'' in Sec. \ref{subsec7.2} above),

  \item we should work with simplicial ribbons in $\cK \times \bZ_N$ (``full ribbons'') rather than with
 simplicial ribbons in $q\cK \times \bZ_N$ (``half ribbons'').

 \end{enumerate}
We sketched such a modification of our approach in Sec. \ref{subsec7.3}
(and a suitable reformulation of it in part \ref{appJ} of the Appendix).
In \cite{Ha7c} we will study this new approach in more
detail\footnote{In fact, since the computations  for general links are quite complicated
we will restrict our attention to the special
group $G=SU(2)$ and we will first compute only the low order terms  in the
 expansion of $\WLO_{rig}(L)$ as an asymptotic series of powers of $1/k$ (for  $k \to \infty$)}.
It remains to be seen whether we really obtain the correct values for $\WLO_{rig}(L)$
in the case of general $L$. Here are some points which suggest that the chances for this being the case
 are quite good.

\begin{itemize}

\item We mentioned in Sec. 6.2 in \cite{Ha7a}
 that also for simplicial ribbon links $L$  fulfilling\footnote{in this case
 some of the simplicial ribbons $R_i$, $i \le m$, contained in $L$
are allowed to be   ribbon analogues of non-trivial framed torus knots}
 only a  weaker version of condition (NCP)' we can evaluate
   $\WLO_{rig}(L)$ explicitly using a suitable modification of the approach used in Sec. \ref{sec5} above
   for proving Theorem \ref{main_theorem}, cf. \cite{Ha9}.

\item The heuristic argument in Appendix B.2 in \cite{Ha7a} shows that
      also for general links we can expect the value of $\WLO(L)$ to be a sum over step functions $B$.
      Using a rigorous version of this heuristic argument it
      should not be difficult to prove that also  $\WLO_{rig}(L)$ defined as in Sec. \ref{sec4} above
      will be a sum over step functions $B$.
      At the moment it is not yet clear whether these step functions have the correct values (or, rather, ``step sizes'')
       and wether the symmetry argument based on the affine Weyl group $\cW_{\aff}$
        we referred to in Sec. \ref{subsec5.6} above will  lead to a sum over  area colorings
       (cf. the sum $\sum_{\varphi \in col(L)} \cdots$ in Eq. \eqref{eqA.4} below)
       also in the case of general $L$.

\item If the sum $\sum_{\varphi \in col(L)} \cdots$ indeed arises also in the general situation
then it is clear\footnote{cf.   Remark \ref{rm_Step5_full_ribbon} above and recall that in \cite{Ha7c}
we will be working with ``full ribbons''}  from Sec. \ref{subsec5.5}
that also for general $L$ we will again obtain the factor $|L|^{\varphi}_1$ appearing in Eq. \eqref{eqA.4} below.
Moreover, also the factor $|L|^{\varphi}_3$ should appear again.
Things are a bit less clear regarding  the factor $|L|^{\varphi}_2$.
What remains to be clarified is whether we  obtain  the correct gleam expressions $\gleam(Y)$ also  for general $L$.

\item The crucial open question is related to the factor $|L|^{\varphi}_4$ in Eq. \eqref{eqA.4},
     which is definitely the most complicated and interesting factor.
         At the moment it is completely open whether this  factor can be obtained using our approach\footnote{
         or a suitable modification  of our approach} by evaluating the appropriate ``2-cluster'' expressions  (which will be very similar to the ones we studied in Sec. \ref{sec6}).\par

     We remark that
      in the case of Chern-Simons theory on $\bR^3$ in axial gauge  we already performed similar calculations in \cite{Ha2} (cf. also \cite{CCFM,Lab})
     and found that the analogous expressions for the ``2-clusters''
       were problematic.  One  serious problem in \cite{Ha2} was that
       the values of the 2-clusters turned out to depend on the implementation of a regularization
       procedure  called   ``loop smearing'' in \cite{Ha2}.
    It is possible that in the continuum approach in \cite{Ha4,Ha6} to CS theory on $\Sigma \times S^1$
    in the torus gauge, which also makes use of ``loop smearing'', there will be a similar ``loop smearing'' dependence.     Fortunately, within the simplicial approach developed and studied
     in \cite{Ha7a},  the present paper, and \cite{Ha7c}
     this complication is essentially\footnote{in fact, the ``half ribbons vs full ribbons'' issue
      mentioned above could be seen as a (fortunately very harmless) simplicial analogue of the issue of ``loop smearing dependence'' in the continuum setting} absent.

\end{itemize}

  \bigskip

 {\it Acknowledgements:}
 I want to thank the anonymous referee of
my paper \cite{Ha4} whose comments motivated me to look for an
alternative approach for making sense of the RHS  of (the original version of) Eq.
\eqref{eq2.48}, which is less technical than the continuum approach
in \cite{Ha3b,Ha4,Ha6}.  This eventually led  to \cite{Ha7a}  and the present paper.\par

I am also grateful to Jean-Claude Zambrini for several comments
which led to improvements in the presentation of  the present paper. \par

Finally, it is a great pleasure for me to  thank  Benjamin Himpel
for  many useful and important comments and suggestions which
 had a major impact on the presentation and
overall structure of the present paper.

 \renewcommand{\thesection}{\Alph{section}}
\setcounter{section}{0}

\section{Appendix: Lie theoretic notation II}
\label{appB}

The following two lists extend the two lists in Appendix A in \cite{Ha7a}.

\subsection{List of notation in the general case}
Recall that in Sec. \ref{subsec2.1} we
fixed a simply-connected compact Lie group $G$ (with Lie algebra $\cG$),
 a maximal $T$ of $G$ (with Lie algebra $\ct$),
  and a  Weyl chamber   $\CW \subset \ct$.

\smallskip

Apart from the  notation given in  Appendix A of \cite{Ha7a}
we  also use the following Lie theoretic notation in the present paper:

\begin{itemize}

\item  $\langle \cdot, \cdot \rangle $: the unique $\Ad$-invariant scalar product  on $\cG$
  such that\footnote{which is equivalent to the condition
 $\langle \alpha, \alpha \rangle_* = 2$ for every {\it long} real root $\alpha \in \ct^*$  where
$\langle \cdot, \cdot \rangle_*$ is the scalar product on $\ct^*$ induced by $\langle \cdot, \cdot \rangle$} $\langle \Check{\alpha}, \Check{\alpha} \rangle = 2$ holds
 for every short real coroot $\Check{\alpha}$ associated to $(\cG,\ct)$.
 Using $\langle \cdot, \cdot \rangle$ we now make the identification $\ct \cong \ct^*$.

\item $\ck$:  the $\langle \cdot , \cdot \rangle$-orthogonal complement of $\ct$ in $\cG$.

\item $\cR_{\bC}$:  the set of {\it complex} roots $\ct \to \bC$  associated to $(\cG,\ct)$

\item $\cR \subset \ct^*$: the set $\{ \tfrac{1}{2 \pi i}  \alpha_c \mid
 \alpha_c \in \cR_{\bC}\}$ of {\it real} roots associated to $(\cG,\ct)$

\item  $\cR_+ \subset \cR$:  the set of positive (real) roots
 corresponding to $\CW$

\item $\Gamma \subset \ct$: the lattice generated by the
set of real coroots associated to $(\cG,\ct)$, i.e. by the set $\{ \Check{\alpha} \mid \alpha \in \cR\}$
  where $\Check{\alpha} =  \tfrac{2 \alpha}{\langle \alpha, \alpha \rangle} \in \ct^* \cong \ct$ is the coroot associated to the root $\alpha \in \cR$.

\item $I \subset \ct$: the kernel of $\exp_{|\ct}:\ct \to T$. From the assumption that $G$ is simply-connected
 it follows that $I = \Gamma$.

\item  $\Lambda \subset \ct^* (\cong \ct)$:  the {\it real} weight lattice associated to $(\cG,\ct)$, i.e.
$\Lambda$ is the lattice which is dual to $\Gamma$.

\item $\Lambda_+ \subset \Lambda$:  the set of  dominant  weights corresponding to the Weyl chamber $\cC$,
i.e. $\Lambda_+ := \overline{\cC} \cap \Lambda$

\item $\rho$: half sum of positive roots (``Weyl vector'')
\item $\theta$: unique long root in  $\overline{\cC}$.
\item  $\cg= 1 + \langle \theta,\rho \rangle$:  the dual Coxeter number of
 $\cG$.

\item $P \subset \ct$: a fixed Weyl alcove

\item $Q \subset \ct$: a subset of $\ct$ of the form
$Q = \{ \sum_i \lambda_i e_i | 0 < \lambda_i < 1 \text{ \ } \forall i \le \dim(\ct)\}$ where $(e_i)_{i \le  \dim(\ct)}$
  is a fixed basis of $\Gamma = I$.

\item $\cW \subset \GL(\ct)$: the Weyl group of the pair $(\cG,\ct)$

\item $\cW_{\aff} \subset \Aff(\ct)$: the ``affine Weyl group of  $(\cG,\ct)$'',
i.e. the subgroup of $\Aff(\ct)$ generated by $\cW$ and the set of translations $\{ \tau_x \mid x \in \Gamma\}$
where $\tau_x: \ct \ni b \mapsto b + x \in \ct$.

\item $\cW_{k} \subset \Aff(\ct)$, $k \in \bN$:  the subgroup of $\Aff(\ct)$
given by $\{ \psi_k \circ \sigma \circ \psi_k^{-1} \mid \sigma \in \cW_{\aff} \}$
where $\psi_k : \ct \ni b \mapsto b \cdot k - \rho \in \ct$
(the ``quantum Weyl group corresponding to the  level $l := k - \cg$'')

\item  $\Lambda^k_+ \subset \Lambda$, $k \in \bN$: the subset of $\Lambda_+$ given by
$\Lambda^k_+ :=   \{ \l \in \Lambda_+  \mid  \langle \l ,\th \rangle \leq k - \cg \}$
(the ``set of dominant weights which are integrable at level $l = k - \cg$'').

\end{itemize}

In the main text, the number $k \in \bN$ appearing above will
be the integer $k$  fixed  in Sec. \ref{subsec2.1}
(which later is assumed to fulfill $k \ge \cg$).

\subsection{List of notation in the special case $G=SU(2)$}

Let us now consider the special group $G = SU(2)$ with the standard  maximal torus $T = \{\exp( \theta \tau) \mid  \theta \in \bR\}$ where
$$\tau :=\left( \begin{matrix} i  && 0 \\ 0 && -i  \end{matrix} \right)$$
Then $\cG = su(2)$ and $\ct = \bR \cdot \tau$.
There are two Weyl chambers, namely $\CW_+$ and $\cC_-$ where $\cC_{\pm}:= \pm [0,\infty) \tau$.
Let us fix $\CW := \CW_+$ in the following.
\begin{itemize}

\item $\langle \cdot, \cdot \rangle$ is  the scalar product on $\cG$ given by\footnote{in view of the formula
 $\Check{\alpha}= \alpha = 2 \pi \tau$ below we see that $\langle \cdot, \cdot \rangle$
indeed fulfills the normalization condition $\langle\Check{\alpha}, \Check{\alpha} \rangle = 2$}
 $$\langle A, B \rangle = -\tfrac{1}{4\pi^2} \Tr_{\Mat(2,\bC)}(AB)  \quad \quad \text{ for all } A, B \in \cG \subset \Mat(2,\bC)$$

\item $\cR_{\bC} = \{\alpha_c, - \alpha_c\}$  where $\alpha_c: \ct \to \bC$ is given by
$\alpha_c(\tau) = 2 i$

\item $\cR = \{\alpha, - \alpha\}$ where $\alpha := \tfrac{1}{2 \pi i} \alpha_c$.
 A short computation shows that $\alpha = 2 \pi \tau \in \ct$
 (recall  that we made  the identification  $\ct \cong \ct^*$).

\item $\cR_+ = \{\alpha\}$
\item $I = \Gamma = \bZ \cdot \Check{\alpha}$ where $\Check{\alpha} = \tfrac{2 \alpha}{\langle \alpha, \alpha \rangle} = \alpha$
\item $\Lambda = \bZ \cdot  \tfrac{\alpha}{2}$
\item $\Lambda_+ = \bN_0 \cdot  \tfrac{\alpha}{2} $
\item $\rho = \tfrac{\alpha}{2}$
\item $\theta = \alpha$
\item $\cg = 2$

\item possible choices for $P$ and $Q$ are  $P = (0,\tfrac{1}{2}) \alpha$
and $Q= (0,1) \alpha$

\item $\cW = \{1, \sigma\}$ where $1 = \id_{\ct}$ and $\sigma(b) = - b$ for $b \in \ct$;
using this and the explicit description of $I=\Gamma$ above
one easily obtains an explicit description of $\cW_{\aff}$ and $\cW_k$

\item $\Lambda^k_+ = \{0, \tfrac{1}{2} \alpha, \ldots,
\tfrac{k-2}{2} \alpha \}$ for $k \in \bN$

\end{itemize}

\section{Appendix: Turaev's shadow invariant}
\label{appA}

Let us briefly recall the definition of
Turaev's shadow invariant in the  situation relevant for us, i.e. for  manifolds $M$
of the form $M=\Sigma \times S^1$
where  $\Sigma$ is an oriented  surface.\par

Let $L= (l_1, l_2, \ldots, l_m)$, $m \in \bN$,  be a framed piecewise smooth link\footnote{this includes
the case of simplicial ribbon links in $q\cK \times \bZ_N$ as a special case, cf. Remark \ref{rm_appA_Ende}
below} in  $M= \Sigma \times S^1$.
For simplicity we will assume that each $l_i$, $i \le m$ is equipped with a ``horizontal'' framing,
cf. Remark 4.5 in Sec. 4.3 in \cite{Ha7a}.
Let $V(L)$ denote the set of points $p \in \Sigma$ where the loops
 $l^i_{\Sigma}$, $i \le m$, cross themselves or each other (the ``crossing points'')
 and $E(L)$ the set of
curves in $\Sigma$ into which the loops $l^1_{\Sigma}, l^2_{\Sigma},
\ldots, l^m_{\Sigma}$ are decomposed when being ``cut''  in the
points of $V(L)$.
 We assume that there are only finitely many
 connected components $Y_0, Y_1, Y_2, \ldots, Y_{m'}$, $m' \in \bN$ (``faces'')
 of  $\Sigma \backslash ( \bigcup_i \arc(l^i_{\Sigma}))$ and set
 $$F(L):= \{ Y_0, Y_1, Y_2, \ldots, Y_{m'} \}.$$
As explained in \cite{Tu2} one can associate in a natural way a
 half integer $\gleam(Y)  \in \tfrac{1}{2} \bZ$, called ``gleam'' of $Y$,
to each face $Y \in F(L)$. In the special case  where the two conditions (NCP)
and (NH) appearing in Sec. 6.1 in \cite{Ha7a} are fulfilled\footnote{cf. Remark \ref{rm_appA_Ende} below
for the relevance of these two conditions for the present paper}   we have the explicit formula
\begin{equation} \label{eqA.1}
\gleam(Y) = \sum_{i \text{ with } \arc(l^i_{\Sigma}) \subset
\partial Y}  \wind(l^i_{S^1}) \cdot \sgn(Y;l^i_{\Sigma}) \in \bZ,
\end{equation}
where $ \wind(l^i_{S^1})$ is the winding number of the loop
 $l^i_{S^1}$ and where $ \sgn(Y;l^i_{\Sigma})$
 is given by
 \begin{equation} \label{eqA.2}
  \sgn(Y; l^{i}_{\Sigma}):=
\begin{cases} 1 & \text{ if  $Y \subset R^+_i$ }\\
-1 & \text{ if  $Y \subset R^-_i$ }\\
\end{cases}
\end{equation} Here  $R^{+}_i$ (resp. $R^{-}_i$) is the unique
connected component $R$ of $\Sigma \backslash \arc(l^i_{\Sigma})$
such that $l^i_{\Sigma}$ runs around $R$ in the ``positive'' (resp.
``negative'') direction. \par

Let $G$  be a simply-connected and simple compact Lie group
  with maximal torus $T$.
  In the following we will use the notation from part \ref{appB} of the Appendix.
  In particular, we have
\begin{equation} \label{eqA.3}
 \Lambda^k_+    = \{ \l \in \Lambda_+  \mid  \langle \l,\th \rangle \leq k - \cg\}
\end{equation}
where $k \in \bN $ is as in Sec. \ref{sec2} above and where
 $\cg= 1 + \langle \theta,\rho \rangle$ is the dual Coxeter number of
 $\cG$.

 \begin{remark} \label{rm_app0} \rm
 Observe that for $k < \cg$ the set $\Lambda^k_+$  is empty
 so  $|L|$ as defined in Eq. \eqref{eqA.4} below will then vanish.
 If $k = \cg$ then $|L|$ will in general not vanish but will still be rather trivial.
 Not surprisingly, the definition of $|L|$ in the literature
 often excludes the situation $k \le \cg$ so our definition of  $|L|$  in Eq. \eqref{eqA.4}
 below is more general than usually.
 \end{remark}

Assume that each loop $l_i$ in the link $L$  is equipped with a ``color'' $\rho_i$,
i.e. a finite-dimensional complex representation of $G$.
By   $\gamma_i \in \Lambda_+$ we denote the highest weight of
 $\rho_i$ and set     $\gamma(e):= \gamma_i$
 for each $e \in E(L)$  where $i \le m$ denotes the unique
    index such that  $\arc(e) \subset \arc(l_i)$.
Finally, let  $col(L)$ be the set of all mappings $\vf: \{Y_0,
Y_1, Y_2, \ldots, Y_{m'}\} \to \Lambda^k_+$ (``area
colorings'').\par

We can now define the ``shadow invariant'' $|L|$ of the
(colored and ``horizontally framed'')   link $L$
associated to the pair $(\cG,k)$  by
 \begin{equation} \label{eqA.4}
|L|:= \sum_{\vf\in col(L)}
|L|_1^{\vf}\,|L|_2^\vf\,|L|_3^\vf\,|L|_4^\vf~
\end{equation}
  with\footnote{see footnote 93 in \cite{Ha7a} for a brief comment on the exact relationship between
  our formula for $|L|$ and the corresponding formula for $|L|$ in \cite{turaev}.}
  \begin{subequations} \label{eqA.5}
\begin{align} |L|_1^\vf&=\prod_{Y \in F(L)} \dim(\vf(Y))^{\chi(Y)}\\
|L|_2^\vf&= \prod_{Y \in F(L)}   \exp(\tfrac{\pi i}{{k}} \langle \vf(Y),\vf(Y) +2\r\rangle)^{\gleam(Y)}\\
\label{eq_XL3}  |L|_3^\vf&= \prod_{e \in E_*(L)} N_{\gamma(e)
\varphi(Y^+_e)}^{\varphi(Y^-_e)} \\
|L|_4^\vf&= \bigl( \prod_{e \in E(L) \backslash E_*(L)} S(e,\vf) \bigr) \times \bigl( \prod_{x \in V(L)} T(x,\vf) \bigr)
  \end{align}
 \end{subequations}
 Here $Y^+_e$ (resp. $Y^-_e$) denotes the  unique face $Y$
  such that $\arc(e) \subset \partial Y$ and,
  additionally,  the orientation on $\arc(e)$ described above
coincides with (resp. is opposite to) the orientation which is
obtained by restricting the orientation on $\partial Y$ to $e$.
Moreover,    we have set
 (for $\lambda, \mu, \nu \in \Lambda^k_+$)
  \begin{equation}  \label{eqA.6}
\dim(\lambda) := \prod_{\a \in \cR_+}{\sin{\pi \langle \l+\r,\a
\rangle \over {k} }\over\sin{\pi \langle \r,\a \rangle \over
{k} }}
 \end{equation}
    \begin{equation} \label{eqA.7}
 N_{ \mu \nu}^{\lambda} := \sum_{\tau \in \cW_{{k}}}
  \sgn(\tau) m_{\mu}(\nu-\tau(\lambda))
 \end{equation}
where $m_{\mu}(\beta)$ is the multiplicity of the weight $\beta$
in the  unique (up to equivalence) irreducible
representation $\rho_{\mu}$ with  highest weight $\mu$
and   $\cW_{{k}}$ is as in part \ref{appB} of the Appendix.
$E_*(L)$ is a suitable subset of $E(L)$ (cf. the notion of ``circle-1-strata''
in  Chap. X, Sec. 1.2 in \cite{turaev}). \par

The  explicit expression for the factors  $T(x,\vf)$ appearing in $|L|_4^\vf$  above
involves the so-called  ``quantum 6j-symbols'' (cf. Chap. X, Sec. 1.2 in \cite{turaev}) associated to $U_q(\cG_{\bC})$
 where $q$ is the root of unity\footnote{We remark that there are  different conventions
  for the definition of $U_q(\cG_{\bC})$. Accordingly, one finds different formulas for $q$ in the literature. For example, using the convention in \cite{Saw03}
  one would be led to the formula
  $q:=  e^{\frac{\pi i}{ D k}}$  where   $D$ is the quotient of the square lengths
of the long and the short roots of $\cG$}
\begin{equation} \label{eq_rootunity} q:= \exp( \tfrac{2 \pi i}{k})
\end{equation}
We omit the explicit formulae for $T(x,\vf)$ and $S(e,\vf)$ since they irrelevant for the present paper.
Indeed, for links $L$ fulfilling
the  aforementioned conditions (NCP) and (NH)
of Sec. 6.1 in \cite{Ha7a}   the set $V(L)$ is empty and
the set $E_*(L)$ coincides with $E(L)$, so  Eq. \eqref{eqA.4} (combined with Eqs  \eqref{eqA.5}) then reduces to
 \begin{equation} \label{eqA.8}
|L| = \sum_{\vf\in col(L)} \biggl( \prod_{i=1}^m N_{\gamma(l_i)
\varphi(Y^+_{i})}^{\varphi(Y^-_{i})} \biggr)  \biggl(  \prod_{Y \in F(L)}
\dim(\vf(Y))^{\chi(Y)}
  \exp(\tfrac{\pi i}{{k}} \langle \vf(Y),\vf(Y) +2\r \rangle )^{\gleam(Y)} \biggr)
  \end{equation}
   where  we have set $Y^{\pm}_{i}:= Y^{\pm}_{l^i_{\Sigma}}$.

\begin{remark} \label{rm_appA_Ende} \rm
The shadow invariant can be defined in a straightforward for every (colored) simplicial ribbon link $L$
in $q\cK \times \bZ_N$ (or in $\cK \times \bZ_N$) in by setting
$$|L| := |L_f|$$
where $L_f$ is the (colored and horizontally framed) piecewise smooth  link  associated to $L$.
Observe that if $L$ fulfills the conditions  (NCP)' and (NH)' of Sec. \ref{subsec4.10} above then
$L_f$ will fulfill the conditions (NCP) and (NH) mentioned above, so in this case Eq. \eqref{eqA.8}
above will again hold.
\end{remark}

\section{Appendix: $BF_3$-theory in the torus gauge}
\label{appC'}

\subsection{$BF_3$-theory}

Let $M$ be a closed oriented 3-manifold, let
$\tilde{G}$ be a simple simply-connected compact Lie group
with Lie algebra  $\tilde{\cG}$, and let
$\tilde{\G} := C^{\infty}(M,\tilde{G})$.\par

For $\tilde{A}\in \tilde{\cA} :=\Omega^1(M,\tilde{\cG})$
and   $\tilde{C} \in  \tilde{\cC} :=
\Omega^1(M,\tilde{\cG})$ and $\Lambda \in \bR$ (the ``cosmological
constant'')  we define\footnote{Here we assume
    for simplicity (cf. Remark \ref{rm_sec1.1} above) that $\tilde{G}$  is Lie subgroup of $U(\tilde{N})$ for some
 $\tilde{N} \in \bN$ and we set $\Tr := \tilde{c} \Tr_{\Mat(\tilde{N},\bC)}$
where $\tilde{c} \in \bR$ is chosen suitably}
\begin{equation} \label{eqB.1}
S_{BF}(\tilde{A},\tilde{C}):= \frac{1}{\pi} \int_{M}
\Tr\bigl(F^{\tilde{A}} \wedge \tilde{C} + \frac{\Lambda}{3}
\tilde{C} \wedge \tilde{C} \wedge \tilde{C}\bigr)
\end{equation}
 where $F^{\tilde{A}} := d\tilde{A} + \tilde{A} \wedge \tilde{A}$.
Let us assume in the following that
 $\Lambda \in \bR_+$  and set
\begin{equation} \kappa := \sqrt{\Lambda}
\end{equation}
 Note that $S_{BF}: \tilde{\cA} \times
\tilde{\cC} \to \bC$ is $\tilde{\G}$-invariant under the
$\tilde{\G}$-operation on
  $\tilde{\cA} \times \tilde{\cC}$
 given by
 $(A,C) \cdot \tilde{\Omega} =
 ( \tilde{\Omega}^{-1}A \tilde{\Omega} + \tilde{\Omega}^{-1}d\tilde{\Omega},\tilde{\Omega}^{-1}C
 \tilde{\Omega})$.

\medskip

It is well-known that in the situation   $\kappa \neq 0$ the
relation
\begin{equation} \label{eqB.2}
S_{BF}(\tilde{A},\tilde{C})=  S_{CS}(\tilde{A} + \kappa \tilde{C}) -
S_{CS}(\tilde{A} - \kappa \tilde{C})
\end{equation}
holds with $S_{CS}=S_{CS}(M,G,k)$
where  $G:= \tilde{G}$ and $ k := \frac{1}{\kappa}$.
 Using the change of
variable $(\tilde{A},\tilde{C}) \to ({A}_1,{A}_2)$ given by
\begin{equation} \label{eqB.3} {A}_1 := \tilde{A}+ \kappa \tilde{C}, \quad {A}_2 := \tilde{A}-\kappa \tilde{C}
\end{equation}
or, equivalently,
\begin{equation} \label{eqB.5}
\tilde{A}:= \tfrac{1}{2} ({A}_1 + {A}_2), \quad
\tilde{C}:= \tfrac{1}{2\kappa} ({A}_1 - {A}_2)
 \end{equation}
we therefore obtain, informally,  for every $\tilde{\chi}:
\tilde{\cA} \times \tilde{\cC} \to \bC$
\begin{align} \label{eqB.4}
& \iint \tilde{\chi}(\tilde{A},\tilde{C}) \exp(i S_{BF}(\tilde{A},\tilde{C})) D\tilde{A} D\tilde{C} \nonumber \\
& \sim \iint {\chi}((A_1,A_2))  \exp(i S_{CS}({A}_1)) \exp(- i S_{CS}({A}_2)) D{A}_1 D{A}_2
\end{align}
where ${\chi}: \cA_1 \times \cA_2 \to \bC$  with $\cA_j:= \Omega^1(M,\cG)$, $j=1,2$,
is the function given
by $\chi((A_1,A_2)) = \tilde{\chi}(\tilde{A},\tilde{C})$.

\medskip

If -- instead of setting $S_{CS}:=S_{CS}(M,G,k)$
with  $G:= \tilde{G}$ and $ k := \frac{1}{\kappa}$ --
we use $S_{CS}:=S_{CS}(M,G,(k_1,k_2))$ with
$G= \tilde{G} \times \tilde{G}$ and $(k_1,k_2) = (1/\kappa,
-1/\kappa)$  (cf. Remark 2.2 in \cite{Ha7a})
then we can\footnote{using  $\cA = \cA_1 \oplus \cA_2$ and $DA=DA_1 DA_2$}   rewrite
Eq. \eqref{eqB.4} as
\begin{align} \label{eqB.6}
& \iint_{\tilde{\cA} \times \tilde{\cC}}
  \tilde{\chi}(\tilde{A},\tilde{C}) \exp(i
S_{BF}(\tilde{A},\tilde{C})) D\tilde{A} D\tilde{C}
 \sim \int_{\cA} {\chi}({A})  \exp\bigl(i S_{CS}({A})\bigr) D{A}
\end{align}
 Thus we see that $BF_3$-theory on $M$ with group $\tilde{G}$ and $\kappa \neq 0$ is
essentially equivalent to CS theory on $M$ with group $G= \tilde{G} \times \tilde{G}$ and $(k_1,k_2) = (1/\kappa,
-1/\kappa)$.

\subsection{$BF_3$-theory on $M=\Sigma \times S^1$
 ``in the torus gauge''}

Let us now consider the special case
where $M=\Sigma \times S^1$,   $\kappa \neq 0$ and $1/\kappa
\in \bN$, and  where
 $\tilde{\chi}: \tilde{\cA} \times \tilde{\cC} \to \bC$ is of the form
\begin{equation}\tilde{\chi}(\tilde{A},\tilde{C})  =
\prod_{i=1}^m \Tr_{\rho_i}(\Hol_{l_i}(
 \tilde{A} + \kappa \tilde{C}, \tilde{A} - \kappa \tilde{C}))
 \end{equation}
(with    $(l_1, l_2, \ldots, l_m)$
and $(\rho_1,\rho_2, \ldots, \rho_m)$ as in Sec. \ref{subsec2.1}).

\smallskip

Let us now apply
``torus gauge fixing''to the expression
\begin{equation} \label{eqBF_expr} \iint_{\tilde{\cA} \times \tilde{\cC}}
  \tilde{\chi}(\tilde{A},\tilde{C}) \exp(i
S_{BF}(\tilde{A},\tilde{C})) D\tilde{A} D\tilde{C}
\end{equation}
More precisely, we will perform  the following three steps:
\begin{itemize}
\item we make a change of variable from
 ``BF-variables''  to
``CS-variables''\footnote{i.e. from $(\tilde{A},\tilde{C})$ to $(A_1,A_2)$} (Step 1)
\item we apply torus gauge fixing (Step 2)
\item we change back to ``$BF$-variables'' (Step 3)
\end{itemize}
Concretely, these three steps are given as follows:
\begin{description}
\item[Step 1:] We replace the expression \eqref{eqBF_expr} by the RHS  of Eq. \eqref{eqB.6}

\item[Step 2:] We perform torus gauge-fixing on the RHS
of Eq. \eqref{eqB.6}, i.e. we replace   the RHS
of Eq. \eqref{eqB.6} by the RHS  of\footnote{cf.  Remark 2.8 in \cite{Ha7a}}  Eq. \eqref{eq2.48}  in Sec. \ref{subsec2.2} above
(in the situation $G= \tilde{G} \times \tilde{G}$, $T:= \tilde{T} \times \tilde{T}$, and $(k_1,k_2) = (1/\kappa,
-1/\kappa)$  where $\tilde{T}$  is  a fixed maximal torus of $\tilde{G}$)

\item[Step 3:] We apply the change of variable
$(A^{\orth} , B) \to (\tilde{A}^{\orth},\tilde{B})$
given by
\begin{subequations} \label{eqB.8'}
\begin{align}
\tilde{A}^{\orth}  & :=  \bigl(\tfrac{ A^{\orth}_1 +
A^{\orth}_2}{2},  \tfrac{A^{\orth}_1 - A^{\orth}_2}{2\kappa}\bigr),\\
 \tilde{B} & := \bigl(\tfrac{B_1 + B_2}{2}, \tfrac{B_1 - B_2}{2\kappa} \bigr)
\end{align}
\end{subequations}
 to the RHS  of Eq. \eqref{eq2.48} (in the situation $G= \tilde{G} \times \tilde{G}$, $T:= \tilde{T} \times \tilde{T}$, and $(k_1,k_2) = (1/\kappa, -1/\kappa)$)
\end{description}
The expression which we obtain after performing the three steps
 above is the analogue of the RHS  of Eq. \eqref{eq_WLO_BF}  above
where instead of the change of variable \eqref{eqB.8} the change of variable
  \eqref{eqB.8'} is used.

\section{Appendix: Some alternatives to  the discretization approach in Sec. \ref{sec7}}
\label{appJ}

\subsection{A reformulation of Sec. \ref{subsec7.3} using the spaces  $\cA^{\orth}_{2N}(K)$, $\cA^{\orth}_{altern,1}(K)$, and $\cA^{\orth}_{altern,2}(K)$}
\label{appJ.1}

We will now reformulate/modify the discretization approach
in  Sec. \ref{subsec7.3} in a suitable way. This will not only lead to
certain stylistic improvements but also to several  insights which should be useful in \cite{Ha7c}.
Let us consider  the space
\begin{equation}
\cA^{\orth}_{2N}(K) := \Map(\bZ_{2N}, \cA_{\Sigma,\cG}(K)) = \Map(\bZ_{2N}, C^1(K_1,\cG) \oplus C^2(K_1,\cG))
\end{equation}
 Clearly, we  have
\begin{equation} \label{decomp_good1}
\cA^{\orth}_{2N}(K) \cong \cA^{\orth}_{altern,1}(K) \oplus \cA^{\orth}_{altern,2}(K)
\end{equation}
where
\begin{subequations}
\begin{align}
\cA^{\orth}_{altern,1}(K) & := \Map(\bZ^{even}_{2N},C^1(K_1,\cG)) \oplus \Map(\bZ^{odd}_{2N},C^1(K_2,\cG)) \\
\cA^{\orth}_{altern,2}(K) & := \Map(\bZ^{even}_{2N},C^1(K_2,\cG)) \oplus \Map(\bZ^{odd}_{2N},C^1(K_1,\cG))
\end{align}
\end{subequations}
with
\begin{subequations}
 \begin{align}
 \bZ_{2N}^{even} & := \{ \pi(t)   \mid t \in  \{2, 4, 6, \ldots, 2N\}  \} \subset  \bZ_{2N}\\
 \bZ_{2N}^{odd} & := \{  \pi(t)  \mid t \in  \{1, 3, 5, \ldots, 2N-1\} \} \subset  \bZ_{2N}
 \end{align}
 \end{subequations}
    $\pi:\bZ \to \bZ_{2N}$ being the canonical projection. \par

Let us now make the connection with the  constructions of Sec. \ref{subsec7.3}.
It is convenient to use the notation  $\cA^{\orth}_{double}(K)$ for what was
denoted by  $\cA^{\orth}(K)$ in Sec. \ref{subsec7.3}. Observe that
\begin{equation}\cA^{\orth}_{double}(K) = \Map(\bZ_N,\cA_{\Sigma,\cG \oplus \cG}(K)) \cong \cA^{\orth}_1(K) \oplus \cA^{\orth}_2(K)
\end{equation}
where we have set for $j=1,2$
\begin{equation} \label{eq_appD_Sec3_dec} \cA^{\orth}_j(K) := \Map(\bZ_N,\cA_{\Sigma,\cG}(K)) \cong
\Map(\bZ_{N},C^1(K_1,\cG)) \oplus \Map(\bZ_{N},C^1(K_2,\cG))
\end{equation}

We now make the identifications
\begin{equation} \label{eq_appJ_identi} \bZ_{2N}^{even}  \cong \bZ_{N} \cong \bZ_{2N}^{odd}
\end{equation}
which are induced by the bijections
$j_{even}:\bZ_{N} \to \bZ_{2N}^{even} $ and $j_{odd}:\bZ_{N} \to \bZ_{2N}^{odd} $ which are given by
$$j_{even}(\pi(t)) = \pi(2t), \quad \quad j_{odd}(\pi(t)) = \pi(2t-1) \quad \forall t \in \{1,2,\ldots,N\}$$
Clearly, the identifications \eqref{eq_appJ_identi} give rise to   identifications
\begin{equation} \label{eq_appJ_identif} \cA^{\orth}_{altern,j}(K) \cong \cA^{\orth}_j(K),  \quad j=1,2,  \quad
\text{ and } \quad
\cA^{\orth}_{double}(K) \cong \cA^{\orth}_{2N}(K)
\end{equation}

Recall that in Sec. \ref{subsec7.3} we worked with ``full ribbons'' $R$, i.e. closed simplicial ribbons
in $\cK \times \bZ = K_1 \times \bZ$ and recall also that such a $R$ induces three loops $l^+$, $l^-$, and $l$ in a natural way, $l^{\pm}$ being simplicial loops in $K_1 \times \bZ$ and $l$ being a simplicial loop in $K_2 \times \bZ$.
 In view of the identification $\bZ_{N} \cong \bZ_{2N}^{even}$
we will now consider $R$ as a closed simplicial ribbon
in $K_1 \times \bZ_{2N}^{even}$ and the loops $l^{\pm}$
 as simplicial  loops in $K_1 \times \bZ^{even}_{2N}$
 and $l$ as a simplicial  loop in $K_2 \times \bZ^{even}_{2N}$
 (here  we have equipped $\bZ^{even}_{2N}$ with the polyhedral cell complex structure
   inherited from $\bZ_N$).

\medskip

Using the identifications above we can now rewrite the discretization approach
in  Sec. \ref{subsec7.3}. Doing so will not only lead to  several stylistic improvements and ``insights''
(cf. Observation 1 and, possibly, Observation 3)
but it also suggests certain modifications which we will incorporate in \cite{Ha7c}
(cf. Observation 2):

\medskip

\noindent {\bf Observation 1:} Recall that on the RHS of  Eq. \eqref{eq_Hol_disc_BF}
    the field component $A^{\orth}_1$  ``interacts'' only with the loops $l^+$ and $l^-$
   while  $A^{\orth}_2$  ``interacts'' with the loop $l$.
   In view of the reformulation we just made this ansatz is very natural:
   If $A^{\orth}_1$ and  $A^{\orth}_2$
  are given as the components w.r.t. the decomposition \eqref{decomp_good1}
  then $A^{\orth}_1$ and $A^{\orth}_2$  will  ``live'' on different edges.
  More precisely, for $t \in \bZ_{2N}^{even}$ we have $A^{\orth}_1(t) \in C^1(K_1,\cG)$
  and $A^{\orth}_2(t) \in C^1(K_2,\cG)$.
  Since the $l^+$ and $l^-$ are simplicial loops in  $K_1 \times \bZ^{even}_{2N}$
  and $l$ is a simplicial loop in  $K_2 \times \bZ^{even}_{2N}$
$A^{\orth}_1$ will indeed only interact with $l^+$ and $l^-$
and $A^{\orth}_2$ will interact only with $l$.

\medskip

\noindent {\bf Observation 2:}
 Recall that the operators $L^{(N)}(B_{\pm})$ appearing on the RHS of Eq. \eqref{eq_RNB-def}
 are given by
 \begin{equation}  \label{eq_LB_Def.Wiederh}
 L^{(N)}(B_{\pm}) := \left( \begin{matrix}
 \hat{L}^{(N)}(B_{\pm}) && 0 \\
0 && \Check{L}^{(N)}(B_{\pm})
\end{matrix} \right)
\end{equation}
with $\hat{L}^{(N)}(B_{\pm})$ and $\Check{L}^{(N)}(B_{\pm}) $
 as in Eq. \eqref{eq_LN_ident1} and Eq. \eqref{eq_LN_ident2}
  in Sec. \ref{subsec4.1} above (with $B$ replaced by $B_{\pm}$
and where the matrix notation refers to the decomposition appearing in Eq. \eqref{eq_appD_Sec3_dec} above.\par
It is natural to ask whether it
 is possible to rewrite or redefine
(if not $L^{(N)}(B_{\pm})$ itself then at least) the product  $\star_K L^{(N)}(B_{\pm})$
  by a formula which involves  the (anti-symmetrized) operators $\bar{L}^{(2N)}(b)$ ,
    $b \in \ct$,      as in\footnote{or rather, the equation which is analogous to Eq. \eqref{eq_def_LOp_c}
     but with $N$ replaced by $2N$}   Eq. \eqref{eq_def_LOp_c}  instead of
  the operators  $\hat{L}^{(N)}(b)$ and $\Check{L}^{(N)}(b)$ appearing in
  Eq. \eqref{eq_def_LOp_a} and Eq. \eqref{eq_def_LOp_b}.
  And indeed, using the aforementioned reformulation of the
  discretization approach of Sec. \ref{subsec7.3} this is  possible,
 as we will now explain.\par

   Let $\star^{(2N)}_K: \cA^{\orth}_{2N}(K) \to \cA^{\orth}_{2N}(K)$
be the operator defined totally analogously as the  operator $\star_K: \cA^{\orth}(K) \to \cA^{\orth}(K)$
in Sec. \ref{subsec4.2} above but with $2N$ playing the role of $N$.
Moreover, let $L^{(2N)}(B_{\pm})$ be the operator on
\begin{equation}
 \cA^{\orth}_{2N}(K) \cong  \bigl( \oplus_{\bar{e} \in  \face_0(K_1 | K_2)}
\Map(\bZ_{2N},\cG)\bigr) \oplus  \bigl(\oplus_{\bar{e} \in  \face_0(K_1 | K_2)}
\Map(\bZ_{2N},\cG) \bigr)
\end{equation}
which is  given by
\begin{equation}
L^{(2N)}(B_{\pm}) := \bigl(\oplus_{\bar{e} \in  \face_0(K_1 | K_2)}
 \bar{L}^{(2N)}(B_{\pm}(\bar{e}))\bigr) \oplus \bigl(\oplus_{\bar{e} \in  \face_0(K_1 | K_2)}
 \bar{L}^{(2N)}(B_{\pm}(\bar{e}))\bigr)
\end{equation}
where  each $ \bar{L}^{(2N)}(b)$ with $b=B_{\pm}(\bar{e})$ is defined  totally analogously
as in Eq. \eqref{eq_def_LOp_c}  above (with $N$ replaced by $2N$).
 Observe that, even though neither of the two operators   $\star^{(2N)}_K$ nor $L^{(2N)}(B_{\pm})$
 leaves the two subspaces $\cA^{\orth}_{altern,j}(K)$, $j=1,2$, of $\cA^{\orth}_{2N}(K)$
 invariant  the composition  $\star^{(2N)}_K L^{(2N)}(B_{\pm}): \cA^{\orth}_{2N}(K) \to \cA^{\orth}_{2N}(K)$
  does. It turns out that  under the identifications \eqref{eq_appJ_identif}  the operator
 $\star^{(2N)}_K L^{(2N)}(B_{\pm})$ is similar but does not quite coincide with the operator
 $ \star_K L^{(N)}(B_{\pm})$ where $L^{(N)}(B_{\pm})$ is as in Eq. \eqref{eq_LB_Def.Wiederh} above.
 In \cite{Ha7c} we will  work with the ``new'' operators $\star^{(2N)}_K L^{(2N)}(B_{\pm})$.

\medskip

\noindent {\bf Observation 3:}
\noindent
Using the reformulation of the discretization approach
in  Sec. \ref{subsec7.3} sketched above
it might be possible to obtain  a  better understanding of the origin or the ``meaning''
 of the $1/2$-exponent appearing in Eq. \eqref{eq_def_DetFPdisc_BF}  in Sec. \ref{subsec7.3} above.\par

       We begin by  having a closer look at the spaces which are relevant in
       continuum  $BF_3$-theory\footnote{in contrast to  part \ref{appC'}  of the Appendix we write $G$ and $T$ etc instead of $\tilde{G}$, $\tilde{T}$; moreover, we often use the subscript ``double''} (cf.  part \ref{appC'} of the Appendix)  and some candidates for a simplicial realization:

  \begin{itemize}
 \item {\it Continuum spaces:} The full space is $\cA_{double}= \Omega^1(\Sigma \times S^1, \cG \oplus \cG)$.
      We have $\cA_{double} = \cA^{\orth}_{double} \oplus \cA^{||}_{double}$ where
     $\cA^{\orth}_{double} = \{ A \in \cA_{double}  \mid A(\partial/\partial t) = 0 \}
      \cong C^{\infty}(S^1,\Omega^1(\Sigma,  \cG \oplus \cG))$ and $\cA^{||}_{double} = \{ A_0 dt \mid A_0 \in  C^{\infty}(\Sigma \times S^1,\cG \oplus \cG)\} \cong C^{\infty}(\Sigma,\Omega^1(S^1,\cG \oplus \cG))$.
     By applying torus gauge fixing we can reduce $\cA^{||}_{double}$ essentially\footnote{here we ignore
     the issue of the well-known  topological obstructions, cf. Sec. 2.2.4 in \cite{Ha7a}}
      to $\cB_{double}= C^{\infty}(\Sigma,\ct \oplus \ct)$ (here we have identified $\cB_{double}$ with the subspace $\{ B dt | B \in  \cB_{double}\}$ of $\cA^{||}_{double}$).  The application of torus gauge fixing
      ``produces'' the heuristic determinant
        $ \prod_{\pm} \bigl[  \Det_{FP}(B_{\pm}) \bigr] =  \prod_{\pm} \bigl[   \det\bigl(1_{\ck}-\exp(\ad(B_{\pm}))_{|\ck}\bigr)\bigr]
        =  \prod_{\pm} \bigl[   \det\bigl(1_{\ck}-\Ad(\exp(B_{\pm}))_{|\ck}\bigr)\bigr]$
      where $1_{\ck}-\Ad(\exp(B_{\pm}))_{|\ck}: C^{\infty}(\Sigma,\ck) \to C^{\infty}(\Sigma,\ck) $
        (cf.  Sec. 2.2.3 in \cite{Ha7a}).

\item {\it Simplicial~spaces,~1.~choice:} A  natural choice for the simplicial analogues of the
continuum spaces mentioned above are the spaces  $\cA_{double}(q\cK) := C^1(q\cK \times \bZ_N, \cG \oplus \cG)$, $\cA^{\orth}_{double}(q\cK):=  \Map(\bZ_N,C^1(q\cK,\cG \oplus \cG))$,
$\cA^{||}_{double}(q\cK): = C^0(q\cK, C^1(\bZ_N,\cG \oplus \cG))$, and $\cB_{double}(q\cK) := C^0(q\cK,\ct \oplus \ct)$.
And in fact, in Sec. \ref{subsec7.3} above we used this definition of $\cB_{double}(q\cK)$ and of $\cA^{\orth}_{double}(q\cK)$
(and for technical reasons we also introduced the subspace
$\cA^{\orth}_{double}(K) := \Map(\bZ_N,C^1(K,\cG \oplus \cG)) $ of $\cA^{\orth}_{double}(q\cK)$).\par
One can hope\footnote{\label{ft_80} as mentioned in Appendix D in \cite{Ha7a} discrete torus gauge fixing can easily be implemented rigorously when working in a setting where spaces of Lie group valued maps are used;
it is not yet clear if/how this is also possible  in the present setting
which works with spaces consisting of  Lie algebra valued maps} that by applying ``discrete torus gauge fixing''
the space  $\cA^{||}_{double}(q\cK)$ can be reduced to the space $\cB_{double}(q\cK)$
(considered as a subspace of $\cA^{||}_{double}(q\cK)$ in the obvious way)
and that this   produces  the determinant
$\prod_{\pm} \bigl[ \prod_{x \in \face_0(q\cK)}  \det\bigl(1_{\ck}-\exp(\ad(B_{\pm}(x)))_{|\ck}\bigr)\bigr]
\bigr]$.
However, even if this argument could be made rigorous it would be totally unclear where the $1/2$-exponent appearing in  Eq. \eqref{eq_def_DetFPdisc_BF}  in Sec. \ref{subsec7.3} above should come from.
Let us therefore consider an alternative choice for the simplicial spaces.

\item {\it Simplicial spaces,~2. choice:}
Above (cf. Observation 1 and Observation 2) we observed\footnote{in fact, for
technical reasons we concentrated on the subspace $\cA^{\orth}_{2\bZ}(K):= \Map(\bZ_{2N},C^1(K,\cG))$
of $\cA^{\orth}_{2\bZ}(q\cK)$ but this is not essential for the present discussion} that
it has several advantages to use  the space
$\cA^{\orth}_{2\bZ}(q\cK):= \Map(\bZ_{2N},C^1(q\cK,\cG))$
instead of the space $\cA^{\orth}_{double}(q\cK)=  \Map(\bZ_N,C^1(q\cK,\cG \oplus \cG))$  as the simplicial analogue of the continuum space $\cA^{\orth}_{double}$
and to work with
 a suitable ``intertwined/interlinked'' direct sum decomposition of $\cA^{\orth}_{2\bZ}(q\cK)$.
  Similarly, we can introduce
$\cA_{2\bZ}(q\cK) := C^1(q\cK \times \bZ_{2N}, \cG)$ and
$\cA^{||}_{2\bZ}(q\cK) := C^0(q\cK, C^1(\bZ_{2N}, \cG))$ as  simplicial analogues
of the continuum spaces
$\cA_{double}$ and $\cA^{||}_{double}$.
It is now natural to ask
\begin{itemize}
\item if one can  find a ``good''
 ``intertwined/interlinked'' direct sum decomposition
 of $\cA^{||}_{2\bZ}(q\cK)$ (possibly,  in the spirit of Appendix \ref{appJ.2} below, i.e. involving
the root space decomposition \eqref{eq_root_dec} in some way),
\item if one can embed $\cB_{double}(q\cK)$
 as a suitable subspace $\cB^{embedded}_{double}(q\cK)$ of $\cA^{||}_{2\bZ}(q\cK)$ which respects the
 ``intertwined/interlinked'' direct sum decomposition of $\cA^{||}_{2\bZ}(q\cK)$
 and, finally,
  \item if one can use a version of ``discrete torus gauge fixing''\footnote{cf. footnote \ref{ft_80} above}  in order to reduce  the space $\cA^{||}_{2\bZ}(q\cK)$ to $\cB^{embedded}_{double}(q\cK)$.
 \end{itemize}
 If all this is possible
 then there  might indeed be a chance of   obtaining a satisfactory justification for
 the $1/2$-exponent appearing in Eq.   \eqref{eq_def_DetFPdisc_BF}  in Sec. \ref{subsec7.3} above.
\end{itemize}

\subsection{Another version of Eq. \eqref{eq_WLO_BF} in Sec. \ref{subsec7.2}}
\label{appJ.2}

 We observed in Remark \ref{rm_7.2b} in Sec. \ref{sec7} above that
 because of the   $\star$-operators on the main diagonal of the $2 \times 2$-matrix appearing in Eq. \eqref{eq_tildeS_CS_explizit} above  we cannot hope
to be able to find a discretized version of the path integral on the RHS  of Eq. \eqref{eq_WLO_BF}
where each of the two components $\Check{\tilde{A}}^{\orth}_1$ and $\Check{\tilde{A}}^{\orth}_2$ ``lives'' either on
 $K_1 \times \bZ_N$ or on $K_2 \times \bZ_N$. \par

 We will now show that by using a more sophisticated change of variable $A^{\orth} \to \tilde{A}^{\orth}$
 instead of   \eqref{eqB.8a}  the $\star$-operator on the main diagonal
 can be eliminated after all, cf. Eq. \eqref{eq_tildeS_CS_explizit_mod2} below.
  This new  change of variable, which we will introduce below,  is based on the ``root space decomposition''
\begin{equation} \label{eq_root_dec}
\tilde{\cG} = \tilde{\ct} \oplus \bigl(\oplus_{\alpha \in \tilde{\cR}_+}  \tilde{\cG}_{\alpha} \bigr)
\end{equation}
of $\tilde{\cG}$ where $\tilde{\cR}_+$ is the set of positive real roots of $\tilde{\cG}$ associated
to  $\tilde{\ct}$ (and to a fixed Weyl chamber of $\tilde{\ct}$)
and  $\tilde{\cG}_{\alpha} \cong \bR^2$ is the ``root space''  corresponding to $\alpha \in \tilde{\cR}_+$.
 For simplicity let us consider only the special case
 $\tilde{G}=SU(2)$ and $\tilde{T} = \{\exp( \theta \tau_0) \mid  \theta \in \bR\}$
with $\tau_0$ given below.
 In this case we  can rewrite Eq. \eqref{eq_root_dec}  as
 \begin{equation} \label{eq_root_dec_simpl}
 \tilde{\cG}=su(2) = \bR \cdot \tau_0 \oplus \bigl(\bR \cdot \tau_{1} \oplus \bR \cdot \tau_{2}\bigr)
 \end{equation}
 where
  $$\tau_0: = \left( \begin{matrix} i  && 0 \\ 0 && -i  \end{matrix} \right), \quad
  \tau_{1} =\left( \begin{matrix} 0  && i \\ i && 0  \end{matrix} \right), \quad
  \tau_{2} =\left( \begin{matrix} 0  && 1 \\ -1 && 0  \end{matrix} \right)$$
Using the concrete basis $(\tau_0,\tau_1,\tau_2)$ we can identify
 $\tilde{\cG}=su(2)$ with $\bR^3$ in the obvious way, which in turn leads to the identification
 \begin{equation} \label{eq_identif.} C^{\infty}(S^1, \cA_{\Sigma,\tilde{\cG}}) \cong C^{\infty}(S^1, \cA_{\Sigma,\bR})^3
 \end{equation}
The new change of variable $ A^{\orth} \to \tilde{A}^{\orth}$
mentioned above is  defined by
\begin{equation} \label{eqB.8a_new}
\tilde{A}^{\orth}   :=  \bigl( \tfrac{Q}{2} ( A^{\orth}_1 + A^{\orth}_2), \tfrac{Q^{-1}}{2} ( A^{\orth}_1 - A^{\orth}_2)\bigr),\\
\end{equation}
where $Q$ is the operator  on
$ C^{\infty}(S^1, \cA_{\Sigma,\tilde{\cG}}) \cong C^{\infty}(S^1, \cA_{\Sigma,\bR})^3   $
given by
$$ Q := \left( \begin{matrix} 1 && 0 && 0 \\
0 && 1 && \star  \\ 0 && 1 && -\star   \end{matrix} \right)$$
Here $1$ denotes the identity operator on $C^{\infty}(S^1, \cA_{\Sigma,\bR})$
and $\star$ the Hodge star operator on $C^{\infty}(S^1, \cA_{\Sigma,\bR})$
 which is induced by the auxiliary Riemannian metric $\mathbf g$
 fixed in Sec. \ref{subsec2.2} above.

\smallskip

Using the changes of variable \eqref{eqB.8a_new} and \eqref{eqB.8c} we arrive at the following modification of Eq. \eqref{eq_tildeS_CS_explizit}
above
\begin{multline}  \label{eq_tildeS_CS_explizit_mod2}
  \bS(\Check{\tilde{A}}^{\orth},\tilde{B})   =   \pi k \ll (\Check{\tilde{A}}^{\orth}_1,\Check{\tilde{A}}^{\orth}_2), \biggl( \begin{matrix}
 J   \ad( \tilde{B}_2) && \star \tfrac{\partial}{\partial t} + J \ad(\tilde{B}_1) \\
\star \tfrac{\partial}{\partial t} + J \ad(\tilde{B}_1)
  && J   \ad(\tilde{B}_2)
\end{matrix} \biggr) \cdot (\Check{\tilde{A}}^{\orth}_1,\Check{\tilde{A}}^{\orth}_2) \gg_{\Check{\tilde{\cA}}^{\orth}}
\end{multline}
where $J$ is the linear operator on  $ C^{\infty}(S^1, \cA_{\Sigma,\tilde{\cG}})$
which is induced in the obvious way by the
linear operator  $J_0$ on $\tilde{\cG}=su(2)$  given by $J_0 \cdot \tau_0 = 0$,
$ J_0 \cdot \tau_{1} =  \tau_{2}$, and $ J_0 \cdot \tau_{2} =  \tau_{1}$.\par

Observe that in contrast to Eq. \eqref{eq_tildeS_CS_explizit} above,
there is no $\star$-operator appearing on the main diagonal
of the matrix operator in Eq. \eqref{eq_tildeS_CS_explizit_mod2}.
It turns out, however, that when trying to discretize the analogue of Eq. \eqref{eq_WLO_BF}  which is obtained after
applying the change of variable \eqref{eqB.8a_new}  instead of \eqref{eqB.8a}
it is still not possible to
 to find an implementation where each of the two components $\Check{\tilde{A}}^{\orth}_1$ and $\Check{\tilde{A}}^{\orth}_2$ ``lives'' either on
 $K_1 \times \bZ_N$ or on $K_2 \times \bZ_N$.
In other words, it is still necessary to use a ``mixed'' implementation (cf.  Remark \ref{rm_7.2b} above).
 This time we have more freedom in choosing the ``mixed'' implementation.
It remains to be seen  whether this has any advantages compared to the original approach in Sec. \ref{subsec7.3}.

\end{document}